\documentclass[a4paper,UKenglish,cleveref,autoref,thm-restate,authorcolumns]{lipics-v2021}

\usepackage{arxiv-lipics-v2021}

\usepackage{ifmtarg}
\usepackage{graphicx}
\usepackage{amsmath,amsthm}
\usepackage{mathtools}
\usepackage[capitalise]{cleveref}
\usepackage{enumerate}
\usepackage[table]{xcolor}
\usepackage{xspace}
\usepackage{environ}
\usepackage{algorithm}
\usepackage[noend]{algpseudocode}
\usepackage{array}
\usepackage{tabularx}
    \newcolumntype{M}[1]{>{\centering\arraybackslash}m{#1}}
\usepackage{tikz}
    \usetikzlibrary{arrows,automata,decorations.pathreplacing}

\newtheoremstyle{problemstyle}{\topsep}{\topsep}{}{0pt}{\sffamily\bfseries}{}{5pt plus 1pt minus 1pt}%
  {\textcolor{lipicsGray}{$\blacktriangleright$}\nobreakspace\thmname{#1}.\thmnote{ #3}}
\theoremstyle{problemstyle}
\newtheorem*{problem}{Problem}

\makeatletter
\newcommand{\zvass}[2]{\@ifmtarg{#1}{VASS+$#2\ZZ$}{$#1$-VASS+$#2\ZZ$}}
\makeatother

\makeatletter
\newcommand{\problemtitle}[1]{\gdef\@problemtitle{#1}}
\newcommand{\probleminput}[1]{\gdef\@probleminput{#1}}
\newcommand{\problemquestion}[1]{\gdef\@problemquestion{#1}}
\NewEnviron{statement}{
  \problemtitle{}\probleminput{}\problemquestion{}
  \BODY
  \par\addvspace{.5\baselineskip}
  \noindent
  \begin{tabularx}{\columnwidth}{@{\hspace{\parindent}}X}
  \textbf{\@problemtitle}
  \\[.3\baselineskip]
  \textbf{Input:} \@probleminput \\[.2\baselineskip]
  \textbf{Question:} \@problemquestion
\end{tabularx}
  \par\addvspace{.5\baselineskip}\noindent
}
\makeatother

\newcommand{\vect}[1]{\mathbf{#1}}
\newcommand{\rank}[0]{\mathbf{rank}}
\newcommand{\diff}[0]{\mathbf{diff}}
\newcommand{\up}[0]{\mathbf{up}}
\newcommand{\dwn}[0]{\mathbf{dwn}}
\newcommand{\qiin}[0]{q_{i,\text{in}}}
\newcommand{\qiout}[0]{q_{i,\text{out}}}
\newcommand{\effect}[0]{\mathbf{eff}}
\newcommand{\oneeffect}[1]{\textup{eff}#1}
\newcommand{\NumberOfSCCs}[0]{s}

\newcommand{\ILP}[0]{\mathit{ILP}}

\newcommand{\class}[1]{\textsc{#1}}
\newcommand{\nl}{\class{NL}\xspace}
\newcommand{\np}{\class{NP}\xspace}
\newcommand{\pspace}{\class{PSpace}\xspace}
\newcommand{\expspace}{\class{ExpSpace}\xspace}
\newcommand{\twoexpspace}{\class{2-ExpSpace}\xspace}
\newcommand{\elementary}{\class{Elementary}\xspace}
\newcommand{\tower}{\class{Tower}\xspace}
\newcommand{\ackermann}{\class{Ackermann}\xspace}
\newcommand{\spacecl}{\class{Space}\xspace}
\newcommand{\Ff}{\mathcal{F}} 
\newcommand{\Fff}{\textsc{Fun}}

\renewcommand{\vec}[1]{\mathbf{#1}}
\newcommand{\abs}[1]{\lvert#1\rvert}
\newcommand{\norm}[1]{\lVert#1\rVert}
\newcommand{\onenorm}[1]{\norm{#1}_1}
\newcommand{\infnorm}[1]{\norm{#1}_\infty}
\newcommand{\maxnorm}[1]{\infnorm{#1}}
\newcommand{\Oo}{\mathcal{O}}
\newcommand{\Oh}{\Oo}
\newcommand{\xproj}[2]{\pi_{#1}(#2)}
\newcommand{\nproj}[1]{\xproj{\NN}{#1}}
\newcommand{\zproj}[1]{\xproj{\ZZ}{#1}}

\newcommand{\set}[1]{\{#1\}}
\newcommand{\sset}{\subseteq}
\newcommand{\setof}[2]{\{#1 \mid #2\}}

\newcommand{\NN}{\mathbb{N}}
\newcommand{\ZZ}{\mathbb{Z}}
\newcommand{\QQ}{\mathbb{Q}}

\newcommand{\Aa}{\mathcal{A}}

\newcommand{\Vv}{\mathcal{V}}

\newcommand{\eff}[1]{\mathbf{eff}(#1)}

\newcommand{\config}[2]{#1(#2)}
\newcommand{\Config}[2]{#1(\vec{#2})}
\newcommand{\size}{\textup{size}}
\newcommand{\unarysize}[1]{\size_\text{u}(#1)}
\newcommand{\binarysize}[1]{\size_\text{b}(#1)}

\newcommand{\reaches}{\xrightarrow{}}

\newcommand{\vectSet}[1]{\mathbf{#1}}

\newcommand{\add}[2]{$\coreadd{#1}{#2}$}
\newcommand{\sub}[2]{$\coresub{#1}{#2}$}
\newcommand{\coreadd}[2]{#1 \,\, +\!\!= \, #2}
\newcommand{\coresub}[2]{#1 \,\, -\!\!= \, #2}
\newcommand{\ZT}{\textbf{zero-test}}
\newcommand{\Mult}{\textbf{mult}}
\newcommand{\WeakMult}{\textbf{weak-mult}}
\newcommand{\ExactMult}{\textbf{exact-mult}}
\newcommand{\Copy}{\textbf{copy}}
\newcommand{\Guess}{\textbf{guess}}
\newcommand{\Move}{\textbf{move}}
\newcommand{\Triple}{\textbf{16-triple}}
\newcommand{\ExpAmplifier}{\textbf{exp-amplifier}}

\newcommand{\testcell}[1]{
    \begin{tikzpicture}
        \draw[draw=black, fill=#1, line width = 0.1mm] (0,0) rectangle (0.5,-0.2);
    \end{tikzpicture}
}

\definecolor{tablegreen}{HTML}{B7D7A8}
\definecolor{tablelime}{HTML}{DDE8A8}
\definecolor{tablered}{HTML}{EFADAD}
\definecolor{tableorange}{HTML}{F4C7B5}

\nolinenumbers 
\hideLIPIcs 
\title{Reachability in VASS Extended with Integer Counters}

\author{Clotilde Bizière}
{University of Bordeaux, France
\and University of Warsaw, Poland}
{clotilde.biziere@u-bordeaux.fr}
{}
{Supported by Polish National Science Centre SONATA BIS-12 grant
number 2022/46/E/ST6/00230.}

\author{Wojciech Czerwi\'nski}
{University of Warsaw, Poland 
    \and \url{https://www.mimuw.edu.pl/~wczerwin}
}
{wczerwin@mimuw.edu.pl}
{https://orcid.org/0000-0002-6169-868X}
{Supported by the ERC grant INFSYS, agreement no. 950398.}

\author{Roland Guttenberg}
{University of Warsaw, Poland}
{r.guttenberg@uw.edu.pl}
{https://orcid.org/0000-0001-6140-6707}
{Supported by the ERC grant INFSYS, agreement no. 950398.}

\author{J\'er\^ome  Leroux}
{University of Bordeaux, France}
{jerome.leroux@labri.fr}
{}
{Supported by the grant ANR-25-CE48-6933 of the French National Research Agency ANR (project CoqoPetri).}

\author{Vincent Michielini}
{University of Bordeaux, France}
{vincent.michielini@u-bordeaux.fr}
{}
{Supported by the ANR grant INTENDED (ANR-19-CHIA-0014).}

\author{\L{}ukasz Orlikowski}
{University of Warsaw, Poland}
{l.orlikowski@mimuw.edu.pl}
{https://orcid.org/0009-0001-4727-2068}
{Supported by the ERC grant INFSYS, agreement no. 950398.}

\author{Antoni Puch}
{University of Warsaw, Poland}
{a.puch@uw.edu.pl}
{}
{Supported by the Polish National Science Centre Maestro grant\linebreak no. 2022/46/A/ST6/00072.}

\author{Henry Sinclair-Banks}
{University of Warsaw, Poland 
    \and Max Planck Institute for Software Systems \hbox{(MPI-SWS)}, Kaiserslautern, Germany 
    \and \url{http://henry.sinclair-banks.com}
}
{hsb@mpi-sws.org}
{https://orcid.org/0000-0003-1653-4069}
{Supported by the ERC grant INFSYS, agreement no. 950398.}



\authorrunning{Bizière, Czerwi{\'n}ski, Guttenberg, Leroux, Michielini, Orlikowski, Puch, Sinclair-Banks}

\acknowledgements{We thank Kyveli Doveri and Georg Zetzsche for inspiring discussions about VASS extended with integer counters.}

\relatedversiondetails[cite=BiziereCGLMOPS26]{Published Version}{https://doi.org/10.4230/LIPIcs.LICS.2026.19} 

\begin{document}

\maketitle

\begin{abstract}
    We consider a variant of VASS extended with integer counters, denoted \zvass{}{}.
    These are automata equipped with $\NN$- and $\ZZ$-counters; the $\NN$-counters are required to remain nonnegative and the $\ZZ$-counters do not have this restriction.
    We study the complexity of the reachability problem for \zvass{}{} when the number of $\NN$-counters is fixed.
    We show that reachability is \class{NP}-complete in \linebreak\zvass{1}{} (i.e.\ when there is only one $\NN$-counter) regardless of unary or binary encoding.
    For $d \geq 2$, using a KLMST-based algorithm, we prove that reachability in \zvass{d}{} lies in the complexity class $\Ff_{d+2}$.
    Our upper bound improves on the naively obtained Ackermannian complexity by simulating the $\ZZ$-counters with $\NN$-counters.

    To complement our upper bounds, we show that extending VASS with integer counters significantly lowers the number of $\NN$-counters needed to exhibit hardness.
    We prove that reachability in unary \zvass{2}{} is \pspace-hard; without $\ZZ$-counters this lower bound is only known in dimension~5.
    We also prove that reachability in unary \zvass{3}{} is \tower-hard.
    Without $\ZZ$-counters, reachability in 3-VASS has elementary complexity and \tower-hardness is only known in dimension~8.
\end{abstract}

\section{Introduction}
\label{sec:introduction}

Vector Addition Systems with States (VASS) together with Petri nets, an equivalent model, are established computational models of concurrency with a variety of applications in computer science~\cite{Reisig13}.
VASS can be seen as automata equipped with nonnegative integer counters.
These counters can be incremented and decremented, but importantly, they cannot be tested for equality with zero. 
A central problem for VASS is reachability: from a given initial configuration (state and counter values), is there a run to the given target configuration?

Research into VASS reachability has long and rich history.
The problem was first considered in the seventies; in 1976, reachability was shown to be \expspace-hard by Lipton~\cite{Lipton76}. 
Decidability of reachability was only proved later by Mayr  in 1981~\cite{Mayr81}. 
Mayr's algorithm was subsequently simplified and reformulated by Kosaraju~\cite{Kosaraju82} and Lambert~\cite{Lambert92}.
After the names of the three inventors, and after Sacerdote and Tenney who provided partial results in 1977~\cite{SacerdoteT77}, the main algorithm used to decide reachability in VASS goes by the acronym ``KLMST''.
There has since been great effort devoted to establishing the  complexity of reachability and it was only recently determined to be \ackermann-complete~\cite{LerouxS19,Leroux21,CzerwinskiO21}.

Despite this breakthrough, many questions concerning reachability problems in VASS and related models are far from being fully understood.
A prominent yet puzzling example is the complexity of the reachability problem in $d$-dimensional VASS (aka.\ $d$-VASS) for fixed $d \in \NN$. 
The complexity has only been settled for dimensions $1$ and $2$~\cite{ValiantP75,HaaseKOW09,BlondinEFGHLMT21}. 
Already for dimension $3$, despite very recent progress, there is a wide complexity gap: with unary encoding, the problem is \np-hard~\cite{ChistikovCMOSW24} and in \twoexpspace~\cite{CzerwinskiJLO25}.
Even though an elementary algorithm is not known for reachability in $4$-VASS, the best-known non-elementary lower bound requires 8 dimensions~\cite{CzerwinskiO22}. 
Moreover, in higher dimensions $d \in \NN$ the gap is huge: reachability in $d$-VASS is known to be $\Ff_{(d-3)/2}$-hard~\cite{CzerwinskiJLLO23} and in $\Ff_d$~\cite{FuYZ24}; here $\Ff_d$ is the $d$-th complexity class defined by the hierarchy of fast growing functions (for details, see~\cite{Schmitz16}).

In this paper, we consider VASS extended with integer counters. 
This model is a kind of automaton equipped with counters that must remain nonnegative ($\NN$-counters) and counters without this restriction ($\ZZ$-counters). 
We use \zvass{d}{k} to refer to a VASS which has $d$ $\NN$-counters and $k$ $\ZZ$-counters.
We also use \zvass{d}{} to refer to the scenario in which the number of $\ZZ$-counters is not specified.

The reachability problem for VASS extended with integer counters in which both $d$ and $k$ are not specified can be easily shown to be \ackermann-complete. 
Indeed, every integer counter can be simulated by the difference of two nonnegative counters.
This means that the additional integer counters do not offer any expressive power. 
However, reachability in \zvass{d}{}, for fixed $d \in \NN$, turns out to be an intriguing, nontrivial problem.

Our main motivation to study reachability in \zvass{}{} is two-fold. 
First, \zvass{}{} are a natural, basic model of computation that we wish to examine in detail.
This model has been implicitly used several times before: Rackoff used \zvass{}{} when analysing self-covering runs that witness unboundedness in VASS~{\cite[Section 4]{Rackoff78}}; they have been used by Demri to prove that the strong promptness detection problem can be decided in exponential space~{\cite[Section 3]{Demri13}}. 
\zvass{}{} also appear in KLMST algorithms for reachability in VASS~\cite{Kosaraju82} (precisely, after up-pumping some counters has occurred and before said counters are later down-pumped). 
As a result, \zvass{}{} have also been indirectly used to show that the regular separability problem of VASS reachability languages is decidable~\cite{Keskin024}; that reachability in VASS equipped with a pushdown stack is decidable~\cite{GuttenbergKM26}; and to prove that mutual reachability in Petri nets with data is decidable~{\cite[Section 5]{KaminskiL24}}.
Even though \zvass{}{} have been discussed and implicitly used by the community, they had not yet been systematically investigated.

The second motivation is that understanding fixed dimension VASS appears to be highly nontrivial and requires novel approaches. 
Adding integer counters immediately yields surprising observations: even though the coverability problem for VASS is \expspace-complete~\cite{Lipton76,Rackoff78}, adding just one integer counter makes coverability as hard as reachability (i.e.\ \ackermann-complete).
This is because one additional integer counter can be introduced to preserve the invariant that the sum of all counters is zero.
We anticipate that studying \zvass{}{}, a very closely related model, may be fruitful for finding new techniques that will also be applicable for fixed dimension VASS.

Additionally, we have observed an application of reachability in \zvass{d}{}, specifically a reduction of the following problem to reachability in \zvass{d}{}.
Given two labelled $d$-VASS $\Vv_1$ and $\Vv_2$, decide whether there exists a word $w_1$ accepted by $\Vv_1$ and a word $w_2$ accepted by $\Vv_2$ with the same Parikh image.
To sketch the reduction, first consider running $\Vv_1$ while guessing $w_1$ and store its Parikh image in some $\ZZ$-counters, zero test the $\NN$-counters using some more $\ZZ$-counters, then run $\Vv_2$ while guessing $w_2$ and subtract its Parikh image from the same $\ZZ$-counters; if the $\ZZ$-counters end with zero value, then the Parikh images were equal.

\subsection*{Our Contributions}

In this paper we present four main contributions regarding  reachability in \zvass{d}{}.

\begin{description}
    \item[Dimension One.]
    In~\cref{sec:dimension-one}, we study \zvass{1}{}; we prove that reachability is \np-complete, regardless of whether unary or binary encoding is used. 
    The main result here is the \class{NP} upper bound (Theorem~\ref{thm:1-zvass-np}). 
    Our result generalises the \np upper bound for binary 1-VASS~\cite{HaaseKOW09} and most interestingly uses a novel approach.
\end{description}

\begin{description}
    \item[Dimension Two.]
    In~\cref{sec:2-VASS+Z}, we show that the reachability problem for unary \zvass{2}{} is \pspace-hard (\cref{thm:unary-2-zvass-psapce-hard}). 
    At a glance, this result may not seem particularly impressive because reachability in unary VASS with only integer counters is already \class{NP}-hard.
    However, it turned out to be rather challenging and in proving~\cref{thm:unary-2-zvass-psapce-hard}, we have found new techniques.
    The first is an intricate construction to produce doubly-exponential counter values (\cref{lem:twoexptriple}). 
    The second is a new flavour of a hard reachability problem in counter automata (\cref{lem:reachability-CA}).
    The third is a novel way of simulating counter automata with VASS provided with a large counter value; we believe this last technique will have applications in the future.
\end{description}

\begin{description}
    \item[Dimension Three.]
    In~\cref{sec:3-VASS+Z}, we prove that reachability in unary \zvass{3}{} is \tower-hard (\cref{thm:unary-3-zvass-tower-hard}). 
    The situation here is the opposite compared to dimension two: the techniques are somewhat standard, but the result is rather strong.  
    Regardless, it is striking to see that reachability in \zvass{d}{} is non-elementary for such a small $d$.
    In particular, this case proves that extending VASS with integer counters makes the problem strictly harder as reachability for $3$-VASS has elementary complexity~\cite{CzerwinskiJLO25}.
\end{description}

\begin{description}
    \item[Higher Dimensions.]
    In~\cref{sec:klm-upper-bound}, we study \zvass{d}{} for an arbitrary fixed dimension $d \in \NN$.
    We show that reachability in \hbox{\zvass{d}{}} is in $\Ff_{d+2}$ (\cref{thm:klm-d-vass}). 
    Notice that this is slightly higher than the best-known $\Ff_d$ upper bound for reachability in $d$-VASS~\cite{FuYZ24}. 
    Our approach is based on the KLMST algorithm and also uses some ideas introduced in~\cite{FuYZ24}. 
    In order to deal with $\ZZ$-counters without the complexity increasing much, we use a new vector space which manages the possible deviations of the $\ZZ$-counters. 
    Our construction is almost optimal in the following sense: any KLMST-based approach cannot attain an upper bound below $\Ff_{d+1}$ (we discuss this in more detail in~\cref{sec:conclusion}).
\end{description}

\subsection*{Table of known results}

Before we give the preliminaries, we present the state of the art results for reachability in VASS, integer VASS, and VASS extended with integer counters together in two tables.
The first and second table outline what is known about reachability in \zvass{d}{k} encoded in unary and binary, respectively.
Do note that for $d = 3$ and fixed $d \geq 4$, the two tables completely coincide.
This is because we do not know of any stronger lower or upper bounds for $3$-VASS, \zvass{3}{k}, or \zvass{3}{} when comparing unary and binary encoding, and for fixed $d \geq 4$, $F_{d}$
Some cells of the tables are coloured according to the following key.
\begin{description}
    \item[Green\testcell{tablegreen}:]  New upper bound that is a main contribution.
    \item[Red\testcell{tablered}:]  New lower bound that is a main contribution.
    \item[Lime\testcell{tablelime}:]  New upper bound following from a main contribution.
    \item[Orange\testcell{tableorange}:] New lower bound following from a main contribution.
\end{description}

\paragraph*{Unary encoding}
{
    \centering
    \begin{tikzpicture}
        \def\ll{1}
        \def\lm{4}
        \def\rm{8}
        \def\rr{12}
        \def\Lline{0}
        \def\lline{2}
        \def\mline{6}
        \def\rline{10}
        \def\Rline{14}
        
        \def\h{0}
        \def\rowzero{-1}
        \def\rowone{-2}
        \def\rowtwo{-3}
        \def\rowthree{-4}
        \def\rowd{-5}
        
        \tikzset{
            tableline/.style={
                draw = black,
                line width = 0.2mm
            }
        }
        
        \tikzset{
            textentry/.style={
                scale = 0.9
            }
        }
        
        
        \fill[tablegreen] (\rline, \rowone+0.5) rectangle (\Rline, \rowone-0.5);
        
        \fill[tablelime] (\mline, \rowtwo-0.5) -- (\rline, \rowtwo+0.5) -- (\rline, \rowtwo-0.5) -- cycle;
        
        \fill[tablered] (\rline, \rowtwo+0.5) -- (\Rline, \rowtwo+0.5) -- (\rline, \rowtwo-0.5) -- cycle;
        \fill[tablelime] (\Rline, \rowtwo-0.5) -- (\Rline, \rowtwo+0.5) -- (\rline, \rowtwo-0.5) -- cycle;
        
        \fill[tableorange] (\mline, \rowthree-0.5) -- (\rline, \rowthree+0.5) -- (\mline, \rowtwo-0.5) -- cycle;
        \fill[tablelime] (\mline, \rowthree-0.5) -- (\rline, \rowthree+0.5) -- (\rline, \rowthree-0.5) -- cycle;
        
        \fill[tablered] (\rline, \rowthree+0.5) -- (\Rline, \rowthree+0.5) -- (\rline, \rowthree-0.5) -- cycle;
        \fill[tablelime] (\Rline, \rowthree-0.5) -- (\Rline, \rowthree+0.5) -- (\rline, \rowthree-0.5) -- cycle;
        
        \fill[tablelime] (\mline, \rowd+0.5) rectangle (\rline, \rowd-0.5);
        
        \fill[tablegreen] (\rline, \rowd+0.5) rectangle (\Rline, \rowd-0.5);
    
        \node[textentry] at (\ll, \h+0.2) {No. $\NN$-};
        \node[textentry] at (\ll, \h-0.2) {counters};
        \node[textentry] at (\rm, \h+0.25) {No. $\ZZ$-counters};
        \node[textentry] at (\lm,\h-0.25) {$k=0$};
        \node[textentry] at (\rm,\h-0.25) {fixed $k \geq 1$};
        \node[textentry] at (\rr,\h-0.25) {variable $k \in \NN$};
        \draw[tableline] (\ll-1, \h+0.5) -- (\rr+2, \h+0.5);
        \draw[tableline] (\lline, \h) -- (\rr+2, \h);
        \draw[tableline] (\ll-1, \h-0.5) -- (\rr+2, \h-0.5);
        
        \draw[tableline] (\Lline, \h+0.5) -- (\Lline, \rowd-0.5);
        \draw[tableline] (\lline, \h+0.5) -- (\lline, \rowd-0.5);
        \draw[tableline] (\mline, \h+0) -- (\mline, \rowd-0.5);
        \draw[tableline] (\rline, \h+0) -- (\rline, \rowd-0.5);
        \draw[tableline] (\Rline, \h+0.5) -- (\Rline, \rowd-0.5);
        
        \node[textentry] at (\ll, \rowzero) {$d=0$};
        \node[textentry] at (\lm, \rowzero+0.2) {\class{NL}-complete};
        \node[textentry] at (\lm, \rowzero-0.2) {\emph{Directed graph reachability}};
        \node[textentry] at (\rm, \rowzero+0.2) {\class{NL}-complete};
        \node[textentry] at (\rm, \rowzero-0.2) {\emph{Fixed no.\ variables ILP}};
        \node[textentry] at (\rr, \rowzero+0.2) {\class{NP}-complete (folklore)};
        \node[textentry] at (\rr, \rowzero-0.2) {See also Prop.~\ref{pro:unary-zvass}};
        \draw[tableline] (\ll-1, \rowzero-0.5) -- (\rr+2, \rowzero-0.5);
        
        \node[textentry] at (\ll, \rowone) {$d=1$};
        \node[textentry] at (\lm, \rowone) {\class{NL}-complete~\cite{ValiantP75}};
        \node[textentry] at (\rm, \rowone+0.2) {\class{NL}-complete~\cite{AtigCHKSZ16}};
        \node[textentry] at (\rm, \rowone-0.2) {See also Cor.~\ref{cor:unary-NL}};
        \node[textentry] at (\rr, \rowone) {\class{NP}-complete (Thm.~\ref{thm:1-zvass-np})};
        \draw[tableline] (\ll-1, \rowone-0.5) -- (\rr+2, \rowone-0.5);
        
        \node[textentry] at (\ll, \rowtwo) {$d=2$};
        \node[textentry] at (\lm, \rowtwo) {\class{NL}-complete~\cite{EnglertLT16}};
        \node[textentry] at (\rm, \rowtwo+0.2) {\class{NP}-hard (Cor.~\ref{cor:unary-2-vass+1z})};
        \node[textentry] at (\rm, \rowtwo-0.2) {In $\Ff^{(k)}_3$ (See Thm.~\ref{thm:klm-d-vass})};
        \node[textentry] at (\rr, \rowtwo+0.2) {\class{PSpace}-hard (Thm.~\ref{thm:unary-2-zvass-psapce-hard})};
        \node[textentry] at (\rr, \rowtwo-0.2) {In $\Ff_4$ (Thm.~\ref{thm:klm-d-vass})};
        \draw[tableline] (\ll-1, \rowtwo-0.5) -- (\rr+2, \rowtwo-0.5);
        
        \node[textentry] at (\ll, \rowthree) {$d=3$};
        \node[textentry] at (\lm, \rowthree+0.2) {\class{NP}-hard~\cite{ChistikovCMOSW24}};
        \node[textentry] at (\lm, \rowthree-0.2) {In \class{2ExpSpace}~\cite{CzerwinskiJLO25}};
        \node[textentry, scale=0.875] at (\rm, \rowthree+0.2) {$\Oh(k)$-\class{ExpSpace}-hard (Cor.~\ref{cor:3-vass+kz})};
        \node[textentry] at (\rm, \rowthree-0.2) {In $\Ff^{(k)}_4$ (See Thm.~\ref{thm:klm-d-vass})};
        \node[textentry] at (\rr, \rowthree+0.2) {\class{Tower}-hard (Thm.~\ref{thm:unary-3-zvass-tower-hard})};
        \node[textentry] at (\rr, \rowthree-0.2) {In $\Ff_5$ (Thm.~\ref{thm:klm-d-vass})};
        \draw[tableline] (\ll-1, \rowthree-0.5) -- (\rr+2, \rowthree-0.5);
        
        \node[textentry] at (\ll, \rowd+0.2) {fixed};
        \node[textentry] at (\ll, \rowd-0.2) {$d\geq 4$};
        \node[textentry] at (\lm, \rowd+0.2) {$\Ff_{(d-3)/2}$-hard~\cite{CzerwinskiJLLO23}};
        \node[textentry] at (\lm, \rowd-0.2) {In $\Ff_d$~\cite{FuYZ24}};
        \node[textentry] at (\rm, \rowd+0.2) {Lower bound from $k=0$};
        \node[textentry] at (\rm, \rowd-0.2) {In $\Ff_{d+1}^{(k)}$ (See Thm.~\ref{thm:klm-d-vass})};
        \node[textentry] at (\rr, \rowd+0.2) {Lower bound from $k=0$};
        \node[textentry] at (\rr, \rowd-0.2) {In $\Ff_{d+2}$ (Thm.~\ref{thm:klm-d-vass})};
        \draw[tableline] (\ll-1, \rowd-0.5) -- (\rr+2, \rowd-0.5);
    \end{tikzpicture}}
\paragraph*{Binary encoding}
{
    \centering
    \begin{tikzpicture}
        \def\ll{1}
        \def\lm{4}
        \def\rm{8}
        \def\rr{12}
        \def\Lline{0}
        \def\lline{2}
        \def\mline{6}
        \def\rline{10}
        \def\Rline{14}
        
        \def\h{0}
        \def\rowzero{-1}
        \def\rowone{-2}
        \def\rowtwo{-3}
        \def\rowthree{-4}
        \def\rowd{-5}
        
        \tikzset{
            tableline/.style={
                draw = black,
                line width = 0.2mm
            }
        }
        
        \tikzset{
            textentry/.style={
                scale = 0.9
            }
        }
        
        
        \fill[tablelime] (\mline, \rowone+0.5) rectangle (\rline, \rowone-0.5);
        
        \fill[tablegreen] (\rline, \rowone+0.5) rectangle (\Rline, \rowone-0.5);
        
        \fill[tablelime] (\mline, \rowtwo-0.5) -- (\rline, \rowtwo+0.5) -- (\rline, \rowtwo-0.5) -- cycle;
        
        \fill[tablelime] (\Rline, \rowtwo-0.5) -- (\Rline, \rowtwo+0.5) -- (\rline, \rowtwo-0.5) -- cycle;
        
        \fill[tableorange] (\mline, \rowthree-0.5) -- (\rline, \rowthree+0.5) -- (\mline, \rowtwo-0.5) -- cycle;
        \fill[tablelime] (\mline, \rowthree-0.5) -- (\rline, \rowthree+0.5) -- (\rline, \rowthree-0.5) -- cycle;
        
        \fill[tablered] (\rline, \rowthree+0.5) -- (\Rline, \rowthree+0.5) -- (\rline, \rowthree-0.5) -- cycle;
        \fill[tablelime] (\Rline, \rowthree-0.5) -- (\Rline, \rowthree+0.5) -- (\rline, \rowthree-0.5) -- cycle;
        
        \fill[tablelime] (\mline, \rowd+0.5) rectangle (\rline, \rowd-0.5);
        
        \fill[tablegreen] (\rline, \rowd+0.5) rectangle (\Rline, \rowd-0.5);

        \node[textentry] at (\ll, \h+0.2) {No. $\NN$-};
        \node[textentry] at (\ll, \h-0.2) {counters};
        \node[textentry] at (\rm, \h+0.25) {No. $\ZZ$-counters};
        \node[textentry] at (\lm,\h-0.25) {$k=0$};
        \node[textentry] at (\rm,\h-0.25) {fixed $k \geq 1$};
        \node[textentry] at (\rr,\h-0.25) {variable $k \in \NN$};
        \draw[tableline] (\ll-1, \h+0.5) -- (\rr+2, \h+0.5);
        \draw[tableline] (\lline, \h) -- (\rr+2, \h);
        \draw[tableline] (\ll-1, \h-0.5) -- (\rr+2, \h-0.5);
        
        \draw[tableline] (\Lline, \h+0.5) -- (\Lline, \rowd-0.5);
        \draw[tableline] (\lline, \h+0.5) -- (\lline, \rowd-0.5);
        \draw[tableline] (\mline, \h+0) -- (\mline, \rowd-0.5);
        \draw[tableline] (\rline, \h+0) -- (\rline, \rowd-0.5);
        \draw[tableline] (\Rline, \h+0.5) -- (\Rline, \rowd-0.5);
        
        \node[textentry] at (\ll, \rowzero) {$d=0$};
        \node[textentry] at (\lm, \rowzero+0.2) {\class{NL}-complete};
        \node[textentry] at (\lm, \rowzero-0.2) {\emph{Directed graph reachability}};
        \node[textentry] at (\rm, \rowzero+0.2) {\class{NP}-complete (folklore)};
        \node[textentry] at (\rm, \rowzero-0.2) {See also~{\cite[Thm.~3.5]{RosierY86}}};
        \node[textentry] at (\rr, \rowzero) {\class{NP}-complete (ILP)};
        \draw[tableline] (\ll-1, \rowzero-0.5) -- (\rr+2, \rowzero-0.5);
        
        \node[textentry] at (\ll, \rowone) {$d=1$};
        \node[textentry] at (\lm, \rowone+0.2) {\class{NP}-complete};
        \node[textentry] at (\lm, \rowone-0.2) {{\cite[Thm.~3.5]{RosierY86}} and~\cite{HaaseKOW09}};
        \node[textentry] at (\rm, \rowone) {\class{NP}-complete (Thm.~\ref{thm:1-zvass-np})};
        \node[textentry] at (\rr, \rowone) {\class{NP}-complete (Thm.~\ref{thm:1-zvass-np})};
        \draw[tableline] (\ll-1, \rowone-0.5) -- (\rr+2, \rowone-0.5);
        
        \node[textentry] at (\ll, \rowtwo) {$d=2$};
        \node[textentry] at (\lm, \rowtwo) {\class{PSpace}-complete~\cite{FearnleyJ15,BlondinFGHM15}};
        \node[textentry] at (\rm, \rowtwo+0.2) {Lower bound from $k=0$};
        \node[textentry] at (\rm, \rowtwo-0.2) {In $\Ff^{(k)}_3$ (See Thm.~\ref{thm:klm-d-vass})};
        \node[textentry] at (\rr, \rowtwo+0.2) {Lower bound from $k=0$};
        \node[textentry] at (\rr, \rowtwo-0.2) {In $\Ff_4$ (Thm.~\ref{thm:klm-d-vass})};
        \draw[tableline] (\ll-1, \rowtwo-0.5) -- (\rr+2, \rowtwo-0.5);
        
        \node[textentry] at (\ll, \rowthree) {$d=3$};
        \node[textentry] at (\lm, \rowthree+0.2) {Lower bound from $d=2$};
        \node[textentry] at (\lm, \rowthree-0.2) {In \class{2ExpSpace}~\cite{CzerwinskiJLO25}};
        \node[textentry, scale=0.875] at (\rm, \rowthree+0.2) {$\Oh(k)$-\class{ExpSpace}-hard (Cor.~\ref{cor:3-vass+kz})};
        \node[textentry] at (\rm, \rowthree-0.2) {In $\Ff^{(k)}_4$ (See Thm.~\ref{thm:klm-d-vass})};
        \node[textentry] at (\rr, \rowthree+0.2) {\class{Tower}-hard (Thm.~\ref{thm:unary-3-zvass-tower-hard})};
        \node[textentry] at (\rr, \rowthree-0.2) {In $\Ff_5$ (Thm.~\ref{thm:klm-d-vass})};
        \draw[tableline] (\ll-1, \rowthree-0.5) -- (\rr+2, \rowthree-0.5);
        
        \node[textentry] at (\ll, \rowd+0.2) {fixed};
        \node[textentry] at (\ll, \rowd-0.2) {$d\geq 4$};
        \node[textentry] at (\lm, \rowd+0.2) {$\Ff_{(d-3)/2}$-hard~\cite{CzerwinskiJLLO23}};
        \node[textentry] at (\lm, \rowd-0.2) {In $\Ff_d$~\cite{FuYZ24}};
        \node[textentry] at (\rm, \rowd+0.2) {Lower bound from $k=0$};
        \node[textentry] at (\rm, \rowd-0.2) {In $\Ff_{d+1}^{(k)}$ (See Thm.~\ref{thm:klm-d-vass})};
        \node[textentry] at (\rr, \rowd+0.2) {Lower bound from $k=0$};
        \node[textentry] at (\rr, \rowd-0.2) {In $\Ff_{d+2}$ (Thm.~\ref{thm:klm-d-vass})};
        \draw[tableline] (\ll-1, \rowd-0.5) -- (\rr+2, \rowd-0.5);
    \end{tikzpicture}
}

These tables do not contain all known complexity bounds for reachability in $d$-VASS.
Notice that for some small $d \geq 4$, $\Ff_{(d-3)/2}$-hardness is not informative; this lower bound becomes substantial when $d \geq 9$.
In fact, reachability in unary 8-VASS is known to be \class{Tower}-hard~{\cite[Theorem 1.4]{CzerwinskiO22}} (note that $\Ff_3 = $ \class{Tower}).
Furthermore, it is known that reachability in unary 5-VASS is \class{PSpace}-hard~{\cite[Theorem 1.2]{CzerwinskiO22}} and reachability in binary 6-VASS is \class{ExpSpace}-hard~{\cite[Theorem 1.3]{CzerwinskiO22}}.

\section{Preliminaries}
\label{sec:preliminaries}

\paragraph*{Basic Notions}
By $\QQ$, $\QQ_+$, $\ZZ$, $\NN$, and $\NN_+$, we respectively denote the sets of rational, positive rational, integer, nonnegative integer (natural), and positive integer numbers. 
For $a, b \in \ZZ$, $a \leq b$, we define $[a, b] \coloneqq \setof{x \in \ZZ}{a \leq x \leq b}$ and for $a \in \NN_+$ we define $[a] \coloneqq [1,a]$. 
The $j$-th coordinate of a vector $\vec{v} \in \QQ^d$ is indexed using square brackets $\vec{v}[j]$. 
We define the \emph{one-norm} of $\vec{v}$ as $\norm{\vec{v}} \coloneqq \onenorm{\vec{v}} = \abs{\vec{v}[1]} + \ldots + \abs{\vec{v}[d]}$ and the \emph{max-norm} of $\vec{v}$ as $\maxnorm{\vec{v}} \coloneqq \max\set{\abs{\vec{v}[1]}, \ldots, \abs{\vec{v}[d]}}$.
We use use the shorthand $\norm{\vec{v}}$ to denote the one-norm of $\vec{v}$, $\norm{\vec{v}} \coloneqq \onenorm{\vec{v}}$.
For vectors $\vec{u} \in \QQ^d$ and $\vec{v} \in \QQ^k$ we define vector $\vec{u} \oplus \vec{v}$ as the vector $\vec{w} \in \QQ^{d+k}$ such that, for $i \in [d], \vec{w}[i] = \vec{u}[i]$ and, for $i \in [d+1, d+k], \vec{w}[i] = \vec{v}[i - d]$; namely the first $d$ components of $\vec{w}$ equal the vector $\vec{u}$ and the last $k$ components $\vec{w}$ the vector $\vec{v}$. 

\paragraph*{Vector Addition Systems with States}
A $d$-dimensional Vector Addition System with States ($d$-VASS) is a non-deterministic finite automaton with $d$ nonnegative integer counters.  
Formally, a VASS consists of a finite set of \emph{states} $Q$ and finite set of \emph{transitions} $T \subseteq Q \times \ZZ^d \times Q$. 
For a transition $t = (p, \vec{v}, q) \in T$, we say that the transition has \emph{effect} $\vec{v} \in \ZZ^d$, which we denote by $\eff{t}$. 
A \emph{configuration} of a $d$-VASS consists of the current state $p \in Q$ and the current counter values $\vec{u} \in \NN^d$; denoted as  $\Config{p}{u} \in Q \times \NN^d$. 
A transition $t = (p, \vec{v}, q)$ can be \emph{fired} from a configuration $\Config{p}{u}$ if $\vec{u} + \vec{v} \geq \vec{0}$.
We write $\Config{p}{u} \xrightarrow{t} \Config{q}{u+v}$.
A \emph{path} is a sequence of transitions $\rho = (t_1, \ldots, t_\ell)$ such that, for every $i$, $t_i = (q_{i-1}, \vec{u}_i-\vec{u}_{i-1}, q_i)$ (i.e. the state that $t_{i-1}$ leads to is the same state that $t_i$ starts from). 
A \emph{run} is a sequence of configurations $(\config{q_0}{\vec{u}_0}, \config{q_1}{\vec{u}_1}, \ldots, \config{q_\ell}{\vec{u}_\ell})$ such that, for every $i \in [\ell]$, $(q_{i-1}, \vec{u}_i-\vec{u}_{i-1}, q_i) \in T$ is a transition; we may also write $\config{q_0}{\vec{u}_0} \xrightarrow{\rho} \config{q_\ell}{\vec{u}_\ell}$.
We may also write $\config{q_0}{\vec{u}_0} \xrightarrow{\rho}_\Vv \config{q_\ell}{\vec{u}_\ell}$ to specify that the run occurs in $\Vv$.
We also extend the definition of effect to paths: the effect of a path is the sum of the effects of its transitions $\eff{\rho} = \eff{t_1} + \eff{t_2} + \ldots + \eff{t_n}$.

We consider a variant of VASS, called \emph{VASS extended with integer counters} (denoted \zvass{}{}).
These are VASS in which some counters are required to remain nonnegative ($\NN$-counters) and others need not ($\ZZ$-counters). 
Syntactically, a \zvass{d}{k} is the same as a \hbox{$(d+k)$-VASS}, namely they consist of a finite set of states $Q$ and a finite set of transitions $T \sset Q \times \ZZ^{d+k} \times Q$.
Semantically, however, a configuration of a \zvass{d}{k} belongs to $Q \times (\NN^d \times \ZZ^k)$, instead of $Q \times \NN^{d+k}$.
We define projections $\pi_\NN : \NN^d \times \ZZ^k \to \NN^d$ by $\nproj{\vec{v}} = (\vec{v}[1], \ldots, \vec{v}[d]) \in \NN^d$, and $\pi_\ZZ : \NN^d \times \ZZ^k \to \NN^d$ by $\zproj{\vec{v}} = (\vec{v}[d+1], \ldots, \vec{v}[d+k]) \in \ZZ^k$.
The definitions of effects, paths, and runs remain unchanged for \zvass{d}{k}.
Though, to be clear, a transition $t = (p, \vec{v}, q)$ can be fired in a configuration $\Config{p}{u}$ if $\nproj{\vec{u}} + \nproj{\vec{v}} \geq \vec{0}$ (i.e.\ only the first $d$ counters must be nonnegative). Notice, that if $k = 0$, a \zvass{d}{k} is a $d$-VASS; and if $d = 0$, a \zvass{d}{k} may be called an integer $k$-VASS.

The main focus of this paper is the reachability problem for \zvass{}{}. 
We consider the situation in which both $d$  (the number of $\NN$ counters) and $k$ (the number of additional integer counters) do \emph{not} form part of the input.

\pagebreak
\begin{problem}[Reachability in \zvass{d}{k}.]\hfill
    \begin{description}
    \item[\hspace{1.15em}Input:] A \zvass{d}{k} and two configurations $s, t$.
    \item[\hspace{1.15em}Question:] Does $s \xrightarrow{*} t$ hold?
    \end{description}
\end{problem}

We also consider the situation in which $k$ (the number of additional $\ZZ$ counters) is part of the input to the reachability problem.

\begin{problem}[Reachability in \zvass{d}{}.]\hfill
    \begin{description}
    \item[\hspace{1.15em}Input:] A \zvass{d}{k}, two configurations $s, t$, and $k \in \NN$.
    \item[\hspace{1.15em}Question:] Does $s \xrightarrow{*} t$ hold?
    \end{description}
\end{problem}

We denote the size of VASS $\Vv$ by $\size(\Vv)$. 
Importantly, we distinguish unary and binary encoding. 
If the transition vectors of a VASS $\Vv = (Q, T)$ are encoded in unary, then we call it a \emph{unary VASS} and we use $\unarysize{\Vv} \coloneqq |Q| + \sum_{(p, \vec{v}, q) \in T} \onenorm{\vec{v}}$.
If numbers on transitions of a VASS $\Vv = (Q, T)$ are encoded in binary then we call it a \emph{binary VASS} and set $\binarysize{\Vv} \coloneqq |Q| + \sum_{(p, \vec{v}, q) \in T} \log_2(\onenorm{\vec{v}}+1) + 1$.

\paragraph*{Fast-growing Functions and Complexity Classes}
We define here a hierarchy of \emph{fast-growing functions} and their corresponding complexity classes. 
Concretely, for each $k \geq 1$ we define a function $F_k: \NN \rightarrow \NN$ and the  complexity class $\Ff_k$ which corresponds to the $k$-th level in the Grzegorczyk Hierarchy~{\cite[Subsections 2.3 and 4.1]{Schmitz16}}.

Let $F_1(n) \coloneqq 2n$ and, for $k > 1$, let $F_k(n) \coloneqq \underbrace{F_{k-1} \circ \cdots \circ F_{k-1}}_{n\ \text{times}}(1)$.

Notice that $F_2(n) = 2^n$ and $F_3(n) = \textup{Tower}(n)$. Using the functions $F_k$, we define corresponding classes $\Ff_i$ as problems solvable in deterministic time $F_i(f(n))$ where $f: \NN \rightarrow \NN$ is a function computable in deterministic time $F_{i-1}^m(n)$ for some fixed $m \in \NN$; concretely we define $\Ff_k \coloneqq \bigcup_{f \in \Fff_{k-1}} \mathrm{DTIME}(F_k(f(n)))$ where $\Fff_k \coloneqq \bigcup_{m \in \NN} \mathrm{FDTIME}(F_k^{m}(n))$.

Intuitively, problems in $\Ff_k$ are those that can be solved in time $F_k(n)$ under ``easier-than'' $\Ff_k$-reductions. 
For example the class $\Ff_3$, known as \tower, contains all the problems,
which can be solved in the deterministic $\text{Tower}(n)$-time, but also for example those, which can be solved in time $\text{Tower}(2^n)$ (because $\text{Tower}(2^n) = \text{Tower} \circ F_2$).
Notice, that there are many known variants of the definitions of the fast-growing function hierarchy. However, the classes $\Ff_k$ are robust with respect to minor modifications to the definition of the hierarchy of fast-growing functions, see~{\cite[Section 4]{Schmitz16}}.

\section{Reachability in \zvass{1}{}}
\label{sec:dimension-one}

In this section, we will focus on \zvass{1}{}; these are VASS with a single natural valued counter and a variable number of integer counters. 
Our main result is that reachability in binary \hbox{\zvass{1}{}} is in \np. 
This extends the best known results: reachability in binary 1-VASS is in \np~\cite{HaaseKOW09} and reachability in binary integer VASS (with a variable number of integer counters) is in \np.
The latter follows from the fact that integer VASS can efficiently be converted into instances of integer linear programming which can be decided in \np~\cite{LinearProgramming}.

\begin{restatable}{theorem}{dimensionOneNP}\label{thm:1-zvass-np}
    Reachability in binary \zvass{1}{} is in \np.    
\end{restatable}

Our main technical contribution of this section is~\cref{lem:linear-form-runs}.
Here we show that reachability in \zvass{1}{} is captured by \emph{linear path schemes}; namely reachability can always be witnessed by a run of the form $\rho_0 \sigma_1^* \rho_1 \cdots \rho_{\ell-1} \sigma_\ell^* \rho_\ell$.
Here, $\rho_0, \rho_1, \ldots, \rho_\ell$ are paths that connect disjoint cycles $\sigma_1, \ldots, \sigma_\ell$ (which each may be taken multiple times).
So long as the underlying paths and cycles are short, and the number of times each cycle is taken is at most exponential, then such a run can be specified with only polynomially many bits.
This is what allows us to prove that reachability in \zvass{1}{} is in \np.

Before detailing~\cref{lem:linear-form-runs}, we will prove a preliminary lemma: reachability in binary \zvass{1}{} can be witnessed by exponential length runs.

\begin{restatable}{lemma}{oneVassExponential} \label{lem:1-vass-exponential}
    Let $\Vv = (Q, T)$ be a \zvass{1}{k} and let $s, t \in Q$ be two states.
    If $\Config{s}{0} \xrightarrow{*} \Config{t}{0}$, then there exists a run $\Config{s}{0} \xrightarrow{\tau} \Config{t}{0}$ such that $\abs{\tau} \leq (nk)^{Ck}$, where $n = \abs{Q} + \sum_{(p, \vec{u}, q) \in T} \onenorm{\vec{u}}$ is the size of $\Vv$ encoded in unary and $C$ is a constant independent of $\Vv$, $s$, and $t$.
\end{restatable}

We will now sketch the ideas required to prove~\cref{lem:1-vass-exponential} (see~\cref{app:dimension-one} for a full proof).
We depend on~{\cite[Lemma 15]{AtigCHKSZ16}} as a core tool.
Roughly speaking,~{\cite[Lemma 15]{AtigCHKSZ16}} states that there is a polynomial semilinear description of the Parikh image of the language recognised by a given one-counter automaton (which is encoded in unary).
A One-Counter Automaton (OCA) is a finite automaton with a single $\NN$-counter which can be incremented, decremented, and zero-tested. 
The transitions are labelled by letters from an alphabet $\Sigma$.
A word $w \in \Sigma^*$ is accepted by an OCA if there exists a run from the initial state (with zero on the counter) to a final state (also with zero on the counter).
The Parikh image is a function $\psi : \Sigma^* \to \NN^{\Sigma}$ that maps a word to the vector whose entries contain the number of times each letter occurs in the word.
The function is also defined over languages, $\psi(L) \coloneqq \set{\psi(w) : w \in L}$.

More precisely,~{\cite[Lemma 15]{AtigCHKSZ16}} states that the Parikh image of the language of an $n$-state OCA $\Aa$, $\psi(L(\Aa))$, can be expressed as the union of linear sets such that 
\begin{enumerate}[(i)]
    \item the number of linear sets in the union is at most $C'n^{20}$, 
    \item the norm of the base vectors does not exceed $C'n^{20}$,
    \item the number of periods in each linear set is at most $C'n^{20}$, and
    \item the norm of the period vectors does not exceed $C'n^{20}$.
\end{enumerate}
Here, $C'$ is some fixed constant that is independent of $\Aa$.

In order to use~{\cite[Lemma 15]{AtigCHKSZ16}}, we need to convert a \zvass{1}{} $\Vv$ into an OCA $\Aa$ and be able to use $\psi(L(\Aa))$ to identify runs from $\Config{s}{0}$ to $\Config{t}{0}$ in $\Vv$.
For this, we use the one $\NN$-counter of $\Aa$ to track the value of the $\NN$-counter in $\Vv$ and we will use the letters that $\Aa$ reads to track the updates that the $\ZZ$-counters receive in $\Vv$.
Precisely, suppose that $\Vv$ is a \zvass{1}{k}, then $\Sigma = \set{a_1, b_1, \ldots, a_k, b_k}$.
For each $i$, the letter `$a_i$' will be read whenever the $i$-th $\ZZ$-counter is incremented (e.g.\ if the $i$-th $\ZZ$-counter is increased by five, then five `$a_i$'s are read) and the letter `$b_i$' will be read whenever the $i$-th $\ZZ$-counter is decremented. 
The number of states in $\Aa$ is exactly the size of $\Vv$ encoded in unary.

Now, given that the $i$-th $\ZZ$-counter initially begins with value~$0$, and the number of `$a_i$'s read and the number of `$b_i$'s read are equal, then the value of the $i$-th $\ZZ$-counter is 0.
Recall that all initial counter values are 0 and all target counter values are zero.
Our goal, therefore, is to determine whether there is a word $w$ accepted by $\Aa$ such that, for every $i$, $(\psi(w))[a_i] = (\psi(w))[b_i]$.

To identify words $w$ such that $w \in L(\Aa)$ and, for every $i$, $(\psi(w))[a_i] =  (\psi(w))[b_i]$, we shall intersect $\psi(L(\Aa))$ with the set $Z \coloneqq \set{\vec{v} \in \NN^\Sigma : \vec{v}[a_i] = \vec{v}[b_i], \text{ for every } i}$.
Fortunately, this set is semilinear.
We can therefore use the fact that, if two semilinear sets intersect, then their intersection contains a ``small'' point.
This fact follows from~{\cite[Theorem 6]{ChistikovH16}} in which the authors provide multi-parameter analysis of the size of the semilinear description of the intersection of two semilinear sets.
More precisely, for our purposes, if $\psi(L(\Aa)) \cap Z$ is not empty, then there exists $\vec{v} \in \psi(L(\Aa)) \cap Z$ whose entries are bounded above by a polynomial in $n$ to the power of $\Oh(\abs{\Sigma})$ (which is some constant multiple of $k$).

Overall, we therefore obtain a bound, that is polynomial in $n$ and exponential in $k$, on the length of a word $w$ such that $w \in L(\Aa)$ and, for every $i$, $(\psi(w))[a_i] = (\psi(w))[b_i]$.
Lastly, we use $w$ to construct a run from $\Config{s}{0}$ to $\Config{t}{0}$ in $\Vv$ to finally prove that reachability in \hbox{\zvass{1}{}} can be witnessed by an exponential length run (\cref{lem:1-vass-exponential}).

We remark that~\cref{lem:1-vass-exponential} suffices to conclude that, for every fixed $k$, reachability in unary \zvass{1}{k} is in \nl.
A short proof sketch of the following corollary can be found in~\cref{app:dimension-one}.

\begin{restatable}{corollary}{unaryNL} \label{cor:unary-NL}
    For every $k \in \NN$, reachability in unary \zvass{1}{k} is in \nl.
\end{restatable}

As outlined earlier, we intend to prove that reachability in binary \zvass{1}{} is in \np, even when the number of $\ZZ$-counters is not fixed.
We require exponential runs (\cref{lem:1-vass-exponential}) in order to prove there are polynomial size linear path schemes (\cref{lem:linear-form-runs}).

\begin{restatable}{lemma}{linearFormRuns}
    \label{lem:linear-form-runs}
    Let $\Vv = (Q, T)$ be a \zvass{1}{k}, let $s(\vec{x}), t(\vec{y}) \in Q \times (\NN \times \ZZ^k)$ be two configurations, and let $M = \max\set{\infnorm{\vec{u}} : (p, \vec{u}, q) \in T}$.
    If $s(\vec{x}) \xrightarrow{*}_\Vv t(\vec{y})$, then there exists a run $s(\vec{x}) \xrightarrow{\rho}_\Vv t(\vec{y})$ and a decomposition
    \begin{equation}\label{eq:lps}
        \rho = \rho_0 \sigma_1^{x_1} \rho_1 \cdots \rho_{\ell-1} \sigma_\ell^{x_\ell} \rho_\ell
    \end{equation}
    that satisfies:
    \begin{enumerate}[(1)]  
        \item $\ell \leq C_1k^2\abs{Q}^2\cdot\log(M)$; 
        \item $\rho_0 \rho_1 \cdots \rho_\ell$ is a path of length at most $C_2\abs{Q}^2$;
        \item $x_1, \ldots, x_\ell < (nk)^{C_3k}$; and
        \item $\sigma_1, \ldots, \sigma_\ell$ are simple cycles.
    \end{enumerate}
    Here, $n = \abs{Q} + (\sum_{(p, \vec{u}, q) \in T} \onenorm{\vec{u}}) + \onenorm{\vec{x}} + \onenorm{\vec{y}}$ denotes the size of the instance of reachability encoded in unary and $C_1$, $C_2$, and $C_3$ are constants independent of $\Vv$, $s(
    \vec{x})$, and $t(\vec{y})$.
\end{restatable}

We will now outline the proof of~\cref{lem:linear-form-runs} (see~\cref{app:dimension-one} for a full proof).
One of the main ideas is to perform a certain cycle replacement technique.
Cycle replacement is first used, though not named, by Rackoff to prove that unboundedness in VASS is witnessed by doubly-exponential length runs~{\cite[Section 4]{Rackoff78}}.
For our purposes, the idea is to carefully remove short cycles from a run that witnesses reachability and group them based on which state the $\NN$-counter attains its least value and whether or not the cycle has nonnegative effect on the $\NN$-counter.

Once the cycles are isolated and grouped, we use Eisenbrand and Shmonin's Carath\'{e}odory style bounds for integer cones~\cite{EisenbrandS06} to identify a polynomial collection of cycles whose effects can be combined to obtain the same effect as the original group of cycles.
Precisely, we will use the following corollary.
In the following statement, the integer cone of a finite set $X \subset \ZZ^d$ is denoted using $X^*$. 

\begin{corollary}[c.f.~{\cite[Theorem 1]{EisenbrandS06}}]
    \label{cor:caratheodory}
    Let $X \subset \ZZ^k$ be a finite set, let $\vec{s} \in X^*$, and let $M = \max\set{\infnorm{\vec{x}} : \vec{x} \in X}$.
    There exists a subset $Y \sset X$ such that $\vec{s} \in Y^*$ and $\abs{Y} \leq 2k\log(4kM)$.
\end{corollary}

Next, our goal is to reconstruct a run from $s(\vec{x})$ to $t(\vec{y})$ using the smaller collections of cycles.
To properly reconstruct a run, we must ensure that the $\NN$-counter remains nonnegative.
If a cycle has nonnegative effect on the $\NN$-counter, then we wish to take this cycle early in the run; otherwise if a cycle has negative effect on the $\NN$-counter, then we wish to take this cycle late in the run.
To achieve this, when removing cycles from the original run, we take care to not remove the first time or last time each state is visited.
This then allows us to insert nonnegative and negative cycles at the earliest and latest possible moments in the run, respectively.

We will use~\cref{lem:1-vass-exponential} to bound the length of the run that we perform cycle replacement on.
This means the total effect of cycles removed and replaced is also exponentially bounded.
We use this bound, along with sharp bounds on the size of solutions to systems of integer linear equations, to argue that $x_1, \ldots, x_\ell$ are exponentially bounded.
This, in turn, will be be used to prove that polynomially many bits suffice for specifying runs given by~\cref{lem:1-vass-exponential}. 

Now, we provide a sketch of the proof of~\cref{thm:1-zvass-np} which follows from \cref{lem:linear-form-runs}; a more detailed proof can be found in~\cref{app:dimension-one}.

\begin{proof}[Proof sketch of~\cref{thm:1-zvass-np}]
    Observe that the path $\rho$ produced by~\cref{lem:linear-form-runs}, underlying the run $s(\vec{x}) \xrightarrow{\rho}_\Vv t(\vec{y})$, can be described by listing: the underlying paths $\rho_0, \rho_1, \ldots, \rho_\ell$; the intermediate simple cycles $\sigma_1, \ldots, \sigma_\ell$; and the number of times each cycle is taken $x_1, \ldots, x_\ell$.
    Given properties (1) -- (4) of~\cref{lem:linear-form-runs}, it is clear that this description has polynomial size.
    In particular, we note that $x_1, \ldots, x_\ell$ can be written using binary encoding and property (3) tells us that $\log(x_1), \ldots, \log(x_\ell) \leq C_3k\cdot\log(nk) \leq C_3k^2\cdot\binarysize{\Vv}$.
    
    It is already known that the validity of linear path scheme runs, like $s(\vec{x}) \xrightarrow{\rho}_\Vv t(\vec{y})$, can be checked in polynomial time, see for example~{\cite[Section V, Lemma 14]{BlondinFGHM15}}.
    In short, one need not check \emph{every} configuration along the run; it suffices to check nonnegativity of the $\NN$-counter along the underlying paths and throughout the first and last iterations of each cycle $\sigma_i$.
    Therefore, only $3\ell$ configurations need to be computed and checked.
\end{proof}

With~\cref{thm:1-zvass-np} in hand, we can conclude that reachability in \zvass{1}{} is \class{NP}-complete, regardless of whether unary or binary encoding is used.
Hardness comes from either of the following facts.
Reachability in binary 1-VASS (without additional $\ZZ$-counters) is already \class{NP}-hard~{\cite[Theorem 3.5]{RosierY86}}~{\cite[Proposition 1]{HaaseKOW09}}. Reachability in unary integer VASS (without an $\NN$-counter) is already \class{NP}-hard (for details see~\cref{pro:unary-zvass} in~\cref{app:dimension-one}).

\begin{corollary}
    Reachability in binary \zvass{1}{} and reachability in unary \zvass{1}{} are both \np-complete.
\end{corollary}

\begin{remark}[Reachability in binary 1-VASS is in \class{NP}] \label{rem:1-VASS-reachability}
    The proof of~\cref{thm:1-zvass-np} (in particular, the proof of~\cref{lem:linear-form-runs}) provides an alternative and arguably simpler proof that reachability in binary 1-VASS is in \class{NP}.
    The original proof that reachability in binary 1-VASS is in \class{NP} used so called ``network flow reachability certificates''~\cite{HaaseKOW09}; here the authors had to carefully analyse these certificates to ensure that they pertained to valid runs.
    The techniques that we have developed to prove that reachability in \zvass{1}{} is in \class{NP} also work for 1-VASS (without additional $\ZZ$-counters).
    In fact, the proof of~\cref{lem:linear-form-runs} becomes simpler with the absence of $\ZZ$-counters\footnote{As it stands Lemmas~\ref{lem:1-vass-exponential} and~\ref{lem:linear-form-runs} currently do not hold for $k=0$ (for example, $\ell = 0$ would be required in condition (1) of~\cref{lem:linear-form-runs}).
    However, one can just add a single trivial $\ZZ$-counter which does not change value from $0$, and thus proceed with $k=1$.}.
    Our alternative proof of reachability in binary 1-VASS is in \class{NP} keeps the same core idea: we mark the first and last occurrences of states, remove simple cycles, use Carath\'{e}odory style bounds to replace the removed cycles with a smaller collection of cycles, and then reinsert these cycles back into a short underlying path.
\end{remark}

\section{Lower Bounds for \zvass{2}{} and \zvass{3}{}}
\label{sec:lower-bounds}
In this section, we will extensively use counter automata and certain results about them.

\paragraph*{Counter Automata}

A Counter Automaton (CA) is a VASS with an additional type of transition which can be fired only if particular counters are equal to zero. 
We call these transitions \emph{zero-testing} transitions. Formally, a $d$-counter automaton ($d$-CA) is a tuple $\mathcal{Z} = (Q, T, Z)$ where $(Q, T)$ is a $d$-VASS and $Z \subseteq Q \times [1,d] \times Q$ is the set of zero-testing transitions. 
The semantics of zero-testing transitions are defined as follows. 
Let $(p, i, q) \in Z$ be a zero-testing transition and let $\vec{v} \in \NN^d$. Then $p(\vec{v}) \xrightarrow{(p,i,q)} q(
\vec{v})$ is a run if $\vec{v}[i] = 0$, namely the transition $(p,i,q)$ can only be taken from a configuration $p(\vec{v})$ if the zero-tested $i$-th counter indeed has zero value. 

The reachability problem for a $d$-counter automaton, namely the question whether a given configuration is reachable from another given configuration, is undecidable already for two-counter automata~\cite{minsky}. 
Restricted versions of the reachability problem form a family of natural, computationally hard problems for various complexity classes. 

In \Cref{sec:2-VASS+Z}, the $2^{f(n)}$-length bounded reachability problem will be useful for us. Note, that it is slightly different from the more common problem in the literature, in which the  bound applies to the counter values instead of to the length of the run (see~\cite{fischer1968counter}). 

\begin{problem}[The $2^{f(n)}$-length bounded reachability problem for $d$-counter automata.]\hfill
    \begin{description}
    \item[\hspace{1.15em}Input:] A $d$-CA $\Aa$ with an initial state $q_I$ and final state $q_F$, and $n \in \NN$ encoded in unary.
    \item[\hspace{1.15em}Question:] {Does there exist a run of length $2^{f(n)}$ from $q_I(\vec{0})$ to $q_F(\vec{0})$ in $\Aa$?}
    \end{description}
\end{problem}

We say that a function is \emph{space-constructible} if there exists a Turing machine $\mathcal{M}$ which on input $1^n$ halts with output $1^{f(n)}$ that uses $f(n)$ space and $f(n) \geq \log(n)$. 
We claim that the following problem is $\spacecl(f(n))$-hard. 
The proof of~\cref{lem:reachability-CA} an be found in~\cref{app:lower-bounds}.

\begin{restatable}{lemma}{CA}\label{lem:reachability-CA}
The \hbox{$2^{f(n)}$-length} bounded reachability problem for three-counter automata is $\spacecl(f(n))$-hard.
\end{restatable}

In~\cref{sec:3-VASS+Z}, we use a version of the reachability problem for counter automata in which counter values are bounded by the \tower function. 
We say that a run of a $d$-CA is $b$-bounded if the sum of all the $d$ counters along this run never exceeds $b \in \NN$.  

\begin{problem}[The \tower-bounded reachability problem for $d$-CA.]\hfill
    \begin{description}
    \item[\hspace{1.15em}Input:] A $d$-CA $\Aa$ with an initial state $q_I$ and final state $q_F$, and $n \in \NN$ encoded in unary.
    \item[\hspace{1.15em}Question:] {Does there exist a $\tower(n)$-bounded run from $q_I(0^d)$ to $q_F(0^d)$ in $\Aa$? }
    \end{description}
\end{problem}

The \tower-bounded reachability problem for three-counter automata is well known to be \tower-complete (for example, see~{\cite[Theorem 7]{CzerwinskiO22}}). 
For the following theorem, one can use the famous \emph{G\"{o}del encoding} of vectors $(x, y, z)$ as the number $2^x \cdot 3^y \cdot 5^z$ and use the second counter as an ancillary counter to perform additive operations to $x$, $y$, and $z$.

\begin{theorem}\label{thm:2-CA-tower}
    The \tower-bounded reachability problem for 2-CA is \tower-complete.
\end{theorem}

\paragraph*{Counter Programs}
In Sections~\ref{sec:2-VASS+Z}~and~\ref{sec:3-VASS+Z}, we often present VASS as counter programs just like Esparza did in~{\cite[Section 7]{Esparza96}}; for more recent examples, see~\cite{CzerwinskiO21,CzerwinskiO22,Lasota22,ChistikovCMOSW24}. 
Counter programs are numbered lists of increment or decrement instructions operating on variables representing VASS counters.
We also allow for loop instructions, which can be naturally interpreted as cycles in VASS. 
We sometimes use other counter programs as macros inside a main counter programs. 
We also use for-loop instructions as a shorthand for a sequence of similar instructions.

\subsection{Reachability in unary \zvass{2}{} is \pspace-hard}
\label{sec:2-VASS+Z}

This section is devoted to proving the following theorem.

\begin{theorem} \label{thm:unary-2-zvass-psapce-hard}
    Reachability in unary \zvass{2}{} is \pspace-hard.
\end{theorem} 

The proof follows from two observations (Lemmas~\ref{lem:twoexptriple}~and~\ref{lem:big-numbers-sufficient}) that are interesting in their own right.
Roughly speaking,~\cref{lem:twoexptriple} says that \zvass{2}{} are able to compute doubly-exponential numbers, and~\cref{lem:big-numbers-sufficient} says that this is the key ingredient that suffices for \pspace-hardness. 
However, we formulate both lemmas with greater generality; we believe these statements might be useful in contexts outside of this paper.

Before stating both lemmas, we shall formalise what it means for a \zvass{}{} to compute a number or a vector. 
We say that a \zvass{d}{k} $\Vv$ \emph{computes} a set of vectors $S \subseteq \ZZ^m$ from a vector $\vec{u} \in \NN^m$ if there are distinguished states $q_I$ and $q_F$ such that $S$ is the set of all vectors $\vec{v}$ such that there is a run from $q_I(\vec{u} \oplus 0^{d+k-m})$ to a configuration $q_F(\vec{v} \oplus 0^{d+k-m})$.
If $\Vv$ computes $S$ from $0^d$, then we simply say that $
\Vv$ computes $S$. 
Notice that we do not require the computed vectors $\vec{v}$ to be located only on $\NN$-counters or only on $\ZZ$-counters.

Observe that we can use $\ZZ$-counters to perform zero tests on the other counters.
Importantly, we also assume that the zero tests are fired in a fixed order and from fixed states of the \zvass{2}{}.
More precisely, in our case the considered \zvass{2}{} consist of a sequence of Strongly Connected Components (SCCs) and only zero-testing transitions occur between SCCs. 
Indeed, in order to zero test a counter, we maintain a copy on some auxiliary $\ZZ$-counter and at the moment the zero test is taken, we ``freeze'' (no longer update) the copy, and add the requirement that the copy has zero value in the target configuration of the reachability instance. 
Assuming that the considered \zvass{2}{} has polynomially many auxiliary $\ZZ$-counters, we can simulate polynomially many zero tests.

\begin{claim}\label{cl:zero-tests}
    One can assume, without loss of generality, that the considered \zvass{2}{} uses polynomially many zero tests.
\end{claim}

We are now ready to formulate the two main lemmas.

\begin{restatable}{lemma}{twoexptriple}\label{lem:twoexptriple}
    Let $n, A \in \NN$.
    There is exists a \zvass{2}{}, with polynomial size in $n$ and $A$, that computes the set of triples $(A^{2^n}, B, B \cdot A^{2^n})$, for any $B \in \NN$.
\end{restatable}

Lemma~\ref{lem:twoexptriple} is proved in~\cref{app:lower-bounds}.
The main idea behind the proof of~\cref{lem:twoexptriple} is that we start from a triple $(A, B, B \cdot A)$ and repeat the following procedure $n$ times: from a triple $(C, B, CB)$, we produce a triple $(C^2, B', C^2 B')$, for some $B' \in \NN$. 
As $n$-fold squaring $A$ yields $A^{2^n}$, we obtain the desired triple after the procedure is repeated $n$ times.

\begin{lemma}\label{lem:big-numbers-sufficient}
    For any $d \geq 2$ and space-constructible function $f$, if for every $n \in \NN$, there exists a \zvass{d}{k}, of size polynomial in $n$, that computes $(2^{f(n)}, 7^{2^{f(n)}})$, then reachability in \zvass{d}{k} is \class{Space$(f(n))$}-hard. 
    In particular, if there exists a  \zvass{d}{} that can compute $(2^n, 7^{2^n})$, then reachability in \zvass{d}{} is \pspace-hard.
\end{lemma}

\begin{remark}
    Notice that Lemma~\ref{lem:big-numbers-sufficient} may be seen as a weak version of an ultimate conjecture for hardness of reachability problem in infinite-state systems with counters. (See Conjecture~\ref{con:long-runs-iff-hardness}, that roughly states that enforcing long runs implies hardness.) 
    The multiplication triple technique from~\cite{CzerwinskiLLLM19} may be informally described as follows: in order to show \class{Space($f(n)$)}-hardness of a class of VASS, it suffices to compute triples $(2^{f(n)}, B, 2^{f(n)} \cdot B)$. 
    One would hope that it suffices to just compute one number that is exponential in $f(n)$. 
    \cref{lem:big-numbers-sufficient} is a step in this direction, as it shows that it suffices to compute the pair of numbers $(2^{f(n)}, 7^{2^{f(n)}})$ (note that the constant $7$ is not special and can be changed). 
    This pair of numbers is reminiscent of the quadratic pairs technique of~\cite{CzerwinskiO22} that, roughly speaking, shows that computing the pair $(2^{f(n)}, 4^{f(n)})$ suffices for \class{Space($f(n)$)}-hardness.
    However, use the quadratic pair of counters for hardness requires three $\NN$-counters to be used in total, but~\cref{lem:big-numbers-sufficient} only needs two $\NN$-counters.
\end{remark}

Before proving~\cref{lem:big-numbers-sufficient}, we show how Lemmas~\ref{lem:twoexptriple} and~\ref{lem:big-numbers-sufficient} are used to prove~\cref{thm:unary-2-zvass-psapce-hard}.

\begin{proof}[Proof of Theorem~\ref{thm:unary-2-zvass-psapce-hard}]
    Given~\cref{lem:twoexptriple}, we have \zvass{2}{} that can compute triples $(A^{2^n}, B, B \cdot A^{2^n})$. 
    This immediately implies that \zvass{2}{} can in particular compute exponential numbers of the form $A^{2^n}$, in particular with $A = 7$. 
    We can simply nondeterministically decrease the second and third counter in the triple, and at the end zero test them.
    In order to satisfy assumptions of~\cref{lem:big-numbers-sufficient}, we need to show also how to compute $2^n$ with \zvass{2}{}. 
    Given~\cref{cl:zero-tests}, this is straightforward: one can use linearly many zero tests to $n$-times multiply the constant $2$ with itself.
    Assume that $x$ and $y$ are the two $\NN$-counters and they both start with value $0$. We first increment $x$ by $1$ and then $n$-times multiply $x$ by $2$ (using $y$ as an ancillary counter $y$).
    For counters $t, t_1, \ldots, t_k$, we write $t \reaches t_1, \ldots, t_k$ to represent the following operation: decrement $t$ and simultaneously increment (each) $t_1, \ldots, t_k$ by one.
    
    Exact multiplication by any rational number $\frac{a}{b}$, for $a, b \in \NN$,
    can be implemented by the following short counter program denoted $\ExactMult(x, y, a/b)$:
    
    \begin{algorithmic}[1]
    \Loop \ $x \reaches y$
    \EndLoop
    \State \ZT($x$)
    \Loop \ \sub{y}{b}; \add{x}{a}
    \EndLoop
    \State \ZT($y$)
    \end{algorithmic}

    Now that we have \zvass{2}{} that can compute pairs $(2^n, 7^{2^n})$, by~\cref{lem:big-numbers-sufficient} we conclude that the reachability problem for \zvass{2}{} is \pspace-hard.
\end{proof}

\begin{proof}[Proof of Lemma~\ref{lem:big-numbers-sufficient}]
    We will use the hardness result formulated in~\cref{lem:reachability-CA}.
    The main challenge is to simulate a three-counter automaton with a \zvass{2}{} while maintaining the length the run (in order to check reachability is witnessed by a run of length at most $2^{f(n)}$).

    Before diving into the proof, we will intuitively explain our construction.
    A first idea would be to keep the three counters on VASS counters and somehow try to implement zero tests on VASS. 
    However, the \zvass{2}{} has only two $\NN$-counters and therefore such a simulation is hopeless; we cannot even guarantee nonnegativity of the counters. 
    Instead we propose a different, novel technique to simulate 3-CA.

    Suppose $\Aa$ is an arbitrary 3-CA; we shall use $x_1$, $x_2$ and $x_3$ to refer to the three counters.
    We shall construct a \zvass{2}{} that we denote as $\Vv$.
    The first try would be to encode counters $x_1$, $x_2$, and $x_3$ as $2^{x_1} \cdot 3^{x_2} \cdot 5^{x_3}$ and keep this encoding on one the two $\NN$-counters of $\Vv$.
    We denote $B_1 = 2$, $B_2 = 3$, and $B_3 = 5$.
    Let us use $x$ to denote the $\NN$-counter storing $2^{x_1} \cdot 3^{x_2} \cdot 5^{x_3}$ and we use $\bar{x}$ to denote the other $\NN$-counter of $\Vv$.
    If one of the counters $x_i$ is increased in $\Aa$, then $\Vv$ simulates this by ``weakly-multiplying'' $x$ by $B_i$ (using $\bar{x}$ as an ancillary counter).
    In other words $\Vv$ applies the following counter program, denoted $\WeakMult(x, \bar{x}, a/b)$, for $B_i = a/b$.

    \begin{algorithmic}[1]
        \Loop \ $x \reaches \bar{x}$
        \EndLoop
        \Loop \ \sub{\bar{x}}{b}; \add{x}{a}
        \EndLoop
    \end{algorithmic}

    If both loops are maximally iterated, then value of $x$ will be multiplied by $B_i$ exactly.
    However, the multiplication is ``weak'': the value of $x$ can be multiplied by $B_i$, but $x$ could also end up with some smaller value (because there is nothing forcing the VASS to maximally iterate the loops).

    Let us explain why weak multiplication is still useful even though it is not exact.
    The idea is that, for any $A$,~\cref{lem:big-numbers-sufficient} gives us the pair $(2^{f(n)}, A^{2^{f(n)}})$.
    Therefore, if one is able to weakly multiply some counter by $A$ exactly $2^{f(n)}$ many times, then we can compare its final value with the value $A^{2^{f(n)}}$.
    If the counter is required to equal $A^{2^{f(n)}}$ at the end of the run, then it must have been the case that every weak multiplication was actually an exact multiplication (otherwise, this counter would end with a value less than $A^{2^{f(n)}}$).

    We are not out of the woods yet: if we decrease some counter $x_i$ then $\Vv$ is expected to divide $x$ by $B_i$. 
    It is possible to weakly divide counter values in \zvass{2}{} in the same way that one can weakly multiply, but this creates the following issue.
    Given that weak multiplication can result in smaller than desired values, and weak division can result in greater than desired values, one can possibly balance weak multiplications and weak divisions carefully so that the result is the same as if the multiplications and divisions were computed exactly. 
    This possibility of deviating from the true desired counter values seems to be uncontrollable, even with the pair $(2^{f(n)}, A^{2^{f(n)}})$ at hand.

    Now we shall introduce a novel idea to handle this issue.
    Instead of $2^{x_1} \cdot 3^{x_2} \cdot 5^{x_3}$, the counter $x$ stores $2^{x_1} \cdot 3^{x_2} \cdot 5^{x_3} \cdot 7^L$, where $L$ is the length of the current run. 
    As $7$ is greater than $2$, $3$, and $5$, any step of the 3-CA will result in $x$ increasing (in short, $7^L$ is a dominant term). 
    For example, if the counter $x_i$ is decremented by $\Aa$, then in order to maintain the encoding $x = 2^{x_1} \cdot 3^{x_2} \cdot 5^{x_3} \cdot 7^L$, we wish to weakly multiply $x$ by $7 / B_i > 1$.
    Therefore, in every cases, $\Vv$ simulates a step of $\Aa$ with weak multiplication. 
    This means that any weak computation results in the value being smaller than desired (not smaller or bigger depending on the operation). 
    This behaviour, thankfully, is controllable.
    Precisely, if after exactly $2^{f(n)}$ steps, the values of counters $x_1$, $x_2$, and $x_3$ are zero, then the encoding should be $2^0 \cdot 3^0 \cdot 5^0 \cdot 7^{2^{f(n)}} = 7^{2^{f(n)}}$. 
    Recall that we were able to exactly compute the number $7^{2^{f(n)}}$, and store it on some auxiliary $\ZZ$-counter.
    This allows us to check $x = 7^{2^{f(n)}}$ at the end of the run.
    If equality holds, then every weak multiplication must have been performed exactly, and so the simulation of $\Aa$ is faithful.

    We will now describe the construction more precisely. 
    The \zvass{2}{} $\Vv$ consists of two $\NN$-counters $x$ and $\bar{x}$, $\ZZ$-counters $y_1$, $y_2$, $y_3$, $z_1$, and $z_2$,
    as well as several additional, unnamed auxiliary $\ZZ$-counters used for zero tests.
    Given~\cref{lem:big-numbers-sufficient}, the considered systems (\zvass{2}{} in our case, \zvass{d}{} more generally) can compute pairs $(2^{f(n)}, 7^{2^{f(n)}})$. 
    We assume, without loss of generality, that the pair is computed and stored on counters $z_1$ and $z_2$.
    All the other counters $x$, $\bar{x}$, $y_1$, $y_2$, and $y_3$ start with value zero.
    We aim to store $2^{x_1} \cdot 3^{x_2} \cdot 5^{x_3} \cdot 7^L$ on $x$ (recall that $L$ is the length of the run). 
    We will use $\bar{x}$ as an ancillary counter, most of the time its value will be $0$.
    The values $x_1$, $x_2$, and $x_3$ are maintained on counters $y_1$, $y_2$, and $y_3$, respectively. 
    The value $2^{f(n)} - L$ will be stored on $z_1$.
    We will check that $z_1$ has a nonnegative value; this implies that $L \leq 2^{f(n)}$.
    In fact, our simulation will have the property that the
    sum $x + \bar{x}$ will always be upper bounded by
    $2^{x_1} \cdot 3^{x_2} \cdot 5^{x_3} \cdot 7^L$.
    At the end, we will check the equality $x = z_2$ and this equality will hold if and only if the simulation of $\Aa$ by $\Vv$ was faithful (the value $z_2$ is not changed during simulation).
    Although $y_1$, $y_2$, and $y_3$ track the values of $x_1$, $x_2$, and $x_3$, we are not able to check nonnegativity or perform an arbitrary number of zero tests on them because $y_1$, $y_2$, and $y_3$ are $\ZZ$-counters of $\Vv$.  
    They are simply present to check that their values are zero at the end of the run.

    Recall that, to begin with, $x_1 = x_2 = x_3 = L = 0$. 
    We set value of the counter $x$ to be $2^0 \cdot 3^0 \cdot 5^0 \cdot 7^0 = 1$.
    We need to explain in detail how $\Vv$ simulates steps of $\Aa$: what happens when a counter of $\Aa$ is incremented, decremented, or zero-tested.
    First, all increments and decrements $x_i$ implies that $y_i$ is incremented or decremented, respectively. 
    The more challenging part is the simulation on $x$ and $\bar{x}$. 
    When counter $x_i$ is increased, the encoding on $x$ should be multiplied by $7 \cdot B_i$, as both $x_i$ and $L$ are increased by one. 
    We implement this with an instance of the weak multiplication counter program $\WeakMult(x, \bar{x}, 7 \cdot B_i)$. 
    Similarly, when $x_i$ is decreased, $x$ should be multiplied by $7 / B_i > 1$ and that is implemented by firing the weak multiplication counter program $\WeakMult(x, \bar{x}, 7 / B_i)$. 
    Notice that if the counter $x_i$ in $\Aa$ is zero and cannot be decremented then value of the counter $x$ in $\Vv$ is not divisible by $B_i$. 
    The \zvass{2}{} $\Vv$ does not explicitly prohibit the firing of this transition, however after the firing, the sum $x + \bar{x}$ will be multiplied by less than $7 / B_i$. 
    As we will soon see, this will result in the target configuration not being reached.

    Let us now inspect how zero tests of $x_i$ are simulated. 
    Notice that $x_i = 0$ if and only if $B_1^{x_1} \cdot B_2^{x_2} \cdot B_3^{x_3} \cdot 7^L$ is not divisible by $B_i$.
    Therefore, the zero test is simulated as a test for \emph{non}-divisibility of $x$ by $B_i$.
    Importantly, since the zero test of $x$ is a step in $\Aa$, we must not forget to (weakly) multiply $x$ by $7$.
    Accordingly, we design the zero test simulation macro such that value $x + \bar{x}$ is multiplied by $7$ if and only if $x$ was originally not divisible by $B_i$. 
    At the same time, it is important that before the simulation $\bar{x}$ had zero value, and afterwards $\bar{x} = 0$.

    Altogether, we nondeterministically guess the residuum $r \in \{1, \ldots, B_i - 1\}$, subtract $r$ from $x$, check divisibility by $B_i$, multiplying $x$ by $7$, and then add $7r$ back to $x$. 
    This schema is realised by the following counter program:

    \begin{algorithmic}[1]
        \State \Guess \text{ } $r \in \{1, \ldots, B-1\}$
        \State \sub{x}{r}
        \Loop
            \State \sub{x}{B}; \add{\bar{x}}{B}
        \EndLoop
        \Loop
            \State \sub{\bar{x}}{B}; \add{x}{7B}
        \EndLoop
        \State \add{x}{7r}
    \end{algorithmic}

    It is easy to verify that if $x$ is not divisible by $B$, then final value of $x$ is can be multiplied by $7$ by selecting the correct residuum $r$ on line 1, and then maximally iterating the loops on lines 3 and 5.
    On the other hand, if $x$ is divisible by $B$ then no matter what one guesses for $r$, the counter $x$ during the loop iterations on line 3 is not divisible by $B$. 
    Therefore, the loops in lines 3 and 5 multiply value $x + \bar{x}$ by some value less than $7$.

    It remains to describe what $\Vv$ does after simulating the whole run of $\Aa$. 
    If $\Vv$ correctly and faithfully simulated a run of $\Aa$ to $x_1, x_2, x_3 = 0$ with $L < 2^{f(n)}$, then $\Vv$ reached, with maximal iterations of weak multiplication loops, the following counter values:
    \begin{equation}\label{eq:final}
        (x, \bar{x}, y_1, y_2, y_3, z_1, z_2) = (7^L, 0, 0, 0, 0, 2^{f(n)}-L, 7^{2^{f(n)}}).
    \end{equation}
    It is important that if $y_1 = y_2 = y_3 = 0$, then it is true that 
    \begin{equation}\label{eq:invariant}
        (x + \bar{x}) \cdot 7^{z_1} \leq z_2.
    \end{equation}
    In order to finally check whether the above seven-tuple is of the correct form, we need to check $\bar{x} = y_1 = y_2 = y_3 = 0$ and $x \cdot 7^{z_1} = z_2$. 
    This final check is realised by the following counter program:

    \begin{algorithmic}[1]
        \State \ZT($\bar{x}$); \ZT($y_1$); \ZT($y_2$); \ZT($y_3$)
        \Loop
            \State \WeakMult($x$, $\bar{x}$, 7)
            \State \sub{z_1}{1}
        \EndLoop
        \State \ZT($z_1$)
        \Loop
            \State \sub{x}{1}; \sub{z_2}{1}
        \EndLoop
        \State \ZT($x$); \ZT($z_2$)
    \end{algorithmic}

    To see that this counter program indeed guarantees Equation~\eqref{eq:final}, first notice that in line 5, $z_1$ has value zero.
    It implies that initially $z_1 \geq 0$, so indeed $L \leq 2^{f(n)}$, as required.
    Observe also that in line 5 we need to have $x = z_2$, so $(x + \bar{x}) \cdot 7^{z_1} \geq z_2$. Each iteration of the loop in line 2 does not increase the value of $(x + \bar{x}) \cdot 7^{z_1}$
    and initially this value was at most $z_2$ by Equation~\eqref{eq:invariant}.
    That means that equality holds in both cases and, in particular, $(x + \bar{x}) \cdot 7^{z_1} = z_2$ initially. 
    Together with the zero-tests performed in line 1, we conclude that Equation~\eqref{eq:final} holds.

    Let us now put all the pieces together to show that $\Vv$ indeed faithfully simulates $\Aa$. 
    At the beginning we have $x + \bar{x} = 2^{x_1} \cdot 3^{x_2} \cdot 5^{x_3} \cdot 7^L$.
    Every step of the simulation of $\Aa$, namely any increment, decrement, or zero test of a counter $x_i$, is designed such that the above equality implies that the step was faithfully simulated.
    Once this equality is lost (i.e. when $x + \bar{x} < 2^{x_1} \cdot 3^{x_2} \cdot 5^{x_3} \cdot 7^L$), equality can never be recovered. Therefore, to guarantee the faithful simulation this equality need only be checked once at the end of the rum; this equality is indeed checked by the previously described final check. 
\end{proof}

\subsection{Reachability in unary \zvass{3}{} is \tower-hard.} 
\label{sec:3-VASS+Z}

This section is devoted to showing the following theorem.

\begin{theorem}\label{thm:unary-3-zvass-tower-hard}
    Reachability in unary \zvass{3}{} is \tower-hard.
\end{theorem} 

We achieve this, for the most part, by appropriately modifying the proof that reachability in 8-VASS is \tower-hard~\cite{CzerwinskiO22}.
We first briefly describe the overall ideas before diving into the details. 
We reduce from the \tower-complete problem: \tower-bounded reachability in two-counter automata (see~\cref{thm:2-CA-tower}). 

The challenge is to simulate \tower-boundedness of the counters and
the zero tests. 
This is because two of the $\NN$-counters of a \zvass{3}{} can be dedicated to storing the value of the two counters of a given 2-CA; this means that we do not need to be concerned with issues related to nonnegativity or the final values reached.
An established technique for simulating zero tests of bounded counters is the ``multiplication triples technique''.
This technique has been used, in particular, for showing \tower-hardness of reachability in 8-VASS and \ackermann-hardness of  reachability in VASS without a fixed number of counters~\cite{CzerwinskiLLLM19,CzerwinskiO21}.
The following lemma reformulates the multiplication triples technique for \zvass{3}{}.

\begin{restatable}{lemma}{threeVASSZtriples}\label{lem:3-VASS+Z-triples}
    For every 2-counter automaton $\Aa$ and in polynomial time, one can construct a unary \zvass{3}{} $\Vv$ with two distinguished states $q_I$, $q_F$ such that $\Aa$ has a $B$-bounded accepting run in which it performs at most $C$ if and only if there is a run from $q_I(0, 0, B, 2C, 2BC)$ to $q_F(0, 0, B, 0, 0)$ in $\Vv$ (here, we assume the first three counters are the $\NN$-counters).
\end{restatable}

The proof of~\cref{lem:3-VASS+Z-triples} analogous to the proof of~{\cite[Lemma 10]{CzerwinskiO22}} (which has a very similar statement).
One does need to make some small modifications; we present the whole proof in~\cref{app:lower-bounds}.

With~\cref{lem:3-VASS+Z-triples} in hand, it is enough to present a \zvass{3}{} that computes the triple $(\tower(n), 2C, 2C \cdot \tower(n))$, for an arbitrarily large $C \in \NN$.
This indeed enables us to simulate a $\tower(n)$-bounded runs of a given 2-CA and, thus we can deduce~\cref{thm:unary-3-zvass-tower-hard}. 
Therefore, our attention is now focused on constructing a \zvass{3}{} that computes $(\tower(n), 2C, 2C \cdot \tower(n))$, for an arbitrarily large $C \in \NN$. 

We call $(B, 2C, 2BC)$ a \emph{$B$-triple}. 
The idea is to first construct a $B$-triple, for some small $B$ and arbitrarily large $C$, and then use this $B$-triple to construct a $B'$-triple for $B' > B$. 
This construction is then repeated to generate large $B$-triples.
First, observe that it is trivial to construct a $16$-triple. 
It can be realised by the simple counter program presented
below.
Assume that the three $\NN$-counters are denoted $x$, $y$, and $z$. 
We use $\Triple(x, y, z)$ to refer to this program. 
The number $C$ corresponds to the number of iterations of the loop in line 2 (clearly $C$ can be arbitrarily large).

\begin{algorithmic}[1]
    \State \add{x}{16}
    \Loop
        \State \add{y}{2}; \add{z}{32}
    \EndLoop
\end{algorithmic}

Now we show how to construct a \zvass{3}{} that, from the triple $(B, 2C, 2 BC)$, produces the triple $(2^B, 2C', 2 \cdot 2^B \cdot C')$. 
Moreover, the produced $C'$ can be arbitrarily high (just like $C$).  
We shall assume that $B$ is divisible by $8$.
We also want $B = \tower(k)$ to hold for some $k \in \NN$ (this is why we started with $B = 16 = \tower(3)$ earlier).
More concretely, we will present a counter program denoted \ExpAmplifier($x_1, x_2, x, y, z, x', y', z', u$) that has three $\NN$-counters $x_1$, $x_2$, and $x$; six named $\ZZ$-counters $y$, $z$, $x'$, $y'$, $z'$, and $u$; and several unnamed $\ZZ$-counters that will be used for zero tests. 
We assume that initially $(x, y, z) = (B, 2C, 2BC)$ and all other counters have zero value.  
We will produce the triple $(2^B, 2C', 2 \cdot 2^B \cdot C')$ onto the counters $(x, y, z)$.

\cref{lem:3-VASS+Z-triples} tell us that the triple $(B, 2C, 2BC)$ can be used to simulate $C$ zero tests on $B$-bounded $\NN$-counters. However, we prefer here to use the triple $(B, 2C, 2BC)$ a bit differently: we use it to perform $B/2$ zero tests on $2C$-bounded $\NN$-counters. 
As we assume that $C$ is arbitrarily large, we are essentially allowing for the simulation of $B/2$ zero tests on $\NN$-counters without any assumptions on boundedness. 

It is important to emphasise that in our counter programs we use, in a sense, two kind of zero tests. 
One kind does not use any of the constructed triples $(B, 2C, 2BC)$ and relies instead on auxiliary $\ZZ$-counters. 
These zero tests can be applied on any counters, but we can only use  polynomially many of them (\cref{cl:zero-tests}).
The other kind of zero test do use the constructed triples $(B, 2C, 2BC)$. 
Here, we can allow for $B/2$ such zero-tests, but they will only ever be used on counters $x_1$ and $x_2$. 
For the remainder of this section, unless it is otherwise specified, we use zero tests of the first kind (one that uses an auxiliary $\ZZ$-counter, not a triple).

Before presenting \ExpAmplifier, we give a short counter program \Move($t, t'$) for the operation of moving a value from counter $t$ to counter $t'$. 
Initially, the counters have values $(t, t') = (k, 0)$ for some $k \in \NN$, and at the end of \Move($t, t'$), the counters have values $(t, t') = (0, k)$. 

\begin{algorithmic}[1]
    \Loop \ $t \reaches t'$
    \EndLoop
    \State \ZT($t$)
\end{algorithmic}

Recall now that \ExactMult($t, t', a/b$) exactly multiplies the value on counter $t$ by $\frac{a}{b}$.
Recall that this macro uses two zero-tests: one on counter $t$ and one on counter $t'$.
Here is the only place in which we use the zero tests of the second kind (which uses a constructed triple $(B, 2C, 2BC)$).
Accordingly, we can repeatedly iterate \ExactMult\ $B/4$ many times  under the condition that it uses counters $x_1$ and $x_2$. 

Now, We are ready now to present the counter program for \ExpAmplifier.

\begin{algorithmic}[1]
    \State \add{x'}{1}
    \Loop
        \State \add{y'}{2}; \add{z'}{2}
    \EndLoop
    \State \Move($x'$,$x_1$)
    \Loop
        \State \ExactMult($x_1$, $x_2$, 256)
        \State \add{u}{1}
    \EndLoop
    \State \Move($x_1$, $x'$)
    \State \Move($z'$,$x_1$)
    \Loop
        \State \ExactMult($x_1$, $x_2$, 256)
        \State \sub{u}{1}
    \EndLoop
    \State \Move($x_1$, $z'$)
    \State \ZT($y$); \ZT($z$); \ZT($u$)
    \Loop \ \sub{x}{1}
    \EndLoop
    \State \ZT($x$)
    \State \Move($x', x$); \Move($y', y$); \Move($z', z$)
\end{algorithmic}
\label{page:exp-amplifier}

Let us now inspect the code of \ExpAmplifier\ and argue its correctness. 
We shall use $C'$ to denote the number of iterations of the loop in line 2. 
After line 3, the values of $(x', y', z')$ are $(1, 2C', 2C')$ and note that $C'$ can be arbitrarily large. 
The aim of the majority of the rest of  \ExpAmplifier\ is to multiply both $x'$ and $z'$ by $2^B$.
Let $K_1$ and $K_2$ denote the number of iterations of the loops in line 5 and line 10, respectively.
Notice now that lines 4 -- 8 multiply $x'$ by $256^{K_1}$ and increase counter $u$ by $K_1$. 
The loop in line 5 does most of the work; lines 4 and 8 exist to move the value into $x_1$ (an $\NN$-counter).
Similarly lines 9 -- 13 multiply $z'$ by $256^{K_2}$ and decrease counter $u$ by $K_2$. 
After line 13, $u = K_1 - K_2$. 
As $u$ is zero tested in line 14, we know that $K_1 = K_2$. 
Moreover, in line 14 $y$ and $z$ both have zero value. 
As $y = 0$, we know that exactly $B/2$ zero tests using triples were performed, and $z = 0$ guarantees that all of them were indeed performed correctly. 
The number of zero tests using triples in each \ExactMult\ is two, so $2 K_1 + 2 K_2 = B/2$. 
Together with $K_1 = K_2$, we deduce that $K_1 = K_2 = B/8$. 
Therefore lines 4 -- 13 multiplied both $x'$ and $z'$ by $256^{B/8} = (2^8)^{B/8} = 2^B$. 
Thus, after line 14 the values $(x', y', z')$ are $(2^B, 2C', 2 \cdot 2^B \cdot C')$. 
The aim of lines 15 and 16 is to decrease $x$ to zero. 
Finally, line 17 moves the triple $(2^B, 2C', 2 \cdot 2^B \cdot C')$ from counters $(x', y', z')$ onto counters $(x, y, z)$, as required.

We are ready to present a counter program producing $\tower(n)$-triples.

\begin{algorithmic}[1]
    \State \Triple(x, y, z)
    \For {$i = 1, \ldots, n-3$}
        \State \ExpAmplifier($x_1$, $x_2$, $x$, $y$, $z$, $x'$, $y'$, $z'$, $u$)
    \EndFor
\end{algorithmic}

The triple $(16, 2C, 32 C)$ produced by \Triple($x, y, z$) can equivalently be seen as a $\tower(3)$-triple ($16 = \tower(3)$). 
Then, firing \ExpAmplifier\ $n-3$ times amplifies this triple to the $\tower(n)$-triple.
This finishes the proof of~\cref{thm:unary-3-zvass-tower-hard}.

\section{An Upper Bound for \zvass{d}{}}
\label{sec:klm-upper-bound}

This section is dedicated to proving the following theorem.

\begin{theorem}\label{thm:klm-d-vass}
    Reachability in \zvass{d}{} is in $\Ff_{d+2}$.
\end{theorem}

The $\Ff_{d+2}$ algorithm that we construct to prove~\cref{thm:klm-d-vass} is an extension of the KLMST algorithm for reachability in $d$-VASS; it is named KLMST after its inventors Kosaraju, Lambert, Mayer, Sacerdote, and Tenney \cite{SacerdoteT77, Mayr81, Kosaraju82, Lambert92}. 
In this section, we do not assume any prior knowledge of this algorithm. 
We will introduce all relevant notions, starting with the general structure of the algorithm.

The two main notions the algorithm revolves around are \emph{generalised VASS} and the subclass of \emph{perfect VASS}.
Let us start by introducing generalised VASS.
Many variants of generalised VASS appear in the literature: marked graph transition sequences in~\cite{Lambert92}, KLM sequences in~\cite{LerouxS19, FuYZ24}, \(Semil\)-VASS in \cite{GuttenbergCL25}.
The terminology ``generalised VASS'' was first used in \cite{Kosaraju82}.

\begin{definition}
    A \emph{generalised VASS} with \(d\) $\NN$-counters and \(k\) \(\ZZ\)-counters extends \hbox{\zvass{d}{k}} as follows.
    Each transition \(e\) that leaves an SCC is additionally labelled with an \hbox{\(\omega\)-configuration} \(test_e \in (\NN \cup \{\omega\})^d\).
    The semantics of such a transition $e$ are defined as follows.
    For \(\vect{x}, \vect{y} \in \NN^{d}\) and \(\vect{z}_1, \vect{z}_2 \in \ZZ^k\), we have \(p(\vect{x},\vect{z}_1) \xrightarrow{e} q(\vect{y}, \vect{z}_2)\) if and only if \begin{enumerate}[(i)]
        \item $e$ is a transition from $p$ to $q$, 
        \item \((\vect{y}, \vect{z}_2)=(\vect{x}, \vect{z}_1)+\effect(e)\), and
        \item \(\vect{x}[j]=test_e[j]\), for all \(j\) such that \(test_e[j] \neq \omega\). 
    \end{enumerate}
\end{definition}

We remark that we always use a binary encoding for generalised VASS.

In a generalised VASS, a transition \(e\) leaving an SCC may zero test some \(\NN\)-counters. 
As already mentioned in Section \ref{sec:lower-bounds}, a zero test can can be simulated using a dedicated \(\ZZ\)-counter, hence generalised \zvass{d}{} are equivalent to \zvass{d}{}, though the explicit tests can be convenient.

We define the reachability problem for generalised VASS in the obvious way.
Given a generalised VASS \(\Vv\), and two configurations \(\vect{s}, \vect{t}\), does \(\vect{s} \xrightarrow{*}_{\Vv} \vect{t}\) hold?

Perfect VASS are harder to define; the definition of perfectness will appear after the following introduction to the overall structure of the reachability algorithm.
For now, one should think that perfectness is some relatively easy-to-decide property of generalised VASS that is strong enough to imply that reachability holds.

To decide reachability, one first checks whether the current reachability query \((\vect{s}, \Vv, \vect{t})\) is perfect.
If it is, then reachability holds.
Otherwise, \((\vect{s}, \Vv, \vect{t})\) is decomposed into finitely many queries \((\vect{s}_i, \Vv_i, \vect{t}_i)\) 
such that 
(i) $\vec{s} \xrightarrow{*}_\Vv \vec{t}$ if and only if for some \(i\), $\vec{s}_i \xrightarrow{*}_{\Vv_i} \vec{t}_i$, and 
(ii) each instance \((\vect{s}_i, \Vv_i, \vect{t}_i)\) is smaller than the original instance \((\vect{s}, \Vv, \vect{t})\) with respect to a well-founded rank.
Finally, we solve these smaller queries by recursion.
Termination and the complexity of this algorithm will be based upon the well-founded rank. 
Correctness of the algorithm relies on the following lemmas.

\begin{restatable}{lemma}{LemmaPerfectnessImpliesRun} \label{LemmaLabelPerfectnessImpliesRun}
    If \((\vect{s}, \Vv, \vect{t})\) is perfect, then $\vec{s} \xrightarrow{*}_\Vv \vec{t}$. 
\end{restatable}

\begin{restatable}{lemma}{LemmaKLMSTDecomposition} \label{LemmaLabelKLMSTDecomposition}
    There is an elementary-time algorithm that decides whether a given query \((\vect{s}, \Vv, \vect{t})\) is perfect, and if not outputs finitely many queries \((\vect{s}_i, \Vv_i, \vect{t}_i)\) such that \(\rank(\Vv_i) < \rank(\Vv)\) and $\vec{s} \xrightarrow{*}_\Vv \vec{t}$ if and only if for some \(i\), $\vec{s}_i \xrightarrow{*}_{\Vv_i} \vec{t}_i$.
    
    Furthermore, if only the least important coordinate of the rank decreases, then every query \((\vect{s}_i, \Vv_i, \vect{t}_i)\) has polynomial size with respect to the size of \((\vect{s}, \Vv, \vect{t})\).
\end{restatable}

The rest of this section is structured as follows.
In~\cref{SubsectionDefinitionRank}, we define the rank. In~\cref{SubsectionDefinitionPerfectness}, we define perfectness. 
The proof of~\cref{LemmaLabelPerfectnessImpliesRun}, which is similar to proving the same statement about VASS without additional integer counters, can be found in~\cref{app:klm-upper-bound}.
In~\cref{SubsectionDecidingPerfectnessAndDecomposition}, we prove~\cref{LemmaLabelKLMSTDecomposition}.
In~\cref{SectionKLMSTComplexity}, we will analyse the complexity.

\subsection{Definition of the Rank} \label{SubsectionDefinitionRank}

Roughly speaking, for readers familiar with \emph{cycle spaces}, our rank is an extension of the rank defined by counting the number of states with a given dimension of the cycle spaces~\cite{LerouxS19}.

We start by defining \emph{cycle spaces}. 
For every state \(q \in Q\), we consider the effect vectors of all cycles on state \(q\); let \(\vectSet{V}(\Vv, q)\) be the vector space generated by all effects of cycles on $q$. Equivalently, \(\vectSet{V}(\Vv, q)\) is the vector space generated by all effects of simple cycles \emph{in the SCC} that contains \(q\). 
In particular, notice that \(\vectSet{V}(\Vv, q)\) only depends on the SCC that contains $q$, not on \(q\) itself.

Recall that \(\pi_{\NN}(\vec{v})\) is the projection of a vector $\vec{v} \in \NN^d \times \ZZ^d$ to $(\vec{v}[1], \ldots, \vec{v}[d])$, that is the projection to the \(\NN\)-counters.
We define the \emph{\(\NN\)-part of the rank}, \(\rank_{\NN}\in \NN^d\), as
\[\rank_{\NN}[i] \coloneqq |\{q \in Q : \dim(\pi_{\NN}(\vectSet{V}(\Vv,q)))=i\}|, \text{ for }i \in [d].\]
In other words, $\rank_{\NN}[i]$ is the number of states \(q \in Q\) whose cycle space \(\vectSet{V}(\Vv,q)\) has dimension \(i\) when projected to \(\NN\)-counters.
If the dimension of \(\pi_{\NN}(\vectSet{V}(\Vv,q))\) is \(0\), this state does not influence the rank, thus \(\rank_{\NN}\) has \(d\) components. So far, this is the same as the rank for VASS (without additional $\ZZ$-counters) introduced in~\cite{LerouxS19}.
The main difference is that our rank has the following extra component.
We define the \(\ZZ\)-vector space \(\vectSet{V}_{\ZZ}(\Vv)\) as
\[\vectSet{V}_{\ZZ}(\Vv):=\sum_{q \in Q} \vectSet{V}(\Vv,q) \cap (\{0^d\} \times \QQ^k).\]
In other words, we consider all elements of \(\vectSet{V}(\Vv,q)\) with zero effect on the  \(\NN\)-counters, and take their Minkowski sum. 
Intuitively speaking, this characterises the complexity of the behaviour of the \(\ZZ\)-counters.

Overall, we define the rank as
\begin{equation*}
    \rank \coloneqq (\dim(\vectSet{V}_{\ZZ}(V)), \rank_{\NN}).
\end{equation*}
Importantly, ranks are compared lexicographically with the first component, \(\dim(\vectSet{V}_{\ZZ}(V))\), taking the highest precedence, and \(\rank_{\NN}[j]\) has higher precedence than \(\rank_{\NN}[i]\), for \(j>i\).

\subsection{Perfect Instances of Reachability} \label{SubsectionDefinitionPerfectness}

We assume, without loss of generality, that the given generalised VASS \(\Vv\) is a sequence of SCCs connected by single transitions.
This can be guaranteed by splitting \(\Vv\) into one VASS \(\Vv_\pi\), for every sequence of SCCs \(\pi\) in \(\Vv\).
Observe that, while this decomposition takes exponential time (the number of simple paths), any \(\Vv_\pi\) created fulfills \(\size(\Vv_\pi) \leq \size(\Vv)\). 

The basis of perfectness is an integer linear program (ILP) for approximating reachability.
For readers familiar with KLMST for $d$-VASS (without additional $\ZZ$-counters)~\cite{LerouxS19, FuYZ24, GuttenbergCL25}, we remark that our definition of perfectness will differ only by the choice of ILP.

Let \(\Vv=(\Vv_0 e_1 \Vv_1 \dots e_{\NumberOfSCCs} \Vv_{\NumberOfSCCs})\) be a generalised VASS and let \(\vect{s}=(\vect{s}_{\NN}, \vect{s}_{\ZZ})\in \NN^d \times \ZZ^k\), \(\vect{t}=(\vect{t}_{\NN}, \vect{t}_{\ZZ})\) be two configurations. 
Let \(E_0, \dots, E_{\NumberOfSCCs}\) be the transitions of \(\Vv_0, \dots, \Vv_{\NumberOfSCCs}\), respectively. 
As always, \(d\) denotes the number of \(\NN\)-counters and \(k\) denotes the number of \(\ZZ\)-counters. 
We will use \(i\) as an index to refer to the current SCC, and \(j\) to refer to a particular counter.

There are \(2\NumberOfSCCs d+\sum_{i=0}^{\NumberOfSCCs} |E_{i}|\) variables in the ILP. 
The first \(2 \NumberOfSCCs d\) variables correspond to \(2\NumberOfSCCs\) vectors: for every $i \in \set{0} \cup [\NumberOfSCCs-1]$, $\vect{x}_{i}=(x_{i,j})_j \in \NN^d$, and for every $i \in [\NumberOfSCCs]$, $\vect{y}_{i}=(y_{i,j})_j \in \NN^d$.
These variables represent the counter values when the run enters and respectively leaves each SCC.
We set \(\vect{x}_0:=\vect{s}_{\NN}\) and \(\vect{y}_{\NumberOfSCCs}:=\vect{t}_{\NN}\). 
Finally, we also have a variable \(x_e\) for every transition \(e \in E_i\) in an SCC \(\Vv_i\).
These variables count how often the corresponding transition \(e\) is used.

The ILP is comprised of three sets of equations.
The first are the Kirchhoff equations governing reachability in the underlying graph.
Formally, for every \(i \in \{0\} \cup [\NumberOfSCCs]\) and \(q \in Q_i\), we have the equation 
\begin{equation*}
    \sum_{e \in in(q)} x_e -\sum_{e \in out(q)} x_e=\delta_{q, \qiin}-\delta_{q,\qiout},
\end{equation*} 
where \(in(q)\) is the set of transition incomings to $q$,  \(out(q)\) is the set of transitions outgoing from $q$, and 
\(\delta_{q,q'}\) is the Kronecker delta function that takes value $1$ if $q = q'$, otherwise it takes value $0$.
This requires \(q\) to be entered as often as it is left, with the exception that if $q$ is the initial or final state, then the above equations are adjusted by $\pm1$.

The second set of equations take care of the effect of the run.
For \(\ZZ\)-counters we only require a global effect equation:
\begin{equation*}
    \sum_{i=1}^{\NumberOfSCCs} \effect(e_i) + \sum_{i=0}^{\NumberOfSCCs} \sum_{e \in E_i} \effect_{\ZZ}(e) \cdot x_e = \vect{t}_{\ZZ}-\vect{s}_{\ZZ}.
\end{equation*}
In other words, taking transitions as often as prescribed by \(x_e\), and the inter-SCC transitions \(e_i\) once, produces the desired effect \(\vect{t}_{\ZZ}-\vect{s}_{\ZZ}\) on the $\ZZ$ counters.
For \(\NN\)-counters, we have the following equation for each \(i \in \{0\} \cup [\NumberOfSCCs]\): 
\begin{equation*}
    \sum_{e \in E_i} \effect_{\NN}(e) \cdot x_e = \vect{y}_{i}-\vect{x}_i . 
\end{equation*}
This states that the effects on \(\NN\)-counters should match the difference between the claimed entry vectors $\vect{x}_i$ and exit vectors $\vect{y}_i$ of the SCC.

The final set of equations require the transitions \(e_i\) between SCCs \(\Vv_{i-1}\) and \(\Vv_i\) to be taken correctly (i.e. their equality tests are respected).
For every $i \in [s]$, we add the equation(s)
\begin{align*}
    & \vect{y}_{i-1} + \effect_{\NN}(e_i) = \vect{x}_i \text{ and } \\
    & \vect{y}_{i-1}[j] = test_{e_i}[j], \text{ for every } j \text{ such that } test_v[j] \neq \omega.
\end{align*}

Given a reachability query \((\vect{s}, \Vv, \vect{t})\), we use \(\ILP(\vect{s}, \Vv,\vect{t})\) to denote the above ILP.
Before defining perfectness of \((\vect{s}, \Vv, \vect{t})\), we first define when \(\ILP(\vect{s}, \Vv,\vect{t})\) is perfect.

\begin{definition}
    We say that \(\ILP(\vect{s}, \Vv,\vect{t})\) is \emph{perfect} if
    \begin{enumerate}[(i)]
        \item \(\ILP(\vect{s}, \Vv,\vect{t})\) has a solution;
        \item there is a homogeneous solution \(\vect{h}\) such that \(\vect{h}[x_e]>0\), for every transition \(e\); and 
        \item for every \(i \in [\NumberOfSCCs]\) and counter \(j \in [d]\), if \(test_{e_i}[j]=\omega\), then the variable \(x_{i,j}\) is unbounded in the ILP.
    \end{enumerate}
\end{definition}

Condition (i) is as expected; any run \(\rho\) from $\vect{s}$ to $\vect{t}$ of \(\Vv\) gives rise to a solution to the ILP, so the ILP has to be satisfiable.
Condition (ii) essentially states that every transition can be taken arbitrarily often; this is a typical and crucial aspect of perfectness that we will leverage later.
Condition (iii) states, in words, that if the $j$-th counter is not zero tested when entering an SCC \(\Vv_i\), then the $j$-th counter is unbounded at this point in the run.
This condition is most useful in contraposition: if the variable \(x_{i,j}\) is bounded, then \(test_{e_i}[j] \neq \omega\).
This means that if the $j$-th counter is bounded when entering some SCC, then it is tested for a specific value.

The second half of perfectness now requires that \(test_{e_i}+\effect_{\NN}(e_i)\), the \(\omega\)-configuration reached when entering \(\Vv_i\), is ``pumpable'' (inside the SCC $\Vv_i$) in a sense that will be formalised in the next definition. 
The definition refers to ``rigid counters''. 
A counter \(j\) is \emph{rigid} in $\Vv_i$ if its value is uniquely determined by the control state. 
This means that all cycles in $\Vv_i$ have zero effect on rigid counters.
As one might expect, such counters will be removed the moment they are identified, together with all states where this counter has a negative value.
So for the sake of understanding~\cref{def:pumpable}, one should think that no rigid counters are present.
 
\begin{definition} \label{def:pumpable}
    Let \((\vect{s},\Vv,\vect{t})\) be a query. 
    Let $\qiin$ and $\qiout$ be the states in which the SCC \(\Vv_i\) is entered and left respectively.

    We call \(\Vv\) \emph{forward-pumpable} if, for every SCC \(\Vv_i\), there exists a cycle \(\up_i\) at \(\qiin\) that is enabled at \(test_{e_i}+\effect_{\NN}(e_i)\) and has (strictly) positive effect on all \(\NN\)-counters \(j\) such that (i) \(test_{e_i}[j]\neq \omega\) and (ii) $j$ is not rigid in \(\Vv_i\).

    Similarly, we call \(\Vv\) \emph{backward-pumpable} if, for every SCC \(\Vv_i\), there exists a cycle \(\dwn_i\) at \(\qiout\) that is backwards-enabled at \(test_{e_{i+1}}\) and has a (strictly) negative effect on all $\NN$-counters $j$ such that (i) \(test_{e_{i+1}}[j] \neq \omega\) and (ii) $j$ is not rigid in \(\Vv_i\). 
\end{definition} 

Finally, we are ready to define perfectness. 

\begin{definition} \label{def:perfectnes}
      A given reachability query \((\vect{s}, \Vv, \vect{t})\) is \emph{perfect} if $\ILP(\vect{s}, \Vv, \vect{t})$ is perfect, \((\vect{s}, \Vv, \vect{t})\) is forward-pumpable and backward-pumpable, and for all SCCs \(Q_i\), all counters \(j\) rigid in \(Q_i\) and all states \(q \in Q_i\), the value of counter \(j\) at state \(q\) is \(\geq 0\).
\end{definition}

I.e.\ we require all the properties mentioned before, plus we require that rigid counters do not go negative, which, by definition of rigidity, is a purely syntactic requirement.

The proof that perfectness implies a run from $\vec{s}$ to $\vec{t}$ in $\Vv$ (i.e.\ the proof of~\cref{LemmaLabelPerfectnessImpliesRun}), is essentially the same as for the original KLMST algorithm for reachability in VASS (without additional $\ZZ$-counters).
This is because the additional \(\ZZ\)-counters are automatically handled by the ILP and cannot otherwise disturb the validity of the run.
The proof of~\cref{LemmaLabelPerfectnessImpliesRun} can be found in~\cref{app:klm-upper-bound}.

\subsection{Deciding Perfectness and the Decomposition}
\label{SubsectionDecidingPerfectnessAndDecomposition}

Now that we have set up the definitions, the next step is to decide whether a given instance of reachability is perfect and, if not, we must decompose into a collection of lower rank instances.
These procedures for checking perfectness and decomposing form the heart of KLMST. 
Thankfully most of the checks and decompositions for reachability in \zvass{}{} are the same as for VASS. 
The decrease of \(\rank_{\NN}(V)\), which was used as the rank also in \cite{LerouxS19, FuYZ24, GuttenbergCL25}, can simply be copied to almost all decompositions.

We now provide an overview of the various decomposition steps, starting from perfectness of the ILP, which has three parts.

\begin{description}
    \item[Condition (i).] 
    Checking the existence of a solution to an ILP is in \class{NP}; for example, see~\cite{LinearProgramming}.
    Clearly if the ILP has no solution, then there does not exist a run from $\vec{s}$ to $\vec{t}$ in \(\Vv\).
    Since reachability does not hold, the empty decomposition is returned.
    
    \item[Condition (ii).] 
    This condition is changed substantially compared to VASS, and we will detail the decomposition for this condition shortly. 
    
    \item[Condition (iii).] 
    We can check in \class{NP} whether a variable is bounded, and for bounded variables we can compute, in exponential time, the set of possible values.
    This is because every minimal solution to the ILP has polynomial size~\cite{LinearProgramming} when encoded in binary.
    For every one of the possibly exponentially many different values \(b\) of a variable \(y_{i,j}\), we create a separate generalised VASS \(\Vv_b\) in which we set \(test_{e_i}[j] = b\) at the exit of SCC \(\Vv_i\). 
    Symmetrically, we create a generalised VASS for every possible value of \(x_{i,j}\).
    
    \item [Rigid counters.]
    We can check rigidity of counter \(j\) in SCC \(Q_i\) by checking for every simple cycle in \(Q_i\) that its effect on \(j\) is \(0\). If a rigid counter would take a negative value at some state \(q\), then we simply delete the control state \(q\) (as performed in~\cite{LerouxS19}).
    This decomposition step can clearly only decrease the size of the current instance. 
    
    \item[Pumpability.] 
    This property only concerns \(\NN\)-counters.
    The condition only examines the effect on \(\NN\)-counters, and being enabled does not depend on the values of the \(\ZZ\)-counters.
    Thus, we can equivalently check this condition in the projected VASS \(\pi_{\NN}(\Vv)\), where \(\pi_{\NN}\) removes the \(\ZZ\)-counters to obtain a $d$-VASS.
    Without integer counters, we can therefore apply~{\cite[Lemma 4.13]{LerouxS19}} without modification: on every run, some counter \(j\) is bounded by \(\unarysize{\Vv}^{O(d^d)}\), and the decomposition stores the $j$-th counter in the control state.
\end{description}

We will now detail the decomposition for condition (ii). 
Verifying the condition is straightforward: it suffices to check whether there exists a homogeneous solution to the ILP that assigns a positive value.
This can be done in~\class{NP}. 
Importantly, an ILP has a homogeneous solution assigning \(\vect{h}[x_e]>0\) if and only if the corresponding variable \(x_e\) is unbounded. 
Thus, if condition (ii) does not hold, then the variable \(x_e\) is bounded. 
The decomposition here accordingly removes the transition \(e\) from an SCC with the following strategy.
To accommodate solutions of the ILP with $x_e = b$, we create \(b+1\) copies of \((Q_i, E_i \setminus \{e\})\) and connect them sequentially using the $b$ copies of \(e\). 

There is an important difference in our version.
Originally in~\cite{Kosaraju82}, only one transition was ever removed in a decomposition step.
This, however, leads to slower-decreasing rank and the result is that the entire algorithm runs in greater than Ackermannian time. 
This was subsequently improved by Leroux and Schmitz to simultaneously remove all bounded transitions \(e \in E_i\) in some SCC $\Vv_i$, leading to an Ackermannian upper bound~\cite{LerouxS19}.
Again, our main improvement is here: we remove every bounded transition \(e \in E_i\) in \emph{every} SCC \(\Vv_i\) simultaneously; we do not limit ourselves to decomposing a single SCC in a particular step.

Formally, let \(\Vv = (\Vv_0 e_1 \Vv_1 \dots e_{\NumberOfSCCs} \Vv_{\NumberOfSCCs})\) be a generalised VASS. 
We will compute, for every \(i \in \{0\} \cup [\NumberOfSCCs]\), a decomposition \(Decomp(\Vv_i)\), which is a set of generalised VASS, and return \[Decomp(\Vv)=\{\Vv_0' e_1 \Vv_1' \dots e_{\NumberOfSCCs} \Vv_{\NumberOfSCCs}' : \Vv_i' \in Decomp(\Vv_i)\}.\]
In other words, for every $i$, we choose a VASS from \(Decomp(\Vv_i)\) and connect them with transitions $e_i$ as in $\Vv$.

We define \(Decomp(\Vv_i)\) in the same way as in \cite{LerouxS19} by removing \emph{all} bounded transitions from this SCC and choosing an order to apply them in.
Let \(\Vv_i = (Q_i, E_i)\), and let \(E_i' \subseteq E_i\) be the subset of bounded transitions. 
For every minimal ILP solution \(\vect{s}\), and every word \(w \in (E_i')^{\ast}\) with Parikh image \(\psi(w)=\vect{s}[E_i']\), we create a VASS \(\Vv_{i, \vect{s}, w} \in Decomp(\Vv_i)\) as a repetition of \(\Vv_{\text{local}}=(Q_i, E_i \setminus E_i')\). 
We set \(\Vv_{i, \vect{s}, w} \coloneqq \Vv_{\text{local}} \, w_1 \, \Vv_{\text{local}} \, \cdots \, w_\ell \, \Vv_{\text{local}}\),
where \(w = w_1 \cdots w_\ell\).
It is easy to see that this decomposition preserves the reachability relation, it remains to show that the rank decreases. In the following lemma, we say that a \zvass{}{} $\Vv$ has \emph{dimension} \(k\) if \(k=\max_{q \in Q} \dim(\pi_{\NN}(\vectSet{V}(\Vv, q)))\).

\begin{lemma} \label{LemmaBoundedCounterDecreasesRank}
    The decomposition for condition (ii) decreases the rank.
    Moreover, if only the least significant coordinate decreases, then all \(\Vv_i\) of dimension at least 2 remain unchanged, and every \(1\)-dimensional \(\Vv_i\) either remains unchanged, or is split into a \(0\)-dimensional \(\Vv_i'\).
\end{lemma}

\begin{proof}
    Let \(\Vv' \in Decomp(\Vv)\). 
    We will show that \(\rank(\Vv')<\rank(\Vv)\).
    
    \proofsubparagraph*{Case 1:} $\dim(\vectSet{\Vv}_{\ZZ}(\Vv'))<\dim(\vectSet{V}_{\ZZ}(V))$. 
    Since this dimension is the most important component of the rank, the rank has decreased.
    
    \proofsubparagraph*{Case 2:} $\dim(\vectSet{V}_{\ZZ}(\Vv'))=\dim(\vectSet{V}_{\ZZ}(\Vv))$. 
    The new VASS \(\Vv'\) may only have fewer cycle effects than \(\Vv\), so \ \(\vectSet{V}_{\ZZ}(\Vv') \subseteq \vectSet{V}_{\ZZ}(\Vv)\). 
    Since $\vectSet{V}_{\ZZ}(\Vv)$ and $\vectSet{V}_{\ZZ}(\Vv')$ are vector spaces of the same dimension, we deduce that \(\vectSet{V}_{\ZZ}(\Vv')= \vectSet{V}_{\ZZ}(\Vv)\).
    
    Let \(\Vv_i' \in Decomp(\Vv_i)\) be the generalised VASS obtained from \(\Vv_i\) and let \(\vectSet{V}(\Vv_i)\) be the cycle space of \(\Vv_i\) (recall that the cycle space only depends on the SCC). 
    
    \begin{claim}\label{clm:all-transitions-used}
        For every SCC \(\Vv_i\), if \(\dim(\pi_{\NN}(\vectSet{V}(\Vv_i'))) = \dim(\pi_{\NN}(\vectSet{V}(\Vv_i)))\), then all transitions of \(\Vv_i\) were contained in a homogeneous solution \(\vect{h}\). 
    \end{claim}
    Otherwise put, if \(\dim(\pi_{\NN}(\vectSet{V}(\Vv_i'))) = \dim(\pi_{\NN}(\vectSet{V}(\Vv_i)))\), then  \(Decomp(\Vv_i)=\{\Vv_i\}\) was not decomposed.
    
    First we conclude the proof of~\cref{LemmaBoundedCounterDecreasesRank} by explaining how~\cref{clm:all-transitions-used} implies a decrease to the rank.
    Since some transition was bounded, at least one SCC was decomposed. 
    Let~\(d'\) be the maximal dimension such that an SCC of dimension \(d'\) was decomposed. 
    By the above claim, whenever an SCC is decomposed, all resulting states have lower dimension, hence \(\rank_{\NN}(\Vv')[d']<\rank_{\NN}(\Vv)[d']\) has decreased, while \(\rank_{\NN}(\Vv')[i]=\rank_{\NN}(\Vv)[i]\) for all \(i >d'\). 
    This means that the rank decreased. 
    Similarly, the ``moreover'' of~\cref{LemmaBoundedCounterDecreasesRank} follows from~\cref{clm:all-transitions-used}.
\end{proof}
    
\begin{proof}[Proof of~\cref{clm:all-transitions-used}]
    Since \(\Vv_i\) is an SCC, there is a cycle \(\rho\) in \(\Vv_i\) which uses all transitions \(e\) in \(\Vv_i\). 
    As \(\dim(\pi_{\NN}(\vectSet{V}(\Vv_i'))) = \dim(\pi_{\NN}(\vectSet{V}(\Vv_i)))\) has not decreased, and \(\pi_{\NN}(\vectSet{V}(\Vv_i')) \subseteq \pi_{\NN}(\vectSet{V}(\Vv_i))\), we observe that \(\pi_{\NN}(\vectSet{V}(\Vv_i'))=\pi_{\NN}(\vectSet{V}(\Vv_i))\). 
    Therefore, there exist cycles \(\rho_\ell\) in \(\Vv_i'\) and coefficients \(\lambda_\ell \in \QQ\) such that \(\pi_{\NN}(\effect(\rho))=\pi_{\NN}(\sum_{\ell=1}^r \lambda_\ell \cdot \effect(\rho_\ell))\). 
    Viewing the cycles \(\rho_\ell\) as cycles in the SCC \(\Vv_i\), we obtain that \(\effect(\rho)-\sum_{\ell=1}^r \lambda_\ell \cdot \effect(\rho_\ell) \in \vectSet{V}(\Vv_i) \cap \{0^d\} \times \ZZ^k \subseteq \vectSet{V}_{\ZZ}(\Vv)\). 
    Since we are in case 2 of~\cref{LemmaBoundedCounterDecreasesRank}, we have \(\vectSet{V}_{\ZZ}(\Vv) = \vectSet{V}_{\ZZ}(\Vv')\) and obtain \(\effect(\rho)-\sum_{\ell=1}^r \lambda_\ell \cdot \effect(\rho_\ell) \in \vectSet{V}_{\ZZ}(\Vv')\). 
    In other words, there exist cycles \(\rho_{m}'\) in \(\Vv'\) and coefficients \(\lambda_m'\) such that \(\effect(\rho)-\sum_{\ell=1}^r \lambda_\ell \cdot \effect(\rho_\ell)-\sum_{m=1}^{r'} \lambda_m' \cdot \effect(\rho_m')=\vect{0}\).

    We observe that \(sol=\psi(\rho)-\sum_{\ell=1}^r \psi(\rho_\ell)-\sum_{m=1}^{r'} \lambda_m' \cdot \psi(\rho_m')\) fulfills all three sets of equations of the ILP, where \(\psi(\rho)\) is the Parikh image of \(\rho\).
    Since \(sol\) consists only of cycles, it trivially fulfills the Kirchhoff equations, and by construction it both has effect \(0\) on the \(\NN\)-counters in every SCC and has the correct global effect on the $\ZZ$-counters.
    However, the variables may have rational values (instead of values in \(\NN\)). 
    Now we use the fact that all bounded variables have been removed.
    To ensure positive values, take any full support homogeneous solution \(\vect{h}\) of \(\Vv'\) and add it sufficiently many, $c$, times to \(sol\). 
    The resulting vector \(sol' \coloneqq sol+c \cdot \vect{h}\) is still a solution, and now all variables \(x_e\) for edges \(e\) in \(\Vv'\) satisfy \(sol'(x_e)>0\). 
    Observe that \(sol'(x_e)>0\) also for transitions \(e\) of \(\Vv\) which are not in \(\Vv'\): they appeared a positive number of times in \(\rho\), and do not appear in any \(\rho_\ell\) or \(\rho_m'\), since these are cycles in \(\Vv'\). 
    By multiplying the positive rational solution $sol'$ by a common denominator, we obtain a positive integer solution.
\end{proof}

\subsection{The Complexity of the Algorithm} \label{SectionKLMSTComplexity}

There are existing tools to analyse the complexity of reachability in VASS. 
One of the most notable techniques is called ``controlled bad sequence analysis''~\cite{FigueiraFSS11}.
While making some modification, this analysis technique will form the backbone of our complexity analysis.
We will analyse the sequence \((\rank(\Vv_0), \rank(\Vv_1), \dots)\), where \(\Vv_0\) is the initial query, and \(\Vv_{i+1}\) is one of the VASS outputted by~\cref{LemmaLabelKLMSTDecomposition} applied to \(\Vv_i\).

Since the rank is a vector in \(\NN^{d+1}\), and a single decomposition step takes elementary time, one is able to obtain an \(\Ff_{d+3}\) upper bound with the standard bad sequence analysis.
However, we can improve the analysis using the fact that if only the least important component decreases, then there is only a polynomial blowup (which we will establish for our decomposition in a moment).

\begin{definition}
     Let \(m \in \NN\) and let \(f_0,f_1 \colon \NN \to \NN\) be monotone functions with \(f_0(x), f_1(x)\geq x\), for all \(x \in \NN\). 
     Let \(\vect{x}_0, \vect{x}_1, \dots\) be a sequence in \(\NN^d\) that is lexicographically decreasing. 
     Let \(b_0, b_1, \dots\) be the induced sequence of bits where $b_k = 1$ if and only if \(\vect{x}_k\) is greater than \(\vect{x}_{k+1}\) on any of the first \(d-1\) coordinates.
     The sequence \(\vect{x}_0, \vect{x}_1, \dots\) is called \emph{\((f_1,f_0,m)\)-controlled} if \(|\vect{x}_k| \leq f_{b_k}(f_{b_{k-1}}(\dots(f_{b_0}(m))\dots))\), for all \(k \in \NN\).
\end{definition}

This definition exactly describes the situation of Lemma \ref{LemmaLabelKLMSTDecomposition} that, in general, we have an elementary blowup, but if only the least important coordinate decreases, then there is only a polynomial blowup. 
Usually, one considers the case \(f_1=f_0\) and hence the following proposition on \((f,f,m)\)-controlled sequences is standard.

\begin{proposition}[c.f.~\cite{FigueiraFSS11, Schmitz16}] \label{PropositionBadSequenceLength}
    Let \(d \geq 1\) and let \(f\) be a monotone elementary function fulfilling \(f(x)\geq x\), for all \(x\).
    The function \(L_f\) mapping \(m \in \NN\) to the length of the longest decreasing sequence in \(\NN^d\) which is \((f,f,m)\)-controlled fulfills \(L_f \in \Fff_{d+2}\).
\end{proposition}

Using~\cref{PropositionBadSequenceLength}, we can state our upgraded version, that we subsequently use to prove~\cref{thm:klm-d-vass}.
The proof of~\cref{PropositionBadSequenceLengthDoubleControl} can be found in~\cref{app:klm-upper-bound}

\begin{restatable}{proposition}{PropositionBadSequenceLengthBetter} \label{PropositionBadSequenceLengthDoubleControl}
    Let \(d \geq 2\), let \(f_1\) be an elementary function, and let \(f_0\) be a polynomial function. 
    Then the function \(L_{f_1, f_0}\) mapping \(m \in \NN\) to the length of the longest decreasing sequence in \(\NN^d\) which is \((f_1, f_0, m)\)-controlled fulfills \(L_{f_1, f_0} \in \Fff_{d+1}\).
\end{restatable}

Armed with Proposition \ref{PropositionBadSequenceLengthDoubleControl}, we can prove Theorem \ref{thm:klm-d-vass}.

\begin{proof}[Proof of Theorem \ref{thm:klm-d-vass}]
    Let \(\Vv\) be a \zvass{d}{}, let $\vect{s}$ be the initial configuration, and let  $\vect{t}$ be the target configuration.
    Just as sketched, the algorithm repeatedly applies~\cref{LemmaLabelKLMSTDecomposition} to decompose \((\vect{s}_0, \Vv_0, \vect{t}_0)\) into smaller \((\vect{s}_i, \Vv_i,\vect{t}_i)\). 
    Let \(f_1\) be the running time of~\cref{LemmaLabelKLMSTDecomposition}, and let \(f_0\) be the polynomial bound promised if only the least important rank decreases. 
    Then, the sequence of ranks $(\rank(\Vv_0)$, $\rank(\Vv_1), \dots)$, in $\NN^{d+1}$, is an \((f_1, f_0, |\vect{s}|+|\vect{t}|+|V_0|)\)-controlled bad sequence. 
    By~\cref{PropositionBadSequenceLengthDoubleControl}, we decompose at most \(L_{f_1, f_0}(\size(\Vv_0))\) many times, which is in \(\Ff_{d+2}\). 
    Accordingly, the generalised VASS \(\Vv_i\) considered by the algorithm all have size at most \(\Ff_{d+2}\).
    Accordingly,~\cref{LemmaLabelKLMSTDecomposition} is only called on inputs of size at most \(\Ff_{d+2}\). 
    In total, all of these calls can be computed in \(\Ff_{d+2}\).
\end{proof}

Critically, in order to apply the above improved complexity analysis, we must check that if only the least important coordinate of the rank decreases, then the size only increases polynomially.
Strictly speaking, with the decomposition provided so far, this is not the case. We explain here the minor change which has to be performed to obtain \(\Ff_{d+2}\) instead of \(\Ff_{d+3}\).

In our setting ``only the lowest coordinate of the rank decreases'' means: a one-dimensional SCC $\Vv_i$ was decomposed into zero-dimensional ones (which are not counted by the rank).
The blowup can nevertheless be exponential in the pumpability setting: the blowup is only polynomial in the unary size and we store generalised VASS in binary. 
Also, condition (ii) creates \(b\) many copies of an SCC where \(b \in \NN\) is obtained from an ILP solution; $b$ can have exponential magnitude. 
In either case, the blowup can be reduced to a polynomial using the following lemma.
Let \(Rel(\Vv)\) denote the reachability relation of \(\Vv\).

\begin{restatable}{lemma}{LemmaReplaceByPolySize} \label{LemmaReplaceByPolySize}
    There is an elementary-time algorithm that, given a \(0\)-dimensional \zvass{}{} \(\Vv\) whose unary encoding is exponentially large (but uses a linear number of \(\ZZ\)-counters \(k\)), produces finitely many \(0\)-dimensional \zvass{}{} \(\Vv_j\) such that every \(\Vv_j\) has polynomial size (i.e.\ smaller than the size of \(\Vv\)) and \(Rel(\Vv) = \bigcup_{j} Rel(\Vv_j)\).
\end{restatable}

The main idea is that zero-dimensionality implies that the \(\NN\)-counters only depend on the control-state.
This means that \(\Vv\) is essentially an integer VASS.
Accordingly, we can obtain (with Carath\'{e}odory arguments) a semilinear description which is a potentially large union of linear sets using few periods.
The full proof of~\cref{LemmaReplaceByPolySize} is similar to the proof of~\cref{lem:linear-form-runs} (in fact it is simpler) and can be found in~\cref{app:klm-upper-bound}.

\section{Conclusion}
\label{sec:conclusion}

Before presenting some future research direction, we will briefly highlight some of the key takeaways, beyond our main results (Theorems~\ref{thm:1-zvass-np},~\ref{thm:unary-2-zvass-psapce-hard},~\ref{thm:unary-3-zvass-tower-hard}, and~\ref{thm:klm-d-vass}).

The first is that our approach to proving reachability in binary \zvass{1}{} is in \class{NP}, which uses Carath\'{e}odory-style bounds and cycle replacement, appears to both be relatively simple and rather robust.
This is witnessed by the fact the fact that our proof can be used as an alternative, arguably simpler, proof that reachability in binary 1-VASS is in \class{NP}~\cite{HaaseKOW09}.
See~\cref{rem:1-VASS-reachability} on~\cpageref{rem:1-VASS-reachability} for a more detailed discussion.

Another takeaway is that~\cref{lem:big-numbers-sufficient} is close in spirit to the following extremely strong kind of statement: if long runs can be enforced by a system, then reachability is computationally hard. 
Roughly speaking,~\cref{lem:big-numbers-sufficient} states that if some \zvass{d}{k} can compute $(2^{f(n)}, 7^{2^{f(n)}})$, where $f$ is a space-constructible function, then reachability in \zvass{d}{k} is \class{Space($f(n)$)}-hard.
The proof of this statement involves a novel technique for simulating counter machines that only uses weak multiplication and carefully encodes the length of the run into the counter values. 
We use this, with $f$ as a linear function, to prove that reachability in unary \hbox{\zvass{2}{}} is \class{PSpace}-hard (\cref{thm:unary-2-zvass-psapce-hard}) by providing, construction that can compute double exponential values (\cref{lem:twoexptriple}). 
We discuss a closely related conjectured generalisation of this lower bound technique shortly (\cref{con:long-runs-iff-hardness}).

The final takeaway that we will reflect on is a development made in~\cref{sec:klm-upper-bound}.
For $d$-VASS, it is known that reachability is in $\Ff_d$~\cite{FuYZ24} and is $\Ff_{(d-3)/2}$-hard~\cite{CzerwinskiJLLO23}.
Given that the total dimension of a \zvass{d}{k} is $d+k$, one would likely therefore expect reachability in \zvass{d}{k} to be in $\Ff_{d+k}$ and to be $\Ff_{\Omega(d+k)}$-hard.
However, the latter is not true: reachability in \zvass{d}{k} is in $\Ff_{d+2}$ (\cref{thm:klm-d-vass}).
To obtain this upper bound, we make several modifications to the now classic KLMST algorithm for reachability in VASS. 
Most notably is the $\ZZ$-vector space $\vectSet{V}_{\ZZ}(\Vv)$ of the possible effects that cycles can have on integer counters (developed in~\cref{SubsectionDefinitionRank}).
We use the (dimension of the) $\ZZ$-vector space to define a more succinct rank for \zvass{d}{k} compared to $(d+k)$-dimensional rank one might obtain by porting from the $d$-dimensional rank for $d$-VASS~\cite{LerouxS19}.

\subsubsection*{Future Directions}

In~\cref{sec:klm-upper-bound}, we have shown that the reachability problem for \zvass{d}{} is in $\Ff_{d+2}$. 
One extra difficulty compared to VASS was that geometrically two-dimensional VASS have a semilinear reachability relation, whereas geometrically \zvass{2}{} need not have a semilinear reachability relation.
In fact, it is already true that geometrically one-dimensional \zvass{}{} may not have a semilinear reachability relation.
Namely, consider the following is a geometrically one-dimensional \zvass{2}{2} that has a non-semilinear reachability relation.
In the drawing, the first two counters are $\NN$-counters and the last two are $\ZZ$-counters. 

\begin{center}
\begin{tikzpicture}[]
    \tikzset{every edge/.append style={font=\large}}

    \node[draw=black, circle, double, minimum size=0.75cm] (A) at (1,0) {\(q\)};
    \node[draw=black, circle, minimum size=0.75cm] (B) at (3,0) {\(p\)};
    \node[white!100] (C) at (0.2,-0.3) {};
        
    \path[->, thick, out=30, in=150] (A) edge[] node[above] {\((0,0,1,0)\)} (B);
    \path[->, thick, out=210, in=330] (B) edge[] node[right] {} (A);
    \path[->, thick, out=150, in=90, looseness=4] (A) edge[] node[left] {\((1,-1,0,1)\)} (A);
    \path[->, thick, out=-15, in=30, looseness=4] (B) edge[] node[right] {\((-1,1,0,0)\)} (B);
    \path[->, thick] (C) edge[] (A);
\end{tikzpicture}
\end{center}

The sum of the two \(\NN\)-counters is constant, so this \zvass{}{} indeed has geometric dimension \(1\). 
If this sum starts at value \(x_{12}\), and the first \(\ZZ\)-counter ends at value \(x_3\), i.e.\ we perform the outer loop exactly \(x_3\) many times, then the value $x_4$ of the second $\ZZ$-counter is bounded above by $x_4 \leq x_{12} \cdot x_3$. 
The second $\ZZ$ counter tracks how often we use the first inner loop, and the worst-case is that we fully iterate the inner loops for every iteration of the outer loop.
The inequality $x_4 \leq x_{12} \cdot x_3$ defines a non-semilinear reachability relation.
Observe however that in this case, the reachability set is semilinear for \emph{fixed} initial value of \(x_{12}\).
Indeed, it is possible to show that reachability \emph{sets} of \zvass{}{} with geometric dimension \(1\) are semilinear.

For every $d$, there exists a \zvass{d}{} $
\Vv_d$ that has a finite reachability sets of size $F_{d+1}$ with respect to the size of $\Vv_d$
This can be obtained by making a simple modification the $\Ff_d$ construction first presented in~\cite{HowellRHY86} (in particular, see Theorem 2.1, or for an alternative presentation, see~{\cite[Proposition 8]{CzerwinskiO21}}).
One needs to eliminate counters $x_k$ and $x_{k+1}$ at the cost of introducing some $\ZZ$-counters.
With this in mind, it is true that KLMST algorithms cannot have lower than $\Ff_{d+1}$ complexity; such algorithms enumerate the entire reachability set if the given VASS is bounded.
We do, however, conjecture that such an upper bound can be obtained which would improve over our $\Ff_{d+2}$ upper bound (\cref{thm:klm-d-vass}).

\begin{conjecture}
    Reachability in \zvass{d}{} is in $\Ff_{d+1}$.
\end{conjecture}

Recently it has been shown that reachability in $3$-VASS is in \elementary, precisely in \twoexpspace~\cite{CzerwinskiJLO25}, while finite reachability sets for $3$-VASS can be of size up to \tower. 
Therefore, one can conjecture that for low-dimension \zvass{d}{}, the complexity of the reachability problem is also below the size of the reachability set. 
Elementary complexity is excluded for $d = 3$, as we have shown that reachability in \zvass{3}{} is \class{Tower}-hard in~\cref{sec:3-VASS+Z}. 
However, we conjecture that elementary upper bounds can still be obtained in \zvass{2}{}.

\begin{conjecture}
    Reachability in binary \zvass{2}{} is in \elementary. 
\end{conjecture}

An interesting phenomenon which we investigated in \Cref{sec:2-VASS+Z}, is the connection between hardness of reachability and the length of the shortest runs for \zvass{}{}.
Although we have only used this to obtain a lower bound for reachability in \zvass{2}{}, it appears to be true that the length of the shortest runs between two configurations is closely related to specific hardness results more generally.
The following conjecture captures this perceived relationship.

\begin{conjecture}\label{con:long-runs-iff-hardness}
    For $d \geq 2$, reachability in \zvass{d}{} is \class{Space($f(n))$}-complete for a space-constructible function $f$ if and only if the shortest runs between the initial and target configuration are of length $2^{f(\Theta(n))}$.
\end{conjecture}

We believe, that~\cref{con:long-runs-iff-hardness} can be useful for determining the exact complexity of reachability in \zvass{}{} when parametrised by the number of $\NN$ counters.

\bibliographystyle{plainurl}
\bibliography{references}

\appendix

\section{Additional Lower Bounds}
\label{app:additional-results}

The fact that reachability in integer VASS is \class{NP}-complete, regardless as to whether the transition effects are encoded in unary or binary, is a folklore result.
However, to the best of our knowledge, the \class{NP} lower bound for reachability in unary integer VASS has not been published.

\begin{proposition}\label{pro:unary-zvass}
    Reachability in unary integer VASS is \class{NP}-hard.
\end{proposition}

\begin{proof}[Proof of~\cref{pro:unary-zvass}]
    We will reduce from the the subset sum problem.
    Suppose $(x_1, \ldots, x_n, t) \in \NN^{n+1}$ is an arbitrary input to the subset sum problem which asks whether there exists $I \sset [n]$ such that $\sum_{i \in I} x_i = t$.
    Let $d$ be the least integer such that $\sum_{i=1}^n x_i < 2^d$.
    Note that $d \leq \log(\max\set{x_1, \ldots, x_n}) + 1$ is polynomial with respect to the size of the given instance.
    
    We will construct, in polynomial time, a $d$-dimensional integer VASS $\Vv = (Q, T)$
    Roughly speaking, we will use the $d$ counters to store, in binary encoding, the current sum of the selected elements from $\set{x_1, \ldots, x_n}$.
    Let $z_1, \ldots, z_d$ denote the $d$ counters in $\Vv$.
    We will have $n+1$ states $Q = \set{q_0, q_1, \ldots, q_n}$.
    We will have $2n+(d-1)$ transitions in $\Vv$.
    For every $i \in [n]$, there are two transitions: $(q_{i-1}, \vec{v}_i, q_i)$ and $(q_{i-1}, \vec{0}, q_i)$.
    Here, $\vec{v}_i = (b_0, b_1, \ldots, b_{j-1}) \in \set{0,1}^d$ corresponds to the binary encoding of $x_i$, i.e.\ $x_i = \sum_{j=0}^{d-1} b_j\cdot 2^j$.
    We note that taking the transition $(q_{i-1}, \vec{v}_i, q_i)$ corresponds to ``selecting $x_i$'' and taking the transition $(q_{i-1}, \vec{0}, q_i)$ corresponds to ``not selected $x_i$''.
    Additionally, for every $j \in [d-1]$, we will have the transition $(q_n, \vec{u}_j, q_n)$, where $\vec{u}_j$ is the vector such that $\vec{u}_j[j] = -2$, $\vec{u}_j[j+1] = 1$, and for every $k \in [j]\setminus\set{j,j+1}$, $\vec{u}_j[k] = 0$.
    We note that these transitions allow us to shift values from lower order bits to higher order bits.
    
    We choose $\Config{q_0}{0}$ to be the initial configuration.
    Let $\vec{t} \in \set{0,1}^d$ be the vector such that $t = \sum_{j=1}^{d}\vec{t}[j]\cdot2^{j-1}$.
    We choose $\Config{q_n}{t}$ be the target configuration.
    It remains to argue that $\Config{q_0}{0} \xrightarrow{*}_\Vv \Config{q_n}{t}$ if and only if there exists a subset $I \sset [n]$ such that $\sum_{i \in I} x_i = t$. 
    
    First, if $\Config{q_0}{0} \xrightarrow{*}_\Vv \Config{q_n}{t}$, then we can use the transitions selected in the run to tell us which elements from $\set{x_1, \ldots, x_n}$ to select.
    Precisely, if $(q_{i-1}, \vec{v}_i, q_i)$ is taken, then we add $i$ to $I$; otherwise if $(q_{i-1}, \vec{0}, q_i)$ is taken, then we \emph{do not} add $i$ to $I$.
    
    Second, if there exists $I \sset [n]$ such that $\sum_{i \in I} v_i = t$, then we can construct a run from $\Config{q_0}{0}$ to $\Config{q_n}{t}$ by selecting $(q_{i-1}, \vec{v}_i, q_i)$ if $i \in I$, otherwise select $(q_{i-1}, \vec{0}, q_i)$ if $i \notin I$.
    Importantly, by definition of $\vec{v}_i$, we know that after taking the first $n$ transitions, we will arrive in a configuration $\Config{q_n}{s}$ such that 
    \begin{equation*}
        \sum_{j=1}^d \vec{s}[j]\cdot 2^{j-1} = \sum_{i \in I} \sum_{j=1}^d \vec{v}_i[j] \cdot 2^{j-1} = \sum_{i \in I} \sum_{j=0}^{d-1} b_j\cdot2^j = \sum_{i \in I} x_i = t.
    \end{equation*}
    However, it is not a given that $\vec{s} = \vec{t}$.
    For this, we may need to use the transitions $(q_n, \vec{u}_j, q_n)$ to shift the bits in $\vec{s}$ to reach $\vec{t}$.
    For example, it could be true that $t = 6$ and $x_1 = x_2 = 2$ and $x_3 = x_4 = 1$, so $\vec{t} = (1,1,0)$ yet after the first transitions, we obtain counter values $\vec{s} = (0,2,2)$.
    After (some) of these final $(q_n, \vec{u}_j, q_n)$ transitions are taken, we arrive at $\Config{q_n}{t}$.
\end{proof}

\begin{corollary}\label{cor:unary-2-vass+1z}
    Reachability in unary \zvass{2}{1} is \class{NP}-hard.
\end{corollary}

\begin{proof}
    This corollary follows from the proof ideas used to show that reachability in unary 3-VASS is \class{NP}-hard.
    To be precise, in the proof of~{\cite[Theorem 1]{ChistikovCMOSW24}}, the authors first construct a unary 2-VASS with a small number of zero tests for which reachability coincides with the satisfiability of a formula in 3-CNF form (in particular, see~{\cite[Lemma 6 and Claim 7]{ChistikovCMOSW24}}).
    Then, the authors use the ``controlling counter technique'' to simulate the (small number of) zero tests with an additional counter (in particular, see~{\cite[Lemma 5]{ChistikovCMOSW24}}).

    We claim that this additional controlling counter can be a $\ZZ$-counter.
    Accordingly, this implies that reachability in unary \zvass{2}{1} is \class{NP}-hard. 
    Here, we will give a high-level explanation of why this holds.
    For further details, one can follow the proof of~{\cite[Lemma 5]{ChistikovCMOSW24}}.

    For the sake of demonstrating the controlling counter technique, suppose there is a run from $q_0(x_0)$ to $q_4(x_4)$ in a 1-VASS with zero tests in which the counter is zero-tested three times. 
    Consider splitting the run into four segments (at each of the three zero tests):
    \begin{equation*}
        q_0(x_0) 
        \xrightarrow{\pi_1} q_1(x_1) 
        \xrightarrow{\pi_2} q_2(x_2) 
        \xrightarrow{\pi_3} q_3(x_3) 
        \xrightarrow{\pi_4} q_4(x_4).
    \end{equation*}
    To be clear, $q_0(x_0)$ is the initial configuration, and $q_1(x_1)$ is the configuration reached immediately after the first zero-testing transition is taken.
    This means that $x_1 = 0$.
    Similarly, $q_2(x_2)$ and $q_3(x_3)$ are the configurations reached immediately after the second and third zero-testing transitions are taken, respectively.
    Again, this means that $x_2 = x_3 = 0$. 
    Importantly, the one counter of this 1-VASS is an $\NN$-counter; this means that $x_0 + \oneeffect(\pi_1) \geq 0$, $x_0 + \oneeffect(\pi_1) + \oneeffect(\pi_2) \geq 0$, and $x_0 + \oneeffect(\pi_1) + \oneeffect(\pi_2) + \oneeffect(\pi_3) \geq 0$.
    
    The controlling counter is an additional counter that throughout $\pi_1$, $\pi_2$, $\pi_3$, and $\pi_4$ receives three, two, one, and zero times the effect that the original counter receives, respectively.
    The controlling counter starts with three times the original counter's value.
    Now, let's examine the effect of the entire run on the controlling counter.
    To begin with the controlling counter has value $3x_0$, then it receives $3\cdot\oneeffect(\pi_1)$ throughout $\pi_1$, then $2\cdot\oneeffect(\pi_2)$ throughout $\pi_2$, then $\oneeffect(\pi_3)$ throughout $\pi_3$, and finally $0$ throughout $\pi_4$.
    All together, we therefore know that the final value of the controlling counter is then $3x_0 + 3\cdot\oneeffect(\pi_1) + 2\cdot\oneeffect(\pi_2) + \oneeffect(\pi_3)$.
    Therefore, we aim to end in a configuration in which the controlling counter is zero, we can assert: $3x_0 + 3\cdot\oneeffect(\pi_1) + 2\cdot\oneeffect(\pi_2) + \oneeffect(\pi_3) = 0$.
    Recall that $x_0 + \oneeffect(\pi_1) \geq 0$, $x_0 + \oneeffect(\pi_1) + \oneeffect(\pi_2) \geq 0$, and $x_0 + \oneeffect(\pi_1) + \oneeffect(\pi_2) + \oneeffect(\pi_3) \geq 0$.
    Together, since these three inequalities sum to $3x_0 + 3\cdot\oneeffect(\pi_1) + 2\cdot\oneeffect(\pi_2) + \oneeffect(\pi_3) \geq 0$, we know that each of the three inequalities are actually equalities: $x_0 + \oneeffect(\pi_1) = 0$, $x_0 + \oneeffect(\pi_1) + \oneeffect(\pi_2) = 0$, and $x_0 + \oneeffect(\pi_1) + \oneeffect(\pi_2) + \oneeffect(\pi_3) = 0$.
    This means that the original counter must have had zero value exactly at the moments at which zero-testing transitions were taken. 

    To reiterate: a controlling counter can be a $\ZZ$-counter. 
    Given that the original counter is an $\NN$-counter, the three inequalities $x_0 + \oneeffect(\pi_1) \geq 0$, $x_0 + \oneeffect(\pi_1) + \oneeffect(\pi_2) \geq 0$, and $x_0 + \oneeffect(\pi_1) + \oneeffect(\pi_2) + \oneeffect(\pi_3) \geq 0$ hold.
    Thus, to complete the simulation of zero tests, we only require $3x_0 + 3\cdot\oneeffect(\pi_1) + 2\cdot\oneeffect(\pi_2) + \oneeffect(\pi_3) = 0$.
\end{proof}

\section{Missing Proofs of~\cref{sec:dimension-one}}
\label{app:dimension-one}

\oneVassExponential*

\begin{proof}
    Suppose $\Vv = (Q, T)$ is a \zvass{1}{k}. 
    We can assume, without loss of generality, that for every transition $(p, \vec{u}, q) \in T$, $\zproj{\vec{u}} \neq \vec{0}$ (i.e.\ all transitions change the value of at least one $\ZZ$-counter).
    This can be achieved by replacing a transition $(p, \vec{u}, q)$
    such that $\zproj{\vec{u}} = \vec{0}$ with two transitions: $(p, (\nproj{\vec{u}}, 1, 0, \ldots, 0), p')$ and $(p', (0, -1, 0, \ldots, 0), q)$.    
    
    \proofsubparagraph*{Constructing the OCA.}
    We will construct an OCA $\Aa = (\Sigma, Q', \delta,$ $q_0, F)$ where $\Sigma = \set{a_1, b_1, \ldots, b_k, a_k}$ is the alphabet, $Q' \supseteq Q$ is the set of states, $\delta \sset Q' \times (\Sigma \cup \set{\varepsilon}) \times \set{-1, 0, 1, z} \times Q'$ is the set of transitions, $s \in Q'$ is the initial state, and $t \in Q'$ is the sole final state.

    When constructing the states and transitions of $\Aa$, we will roughly copy the structure of the states and transitions of $\Vv$.
    However, since $\Vv$ may increase (or decrease) its $\NN$-counter by more than 1 on any given transition and since $\Vv$ may update several $\ZZ$-counters simultaneously, we must replace transitions in $\Vv$ with several transitions in $\Aa$ that are connected in series.
    Precisely, we will replace a transitions $(p, \vec{u}, q) \in T$ with $\onenorm{\vec{u}}$ many transitions in $\Aa$.
    There will be $\abs{\nproj{\vec{u}}}$ many $\varepsilon$-transitions that increment $\Aa$'s counter if $\nproj{\vec{u}} \geq 1$, otherwise these transitions will decrement $\Aa$'s counter if $\nproj{\vec{u}} \leq -1$.
    For every $i$, there will also be $\abs{(\zproj{\vec{u}})[i]}$ many $a_i$-transitions which do not change the counter value if $(\zproj{\vec{u}})[i] \geq 1$, otherwise there will be $\abs{(\zproj{\vec{u}})[i]}$ many $b_i$-transitions which do not change the counter value if $(\zproj{\vec{u}})[i] \leq -1$.
    Overall, let $n$ denote the number of states of $\Aa$, we know that 
    \begin{equation*}
        n = \abs{Q} + \sum_{(p, \vec{u}, q) \in T} \onenorm{\vec{u}}.
    \end{equation*}

    \proofsubparagraph*{Comparing words accepted by $\Aa$ to runs in $\Vv$.}
    Let $Z$ be set of all Parikh vectors of all words which contain the same number of `$a_i$'s as `$b_i$'s, for every $i$.
    Precisely, 
    \begin{equation*}
        Z \coloneqq \set{ \vec{v} \in \NN^\Sigma : \vec{v}[a_i] = \vec{v}[b_i]}.
    \end{equation*}
    We know that $Z$ has a small semilinear description (in fact $Z$ is a rather simple linear set).
    For every $1 \leq i \leq k$, let $\vec{p}_i$ be the vector such that $\vec{p}_i[a_i] = \vec{p}[b_i] = 1$ and, for every $j \neq i$, $\vec{p}_i[a_j] = \vec{p}_i[b_j] = 0$.
    Then 
    \begin{equation}\label{eq:Z-semilinear}
        Z = \vec{0} + \set{\vec{p}_1, \ldots, \vec{p}_k}^*.
    \end{equation}
    
    \begin{claim}\label{clm:oca-correctness}
        If there exists a word $w \in \Sigma^*$ such that $w \in L(\Aa)$ and $\psi(w) \in Z$, then there exists a run $s \xrightarrow{\tau}_\Vv t$ such that $\abs{\tau} \leq \onenorm{\psi(w)}$.
    \end{claim}
    \begin{proof}
        Recall the structure of $\Aa$: it contains the states of $\Vv$ as $Q' \supseteq Q$ and it has several transitions connected in series corresponding to the transitions of $\Vv$.
        This means that any run in $\Aa$ from $\config{s}{0}$ to $\config{t}{0}$ must, ultimately, have the same underlying structure as a run in $\Vv$ from $\Config{s}{0}$ to $\Config{t}{0}$.
        If $w \in L(\Aa)$, then we know that the $\NN$-counters reaches 0 in $\Vv$ and the condition $\psi(w) \in Z$ ensures that the $\ZZ$-counters reach 0 as well.
        
        Recall also that we assumed that every transition of $\Vv$ changed the value of at least one $\ZZ$-counter.
        This means that between every visit of a state $q \in Q$, at least one letter of $w$ must have been read. 
        Thus, the run in $\Aa$ can only visit states in $Q$ at most $\onenorm{\psi(w)}$ many times, and hence there is a run from $s$ to $t$ in $\Vv$ of length at most $\onenorm{\psi(w)}$.
    \end{proof}
    
    \proofsubparagraph*{Small descriptions of OCA Parikh images.}
    We will now outline~{\cite[Lemma 15]{AtigCHKSZ16}}.
    The authors prove that the Parikh image of the language of a given OCA $\Aa$ has a polynomial semilinear representation.
    Precisely, 
    \begin{equation}
        \label{eq:semilinear-parikh-image}
        \psi(L(\Aa)) = \bigcup_{i \in I} \vec{b}_i + P_i^*.   
    \end{equation}  
    and there exists a constant $C'$ such that 
    \begin{enumerate}[(i)]
        \item $\abs{I} \leq C' n^{20}$; 
        \item for every $i \in I$, $\infnorm{\vec{b}_i} \leq C' n^{20}$;
        \item for every $i \in I$, $\abs{P_i} \leq C' n^{20}$; and
        \item for every $\vec{p} \in \cup_{i \in I} P_i$, $\infnorm{\vec{p}} \leq C' n^{20}$.
    \end{enumerate}
    
    \proofsubparagraph*{Semilinear sets have small vectors in common.}
    We will now outline a corollary of~{\cite[Theorem 6]{ChistikovH16}} that will suffice for our purposes.
    Let $S_1 = \bigcup_{i \in I_1} \vec{b}_{1,i} + P_{1,i}^*$ and $S_2 = \bigcup_{i \in I_2} \vec{b}_{2,i} + P_{2,i}^*$ be two semilinear sets in $\ZZ^d$.
    Then $S_1 \cap S_2$ has a semilinear representation $S_1 \cap S_2 = \bigcup_{i \in I} \vec{b}_i + P_i^*$ such that, for every $i \in I$, 
    \begin{equation}
        \infnorm{\vec{b}_i} \leq ((p_1 + p_2)\cdot \max\set{M_1, M_2})^{\Oh(k)}
    \end{equation} 
    where $p_1 = \sum_{i\in I_1} \abs{P_{1,i}}$, $p_2 = \sum_{i \in I_2} \abs{P_{2,i}}$, $M_1$ is the greatest absolute value of any number occurring in $\vec{b}_{1,i}$ and $P_{1,i}$, and $M_2$ is the greatest absolute value of any value occurring in and $\vec{b}_{2,i}$ and $P_{2,i}$.
    The corollary that we can therefore draw is that if $S_1 \cap S_2$ is non-empty, then there exists a vector $\vec{v} \in S_1 \cap S_2$ such that $\infnorm{\vec{v}} \leq ((p_1 + p_2)\cdot \max\set{M_1, M_2})^{\Oh(k)}$. 
    Indeed, we need not consider any points in $S_1 \cap S_2$ other than the base vectors in its semilinear representation.
    
    For our purposes, we wish to study the intersection of $\psi(L(\Aa))$ and $Z$; these are semilinear sets of  $2k$-dimensional vectors of natural numbers.
    The number of periods in the semilinear representation of $\psi(L(\Aa))$ is at most $C'^2n^{40}$ (see conditions (i) and (iii) of~\cref{eq:semilinear-parikh-image}).
    The number of periods in the semilinear representation of $Z$ is $k$ (see~\cref{eq:Z-semilinear}).
    The greatest absolute value of any number occurring in the semilinear representation of $\psi(L(\Aa))$ is $C'n^{20}$ (see conditions (ii) and (iv) of~\cref{eq:semilinear-parikh-image}).
    Lastly, the greatest absolute value of any number occurring in the semilinear representation of $Z$ is just 1 (see~\cref{eq:Z-semilinear}).
    Thus, we conclude that if $\psi(L(\Aa)) \cap Z$ is non-empty, then there exists a vector $\vec{v} \in \psi(L(\Aa)) \cap Z$ such that 
    \begin{equation*}
        \infnorm{\vec{v}} \leq ((C'^2n^{40} + k)\cdot C'n^{20} )^{\Oh(k)}
    \end{equation*}
    
    To conclude the proof of~\cref{lem:1-vass-exponential}, we combine the above fact with~\cref{clm:oca-correctness}.
    There exists path $\tau$ such that $\Config{s}{0} \xrightarrow{\tau}_\Vv \Config{t}{0}$ and 
    \begin{equation*}
        \abs{\tau} 
        \leq ((C'^2n^{40} + k)\cdot C'n^{20} )^{C''k} 
        \leq (C'^3n^{60} \cdot k)^{C''k} 
        \leq (n k)^{Ck} 
    \end{equation*}
    for some constants $C''$ and $C$.
\end{proof}

\unaryNL*

\begin{proof}
    When $k$ is fixed in~\cref{lem:1-vass-exponential}, observe that reachability is witnessed by runs that are polynomial in $n$ (the size of the \hbox{\zvass{1}{k}} encoded in unary) and hence can be guessed on the fly by a non-deterministic log-space Turing machine.
\end{proof}

\linearFormRuns*

\begin{proof}
    Suppose $s(\vec{x}) \xrightarrow{\tau}_{\Vv} t(\vec{y})$ is a run of minimal length from $s(\vec{x})$ to $t(\vec{y})$.
    By~\cref{lem:1-vass-exponential}, we know that $\abs{\tau} \leq (nk)^{Ck}$, for some constant $C \in \NN$.
    We remark that~\cref{lem:1-vass-exponential} provides exponential runs between configurations $\Config{s}{0}$ and $\Config{t}{0}$.
    We can modify $\Vv$ to include two additional states and transitions that initially add the starting counter values and subtract the target counter values, respectively.
    This is why in the statement of ~\cref{lem:linear-form-runs}, we have $n = \abs{Q} + (\sum_{(p, \vec{u}, q) \in T} \onenorm{\vec{u}}) + \onenorm{\vec{x}} + \onenorm{\vec{y}}$.
    
    \proofsubparagraph*{Marking transitions.}
    Recall that $\tau$ is a sequence of transitions.
    For every state $q \in Q$, we will ``mark'' the transitions in $\tau$ which lead to the first and last visits of $q$.
    This means that if a particular state $q \in Q$ is visited multiple times, then there will be two marked transitions in $\tau$ for $q$.
    If, however, $q$ is only visited once, then only one such transition is marked in $\tau$ and if $q$ is not visited, then there are no marked transitions in $\tau$ for $q$.
    Let $e_q^+$ and $e_q^-$ denote the transitions that lead to the first and last visits of $q$ in $\tau$.
    Overall, there are at most $2\abs{Q}$ many marked transitions in $\tau$.
    
    \proofsubparagraph*{Removing cycles.}
    We shall define multisets $C_q^+$ and $C_q^-$, for every state $q \in Q$, which will contain cycles.
    We outline a procedure for carefully extracting cycles from $\tau$ and placing them in $C_q^+$ and $C_q^-$.
    We will identify a sequence $(\tau_0, \gamma_1, \tau_1, \gamma_2, \tau_2, \ldots)$ where $\tau_j$ is a path and $\gamma_j$ is a simple cycle.
    Initially, $\tau_0 = \tau$.
    The sequence of paths and cycles shall possess the following properties:
    \begin{enumerate}[(i)]
        \item $\tau_j$ contains all marked transitions $e_q^+$ and $e_q^-$ for every $q \in Q$;
        \item $\gamma_j$ is a simple cycle in $\tau_{j-1}$; and 
        \item $\tau_j$ is $\tau_{j-1}$ with $\gamma_j$ removed.
    \end{enumerate}
    The process of removing cycles $\gamma_j$ is defined to terminate once $\abs{\tau_j} < \abs{Q}\cdot(2\abs{Q}+1)$.
    After this process terminates, we will then place the cycles $(\gamma_1, \ldots, \gamma_j)$ into $C_q^+$ and $C_q^-$.
    
    Suppose that after $j$ removals of cycles, we have $\tau_j$ and $\gamma_j$ that satisfy (i) -- (iii).
    Let's assume that $\abs{\tau_j} \geq \abs{Q}\cdot(2\abs{Q}+1)$, otherwise the process would have terminated.
    Consider the first $\abs{Q}\cdot(2\abs{Q}+1)$ transitions in $\tau_j$; they can be split into $2\abs{Q}+1$ many contiguous blocks each consisting of $\abs{Q}$ many transitions.
    By the pigeonhole principle, it is true that every block contains a simple cycle; and there exists a block that does not contain a marked transition.
    Hence, we identify $\gamma_{j+1}$ to be a simple cycle from the block without marked transitions.
    Accordingly, we define $\tau_{j+1}$ as $\tau_j$ with $\gamma_{j+1}$ removed. 
    Clearly, $\tau_{j+1}$ and $\gamma_{j+1}$ satisfy properties (i) -- (iii).
    
    \proofsubparagraph*{Grouping cycles.}
    Suppose that $(\tau_0, \gamma_1, \tau_1, \ldots, \gamma_j, \tau_j)$ is the sequence obtained after the above procedure terminates (i.e.\ $\abs{\tau_j} < \abs{Q}\cdot(2\abs{Q}+1)$).
    Now, we will explain how to group $\gamma_1, \ldots, \gamma_j$ into $C_q^+$ and $C_q^-$.
    Let $\gamma \in (\gamma_1, \ldots, \gamma_j)$ and suppose that $q \in Q$ is a state in which the cycle attains its least value on the $\NN$-counter.
    Precisely, suppose $\gamma = ((q_1, \vec{u}_1, q_2), (q_2, \vec{u}_2, q_3), \ldots, (q_m, \vec{u}_m, q_{m+1}))$, where $q_{m+1} = q_1$ and let $i \in \set{0,1,\ldots,m}$ be the index such that $\pi_\NN(\vec{u}_1) + \ldots + \pi_\NN(\vec{u}_i)$ attains its least value.
    Then $q = q_{i+1}$ is the state in which $\gamma$ attains its least value on the $\NN$-counter.
    We shall place $\gamma$ in $C_q^+$ if $\gamma$ has nonnegative effect on the $\NN$-counter, otherwise $\gamma$ is placed in  $C_q^-$ (if it has a negative effect on the $\NN$-counter).
    
    \proofsubparagraph*{Cycle replacement.}
    The following cycle replacement argument will be repeated multiple times, once for each multiset of cycles $C_q^+, C_q^-$.
    Recall that $M = \max\set{\infnorm{\vec{u}} : (p, \vec{u}, q) \in T}$.
    We know that the effect of any simple cycle belongs to $[-M\abs{Q}, M\abs{Q}]^{k+1}$, hence $\max\set{\infnorm{\eff{\gamma}} : \gamma \in C_q^+} \leq M\abs{Q}$.
    There may be multiple copies of the same cycle in $C_q^+$ and there may even be multiple cycles with the same effect in $C_q^+$.
    Let $X$ be the set of effects of cycles in $C_q^+$, precisely $X = \set{\eff{\gamma} : \gamma \in C_q^+}$.
    For an effect vector $\vec{x} \in X$, we shall nominate (arbitrarily) a representative cycle, called $\gamma_\vec{x}$, such that $\gamma_{\vec{x}} \in C_q^+$ and $\eff{\gamma_\vec{x}} = \vec{x}$.
    Let $\vec{s} = \sum_{\gamma \in C_q^+} \eff{\gamma}$; clearly $\vec{s} \in X^*$.
    
    We shall now use Eisenbrand and Shmonin's Carath\'{e}odory style bound for integer cones (\cref{cor:caratheodory}) to reduce the number of distinct cycles that we need to use to gain the overall effect $\vec{s}$.
    There exists a subset $Y \sset X$ such that $\vec{s} \in Y^*$ and $\abs{Y} \leq 2(k+1)\cdot\log(4(k+1)M\abs{Q})$.
    Precisely, suppose $\vec{s} = n_1\vec{y}_1 + \ldots + n_t\vec{y}_t$ for some $\vec{y}_1, \ldots, \vec{y}_t \in Y$ and $n_1, \ldots, n_t \in \NN$.
    
    Our goal now is to obtain exponential bounds on $n_1, \ldots, n_t$. 
    For this, we first need to obtain bounds on the norm of $\vec{s}$.
    Recall that the simple cycles in $C_q^+$ were all obtained by removing cycles from the original run $\tau$. 
    The collective length of cycles in $C_q^+$ therefore cannot exceed the length of $\tau$; $\sum_{\gamma \in C_q^+} \abs{\gamma} \leq \abs{\tau}$.
    Thus, 
    \begin{equation*}
        \infnorm{\vec{s}} 
        \leq \sum_{\gamma \in C_q^+} \infnorm{\eff{\gamma}} 
        \leq \sum_{\gamma \in C_q^+} M \cdot \abs{\gamma}
        \leq M\cdot \abs{\tau} 
        \leq M \cdot (nk)^{Ck} 
        \leq (nk)^{Ck+1}.
    \end{equation*}
    
    Now that we have a bound on $\infnorm{\vec{s}}$, we can obtain bounds on $n_1, \ldots, n_t$.
    For this, we use the following corollary which gives us bounds on the size of solutions of systems of integer linear equations $A\vec{x} = \vec{b}$.
    \cref{cor:ilp} differs from Borosh and Treybig's original bound by separating the norms of the $A$ and $\vec{b}$~\cite{BoroshT76}.
    Moreover, we preemptively use Hadamard's inequality to upper bound the maximum absolute value of any full-rank minor of $A$ by $(\sqrt{r} \cdot \infnorm{A})^r$, where $r \leq \min\set{m,n}$ is the rank of $A$.
    This gives us a bound only in terms of $n$, $m$, $\infnorm{A}$, and $\infnorm{\vec{b}}$.
    
    \begin{corollary}[c.f.~{\cite[Theorem 2]{BoroshT76}}] \label{cor:ilp}
        Let $A$ be an $m \times n$ integer matrix and let $\vec{b}$ be an $m$-dimensional integer vector. 
        If $A\vec{x} = \vec{b}$ has an integral solution $\vec{x} \in \NN^n$, then there exists an integral solution $\vec{x} \in \NN^n$ such that $\infnorm{\vec{x}} \leq (n+1) \cdot m^{\frac{m}{2}} \cdot \infnorm{A}^m \cdot \infnorm{\vec{b}}$.
    \end{corollary}
    
    For our purposes, $\vec{s} = n_1\vec{y}_1 + \ldots + n_t\vec{y}_t$ is a system of linear equations in which
    \begin{itemize}
        \item the number of equations is $k+1$,
        \item the number of variables is $t \leq 2(k+1)\cdot\log(4(k+1)M\abs{Q})$, 
        \item the absolute values of the coefficients ($\infnorm{A}$) are $\max\set{\infnorm{\vec{y}_1}, \ldots, \infnorm{\vec{y}_t}} \leq M\abs{Q}$, and
        \item the absolute values of target ($\infnorm{\vec{b}}$) is $\infnorm{\vec{s}} \leq (nk)^{Ck+1}$.
    \end{itemize}
    Accordingly, thanks to~\cref{cor:ilp}, there exists an integral solution $(n_1, \ldots, n_t)$ such that
    \begin{align*}
        \max\set{n_1, \ldots, n_t}
        & \leq 2(k+1)\log(4(k+1)M\abs{Q}) \cdot (k+1)^{\frac{k+1}{2}} \cdot (M\abs{Q})^{k+1} \cdot (nk)^{Ck+1} \\
        & \leq 8(k+1)^2 n^2 \cdot (k+1)^{k+1} \cdot (n^2)^{k+1} \cdot (nk)^{Ck+1} \\
        & \leq 8(2nk)^4 \cdot (2nk)^{k+1} \cdot (nk)^{2k+2} \cdot (nk)^{Ck+1} \\
        & \leq 128(nk)^4 \cdot (nk)^{2k+2} \cdot (nk)^{2k+2} \cdot (nk)^{Ck+1} \\
        & \leq 128(nk)^{(C+4)k+9} \leq (nk)^{C_3k}.
    \end{align*} 
    Here, $C_3 > C$ is some constant.\footnote{In these inequalities, $k \geq 1$ and $n \geq 2$ are assumed. Instances with $n = 1$ are trivial, as there can only be one state in $\Vv$ with zero-effect transitions. Furthermore, one can ensure $k \geq 1$ by introducing an always-zero $\ZZ$-counter.}
    
    Finally, we can replace the multiset of cycles $C_q^+$ with a small collection of cycles that are each taken at most an exponential number of times. 
    We shall replace the multiset of cycles $C_q^+$ with $\gamma_{\vec{y}_1}^{n_1} \, \gamma_{\vec{y}_2}^{n_2} \, \cdots \, \gamma_{\vec{y}_t}^{n_t}$.
    This block of cycles will be taken immediately after the marked transition $e_q^+$ in $\tau_j$.
    We can repeat the previous cycle replacement argument for every multiset of cycles. 
    For every $q \in Q$, $C_q^+$ and $C_q^-$ are replaced by at most $2(k+1)\cdot\log(4(k+1)M\abs{Q})$ many distinct cycles. 
    Suppose, for each $q \in Q$, $\beta_q^+$ is the block of cycles that replaces $C_q^+$ and $\beta_q^-$ is the block of cycles that replaces $C_q^-$.
    In $\tau_j$, we insert $\beta_q^+$ just after $e_q^+$ and $\beta_q^-$ just after $e_q^-$ to obtain a run of the form:
    \begin{equation}\label{eq:replaced-run}
        \rho = \rho_0 \beta_p^+ \rho_1 \cdots \rho_{i-1} \beta_q^+ \rho_i \cdots \rho_{\ell-1} \beta_r^- \rho_\ell,
    \end{equation}
    where $p$, $q$, and $r$ are some arbitrary states for demonstrative purposes.
    To see why the run in~\cref{eq:replaced-run} matches the desired form (as per \cref{eq:lps}), notice that $\beta_q^+, \beta_q^-$ can have empty paths $\varepsilon$ inserted between their cycles.
    For example, one can see $\beta_q^+$ as $\gamma_{\vec{y}_1}^{n_1} \, \varepsilon \,  \gamma_{\vec{y}_2}^{n_2} \, \varepsilon \, \cdots \, \varepsilon \, \gamma_{\vec{y}_t}^{n_t}$.
    
    \proofsubparagraph*{Bounding size parameters of the linear path scheme.}
    First, for (1), we know that $\ell$ is at most the product of the number of blocks and the maximum number of cycles in each block.
    Thus
    \begin{align*}
        \ell 
        & \leq 2\abs{Q} \cdot 2(k+1)\cdot\log(4(k+1)M\abs{Q}) \\
        & = 4\abs{Q}\cdot(k+1)\cdot\log(4(k+1)M\abs{Q}) \\
        & \leq 8k\abs{Q} \cdot \log(8kM\abs{Q}) \\
        & = 8k\abs{Q} \cdot (\log(8) + \log(k) + \log(M) + \log(\abs{Q})) \\
        & \leq 24k^2\abs{Q}^2\cdot\log(M).
    \end{align*}
    
    For (2), we know that the underlying path $\rho_0 \rho_1 \cdots \rho_\ell = \tau_j$ and recall that the cycle removal procedure terminated once $\abs{\tau_j} < \abs{Q} \cdot (2\abs{Q}+1)$. 
    Hence, for $C_2 = 3$, $\abs{\rho_0\rho_1\cdots\rho_\ell} \leq C_2\abs{Q}^2$.
    
    For (3), recall that cycles in the block $\beta_q^+$ are taken $n_i \leq (nk)^{C_3k}$ many times.
    
    Lastly, for (4), we know that the cycles in the blocks $\beta_q^+, \beta_q^-$ were selected as representatives from the multisets of cycles $C_q^+, C_q^-$ that the cycle removal procedure selected.
    Recall that, by definition, the cycle removal procedure only selected simple cycles.
    
    \proofsubparagraph*{Correctness of the linear path scheme.}
    To conclude the proof of~\cref{lem:1-vass-exponential}, it remains to prove that, indeed, $s \xrightarrow{\rho}_\Vv t$.
    
    First recall that, when transforming $\tau$ into $\rho$, we removed cycles to obtain $\tau_j$ and multisets of cycles $C_q^+, C_q^-$, for each $q \in Q$.
    Given properties (ii) and (iii) of the cycle replacement procedure, 
    \begin{equation*}
        \eff{\tau} = \eff{\tau_j} + \sum_{q \in Q} \sum_{\gamma \in C_q^+ \cup C_q^-} \eff{\gamma}.
    \end{equation*}
    Then, we replaced each multiset $C_q^+$ with the block $\beta_q^+$ and $C_q^-$ with $\beta_q^-$.
    Recall, for instance, that 
    \begin{equation*}
        \sum_{\gamma \in C_q^+} \eff{\gamma} = n_1\vec{y}_1 + \ldots + n_t\vec{y}_t = n_1\cdot\eff{\gamma_{\vec{y}_1}} + \ldots + n_t\cdot\eff{\gamma_{\vec{y}_t}} = \eff{\beta_q^+}.
    \end{equation*}
    Together, this means that $\eff{\rho} = \eff{\tau}$.
    So it only remains to analyse the value of the $\NN$-counter in $s \xrightarrow{\rho}_\Vv t$.
    
    Generally, the following is true: let $\gamma$ be a cycle in $\tau$ which (i) has a nonnegative effect on the $\NN$-counter, and (ii) attains its least value on the $\NN$-counter at state $q$.
    Now, let $\tau'$ be a run in which $\gamma$ is moved to an earlier point in the run.
    Precisely, we consider any earlier point to be any point which is before the original occurrence of $\gamma$ and which immediately proceeds a transition that leads to $q$.
    Given that $\Config{s}{0} \xrightarrow{\tau}_\Vv \Config{t}{0}$, we can conclude that $\Config{s}{0} \xrightarrow{\tau'}_\Vv \Config{t}{0}$ because the $\NN$-counter is never lower in the run following $\tau'$ compared to the run following $\tau$.
    Symmetrically, the same is true for cycles that have a negative effect on the $\NN$-counter being taken later in the run.
    We therefore conclude that $s \xrightarrow{\tau'}_\Vv t$ is a valid run.
    
    As an intermediate step before proving that $s \xrightarrow{\rho}_\Vv t$ is a valid run, we will first consider a slightly different run.
    Let $\rho'$ be the run that is the same as $\rho$ but instead of the blocks $\beta_q^+$ and $\beta_q^-$, the multisets of cycles $C_q^+$ and $C_q^-$ are taken (in any order) just after their respective marked transitions.
    Roughly speaking, this can be written as $\rho' = \rho_0 \, C_p^+ \, \rho_1 \, \cdots \, \rho_{i-1} \, C_q^+ \, \rho_i \, \cdots \, \rho_{\ell-1} \, C_r^- \, \rho_\ell.$
    Notice that $\rho'$ only differs from the original run $\tau$ based on when exactly simple cycles are taken: the simple cycles in $\tau$ with positive effect on the $\NN$-counter are taken earlier and simple cycles with negative effect on the $\NN$-counter are taken later in the run.
    By the above, this means that $s \xrightarrow{\rho'}_\Vv t$ is a valid run.
    
    Now, since $\rho'$ and $\rho$ only differ based on the blocks of cycles that are taken, we can conclude that $s \xrightarrow{\rho}_\Vv t$.
    Indeed, the effect on the $\NN$-counter before a block $\beta_q^+$ (or $\beta_q^-$) is the same in $\rho$ and $\rho'$ and the resulting $\NN$-counter after a block is also the same.
    This is because $\eff{\beta_q^+} = \sum_{\gamma \in C_q^+} \eff{\gamma}$.
\end{proof}

\dimensionOneNP*

\begin{proof}
    Let $k$ be a number written in unary, let $\Vv = (Q,T)$ be a \zvass{1}{k}, and let $s, t \in Q \times (\NN\times\ZZ^k)$ be two configurations. 
    Let $M = \max\set{\infnorm{\vec{u}} : (p, \vec{u}, q) \in T}$.
    By~\cref{lem:linear-form-runs}, we know that if $s \xrightarrow{*}_\Vv t$, then there exists a run $s \xrightarrow{\rho}_\Vv t$ where $\rho = \rho_0 \sigma_1^{x_1} \rho_1 \cdots \rho_{\ell-1} \sigma_\ell^{x_\ell} \rho_\ell$ satisfies conditions (1) -- (4) of~\cref{lem:linear-form-runs}.
    
    \proofsubparagraph*{Polynomial size of linear path scheme.}
    In order to specify $\rho$, we only need to write down the underlying paths $\rho_0, \rho_1, \ldots, \rho_\ell$, the intermediate simple cycles $\sigma_1, \ldots, \sigma_\ell$, and the number of times each cycle is taken $x_1, \ldots, x_\ell$.
    Since $\ell \leq C_1k^2\abs{Q}^2\cdot\log(M)$, we only need to specify polynomially many paths, cycles, and natural numbers.
    Each path $\rho_i$ can be written down using polynomially many bits (thanks to property (2)).
    Each cycle $\sigma_i$ can be written down using polynomially many bits because $\abs{\sigma_i} \leq \abs{Q}$ (thanks to property (4)).
    Finally, each natural number $x_i$ can be written down using polynomially many bits using binary encoding because $\log(x_i) \leq C_3k\cdot\log(nk)$ (thanks to property (3)).
    
    \proofsubparagraph*{Checking the run in polynomial time.}
    It remains to argue that the validity of $s \xrightarrow{\rho}_\Vv t$ can be checked in polynomial time. 
    Firstly, we check that indeed $\eff{\rho}$ is equal to the difference in counter values in $t$ and $s$.
    Now, we only need to check that the $\NN$-counter does not drop below zero, rather than one-by-one checking every transition in $s \xrightarrow{\rho}_\Vv t$, it suffices to only check a small number of configurations.
    One must check every configuration along the underlying paths $\rho_i$, however once a cycle is reached, rather than checking every configuration after $\sigma_i^1, \sigma_i^2, \sigma_i^3, \ldots, \sigma_i^{x_i}$, it suffices to only check the first and last iterations of $\sigma_i$. 
    It could be the case, for example, that $\sigma_i$ has nonnegative effect on the $\NN$-counter but the cycle has a prefix which has a negative effect on the $\NN$-counter.
    For this, we must manually check the configurations observed during the first iteration of $\sigma_i$.
    However, once the first iteration has been cleared, we know that the $\NN$-counter remains nonnegative for all subsequent iterations, as $\sigma_i$ had nonnegative effect. 
    Similarly, if $\sigma_i$ has negative effect on the $\NN$-counter than we must manually check the configurations observed during the last iteration. 
\end{proof}

\section{Missing Proofs of~\cref{sec:lower-bounds}}
\label{app:lower-bounds}

\CA*

\begin{proof}
    We reduce from the following classical problem known to be $\spacecl(g(n))$-complete for every space constructible function $g$. 

    \begin{description}
        \item[\hspace{1.15em}Input:] Turing Machine (TM) $\mathcal{M}$ with $2$ tape letters $\set{0, 1}$, a number $n$ encoded in unary, and two distinguished states $q_I, q_F$.
        \item[\hspace{1.15em}Question:] Does $\mathcal{M}$ have a run from $q_I$ with every tape cell containing $0$ to $q_F$ with every tape cell containing $0$ that only accesses the first $g(n)$ first cells?
    \end{description}

    Indeed $\spacecl(g(n))$-membership follows from the fact that one can simulate $\mathcal{M}$ by a Turing Machine using at most $g(n)$ space. 
    The above problem is also $\spacecl(g(n))$-hard as for any problem $L \in \spacecl(g(n))$ there exists a Turing machine $\mathcal{M}_L$ that decides $L$ using at most $g(n)$ space. For input $x$ of length $n$ we can construct Turing Machine $\mathcal{M}$, which simulates $\mathcal{M}_L$ on $x$ in space $g(n)$ and has a run from $q_I$ with whole tape covered with letter $0$ to $q_F$ with whole tape covered with letter $0$ if and only if   $x \in L$. 

    Firstly we construct a TM using only $g(n)+1$ first cells in runs starting from $q_I$ with whole tape covered with letter $0$. This construction is due to following claim:

    \begin{claim}\label{clm:bounded_TM}
        For every TM $\mathcal{M}$ with $2$ tape letters $0, 1$, a number $n$ given in unary and two distinguished states $q_I$ and $q_F$ one can construct in polynomial time TM $\mathcal{M}'$ such that:
        \begin{itemize}
            \item $\mathcal{M}'$ has $3$ tape letters $0, 1, 2$
            \item $\mathcal{M}'$ has two distinguished states $q_I'$ and $q_F'$
            \item In every run from $q_I'$ with whole tape covered with letter $0$ TM $\mathcal{M}'$ uses at most $g(n)+1$ first cells
            \item $\mathcal{M}'$ has a run from $q_I'$ with whole tape covered with letter $0$ to $q_F'$ with whole tape covered with letter $0$ if and only if $\mathcal{M}$ has run from $q_I$ with whole tape covered with letter $0$ to $q_F$ with whole tape covered with letter $0$, which uses at most $g(n)$ first cells of the tape of $\mathcal{M}$
        \end{itemize}
    \end{claim}

    \begin{proof}
        We construct $\mathcal{M}'$ as a modification of $\mathcal{M}$. We create new initial state $q_I'$ and final state $q_F'$ in $\mathcal{M}'$. We set, that $q_F'$ has no outgoing transitions.

        Since $g$ is space-constructible,  we can write in the cell numbered $g(n)+1$ the new letter $2$ (using just $g(n)+1$ cells) and then move the head of $M$ once more back to the first cell. Further we simulate TM $M$ from state $q_I$ on the part of tape, before the letter $2$. If at some point $\mathcal{M}'$ puts its head above $n+1$-th cell once more, we know, that in the simulated run $\mathcal{M}$ used more than $g(n)$ cells and we go to new sink state $q_s$, which has no outgoing transitions. If at some point $\mathcal{M}'$ reaches the state $q_F$ and we want to end the simulation of a run of $\mathcal{M}$ we move head right until we encounter symbol $2$ on the tape. We replace it by $0$ and go to $q_F'$. Observe, that $\mathcal{M}'$ uses only $g(n)+1$ first cells of the tape and it is clear that $M'$ has a run from $q_I'$ with whole tape covered with letter $0$ to $q_F'$ with whole tape covered with letter $0$ if and only if $M$ has run from $q_I$ with whole tape covered with letter $0$ to $q_F$ with whole tape covered with letter $0$, which uses at most $g(n)$ first cells of the tape of $M$.

        This concludes the proof of Claim~\ref{clm:bounded_TM}.
    \end{proof}

    Having this construction, we use the simulation of TM by a $3$-counter automaton, which was presented in \cite{fischer1968counter}. Let the visited part (namely the part where head was present at some point of the run) be of the form
    \[
    a_n a_{n-1}\ldots a_0 h c_0 c_1 \ldots c_k,
    \]
    where $h$ is the letter currently under the head, $a_1 a_2 \ldots a_n$ is the part to the left of the head and $c_1 c_2 \ldots c_k$ is part to the right of the head then the encoding of this tape content will be as follows:
    \begin{itemize}
        \item The first counter stores $a_n \cdot 3^n + a_{n-1} \cdot 3^{n-1} + \ldots + a_0$.
        \item The second counter stores $c_0 + c_1 \cdot 3 + \ldots + c_k \cdot 3^k$.
        \item The third counter is auxiliary and it is zero.
        \item Letter $h$ under the head is stored in the current state.
        \item State of the TM is stored in the current state.
    \end{itemize}
    Now we show, how simulation of replacing $h$ by $d$ and shifting the head right is done. Simulations of other transitions are done similarly. Concretely, the counter automaton multiplies the first counter by $3$ and adds $h$ to it. The multiplication is done by flushing everything while multiplying by $3$ to the third counter, zero-testing first counter and flushing everything back to the first counter, which is verified to be executed fully by zero-testing the third counter. Next, we divide the second counter by $3$ using zero-tests and the third counter. Of course the second counter does not need to be divisible by $3$. Therefore we guess its residue modulo $3$ and decrease it by this value first. Moreover, this residue is put into the state as it will be the letter under the head after the transition. In this way we transformed encoding before the transition to the encoding after executing the transition of TM. The initial and final configurations of TM $M'$ are encoded as $q_I(0,0,0)$ and $q_F(0,0,0)$ for some distinguished states of the counter automaton. Recall, that $M'$ uses only first $g(n)+1$ cells. Therefore we can bound each counter in every run from $q_I(0,0,0)$ by:
    $$\Sigma_{k \in [0, f(n)]} 3 \cdot 3^k < 3^{f(n)+2}$$
    Observe, that  $M$ has run from $q_I$ with whole tape covered with letter $0$ to $q_F$ with whole tape covered with letter $0$, which uses at most $g(n)$ first cells of the tape of $M$ if, and only if, the constructed $3$-counter automaton has run from $q_I(0,0,0)$ to $q_F(0,0,0)$. Because each counter in every run is bounded by $3^{g(n)+2}$ we get that the shortest such run can be of length at most $|Q| \cdot 3^{3g(n)+6} \leq 3^18(g(n)) \leq 2^{cg(n)}$ for some constant $c$. Thus, taking $g(n) = \frac{f(n)}{c}$ and using that $\spacecl(f(n)) = \spacecl(g(n))$ we conclude with the proof of \Cref{lem:reachability-CA}. 
\end{proof}

\twoexptriple*

\begin{proof}
    We first show the following auxiliary lemma.

    \begin{lemma}\label{lem:multiplication}
        For every $k \in \NN$ there is a \zvass{2}{} of size linear in $k$ with distinguished states $q_I$ and $q_F$ such that for $B$, $C$, $D$, $C_1, \ldots, C_k$, $D_1, \ldots, D_k \in \NN$ we have $q_I(B, C, 2BC, C_1, \ldots, C_k) \reaches q_F(D, 0, 0, D_1, \ldots, D_k)$ if and only if $C = C_1 + \ldots + C_k$, $D = B$ and $D_i = B \cdot C_i$ for each $i \in [1,k]$.
    \end{lemma}

    \begin{proof}
        Let us denote the counters with starting values $B, C, 2BC, C_1, \ldots, C_k$ as $x, y, z, u_1, \ldots, u_k$, respectively. 
        Let counters $\bar{x}, v_1, \ldots, v_k$ be auxiliary and start with zero values. Assume that $x$ and $\bar{x}$ are $\NN$-counters and all the others are $\ZZ$-counters. 
        We claim that the following counter program presented in Algoritm~\ref{alg:mult} describes the postulated \zvass{2}{}:

        \begin{algorithm}
            \caption{Counter program $\Mult(x, y, z, u_1, \ldots, u_k)$.}
            \begin{algorithmic}[1]
            \For {i=1, \ldots, k}
                \Loop
                    \Loop
                        \State $x \to \bar{x}$
                        \State \sub{z}{1}
                        \State \add{v_i}{1}
                    \EndLoop

                    \Loop
                        \State $\bar{x} \to x$
                        \State \sub{z}{1}
                    \EndLoop
                    
                    \State \sub{y}{1}
                    
                    \State \sub{u_i}{1}
                \EndLoop
            \EndFor
            \For {i= 1, \ldots, k} \ \ZT($u_i$)
            \EndFor

            \For {i= 1, \ldots, k}
                \Loop \ $v_i \to u_i$
                \EndLoop
            \EndFor

            \For {i= 1, \ldots, k} \ \ZT($v_i$)
            \EndFor
            \end{algorithmic}
            \label{alg:mult}
        \end{algorithm}

        The main idea behind the proof of its correctness is similar to the argument the multiplication triples technique~\cite{CzerwinskiLLLM19}. However, the details are a bit different. Let us analyse the above code. First notice that the sum $x+
        \bar{x}$ is constant, as only transitions modifying these counters (in lines 4 and 8) preserve the sum. Thus $x + \bar{x} = B$ all the time. This means that the loops in lines 3 and 7 can be iterated at most $B$ times, as the first one decreases $x$ and the second one decreases $\bar{x}$. Now inspect the loop in line 2. In its body the counter $y$ is decreased always by 1, while the counter $z$ is decreased in lines 5 and 9, so exactly by the number of iterations of loops in lines 3 and 7. Therefore, while in the body of loop from line 2 the counter $y$ is decreased by 1, the counter $z$ is decreased by at most $2B$. Initial values of $(y, z)$ are $(C, 2BC)$, while their final values are $(0, 0)$. That means that in total $y$ has to be decreased by $C$ and $z$ has to be decreased by $2BC$. Hence in each iteration of the loop in line 2 counter $z$ has to be decreased by exactly $2B$. That means that loops in lines 3 and 7 has to be always iterated exactly $B$ times.
        This, in particular, implies that the final value of $x$ equals $B$, so $D = B$.

        For each $i \in [1,k]$ the counter $u_i$ starts with value $C_i$ and reaches value 0 at line 12. Therefore the loop in line 2 for this $i$ has to be iterated exactly $C_i$ times. That, together with the observation that loops in lines 3 and 7 are always iterated $B$ times implies that the counter $v_i$ is incremented by exactly $B \cdot C_i$ and its value before the line 12 equals $B \cdot C_i$. Notice also that it means that counter $y$ is altogether decremented by $C_1 + \ldots + C_k$ and starting at $C$ it has to reach value 0 at the end of the program. This guarantees the requirement that $C = C_1 + \ldots + C_k$.

        Lines 13-15 copy values from counters $v_i$ to counters $u_i$. That implies that for each $i \in [1,k]$ the final value of $u_i$, denoted by $D_i$ in the lemma statement, equals exactly $B \cdot C_i$. Finally, one can easily observe that the size of the presented counter program is linear in $k$. This finishes the proof of Lemma~\ref{lem:multiplication}
    \end{proof}

    Let us denote by $\Mult(x, y, z, u_1, \ldots, u_k)$ the counter program from Lemma~\ref{lem:multiplication}. With this tool in hand we are ready to write the counter program computing the triples $(A^{2^n}, B, B \cdot A^{2^n})$.
    We will use $2n+4$ explicit counters (plus some implicit ones, for example the ones used inside applications of \Mult, \ZT\ or \Copy). These are the $\NN$-counter $x$, auxiliary $\ZZ$-counter $u$ and $n+1$ pairs of $\ZZ$-counters $y_i$ and $z_i$, for $i \in [0,n]$.
    The other $\NN$-counter $\bar{x}$ is implicitly used inside the $\Mult$ counter program.
    The idea is that at the beginning
    the counters $x, y_0, z_0, y_1, z_1, \ldots, y_n, z_n$ have values
    \[
    A, C_0, 2A C_0, C_1, 2A C_1, \ldots, C_n, 2A C_n
    \]
    for some guessed values $C_0, \ldots, C_n \in \NN$. We claim that after $i$ iterations of the main loop the counters $x, y_i, z_i, \ldots, y_n, z_n$ have values $A^{2^i}, C_i, 2 A^{2^i} C_i, \ldots, C_n, 2 A^{2^i} C_n$.
    We use below a simple subprocedure copying value of one counter (say $x_1$) to another (say $x_2$) and using an auxiliary counter $x_3$. 
    It can be easily implemented by the following counter program, denoted $\Copy(x_1, x_2)$:

    \begin{algorithm}
        \caption{Counter program $\Copy(x_1, x_2)$.}
        \begin{algorithmic}[1]
            \Loop \ $x_1 \to x_2, x_3$
            \EndLoop
            \State \ZT($x_1$)
            \Loop \ $x_3 \to x_1$
            \EndLoop
            \State \ZT($x_3$)
        \end{algorithmic}
        \label{alg:copying}
    \end{algorithm}

    Before further explanation let us consider the code of the counter program
    presented in Algorithm~\ref{alg:main}.
    We assume wlog. that all the counters start with zero values.

    \begin{algorithm}
        \caption{Producing doubly-exponential triples.}
        \begin{algorithmic}[1]
            \State \add{x}{A}

            \For {$i= 0, \ldots, n$}
                \Loop \ \add{y_i}{1}; \add{z_i}{2A}
                \EndLoop
            \EndFor

            \For {$i= 0, \ldots, n-1$}
                \State $\Copy(x, u)$
                
                \State $\Mult(x, y_i, z_i, u, z_{i+1}, z_{i+2}, \ldots, z_n)$
                
                \State \ZT($y_i$); \ZT($z_i$)
                
                \Loop \ \sub{x}{1}
                \EndLoop
                
                \State \ZT($x$)
                
                \Loop \ $u \to x$
                \EndLoop
                
                \State \ZT($u$)
            \EndFor

            \State \ExactMult($z_n, 1/2$)
        \end{algorithmic}
        \label{alg:main}
    \end{algorithm}

    Let's now inspect what happens in the presented code. After line 3 we have $x = A$,
    while for each $i \in [0,n]$ we have $(y_i, z_i) = (C_i, 2A C_i)$ for some guessed values $C_i \in \NN$.
    Now we claim that before the $i$-th iteration of the for-loop in line 4 counters $x$, $u$, $y_i, z_i, \ldots, y_n, z_n$ have values $A^{2^i}, 0, C_i, 2 A^{2^i} C_i, \ldots, C_n, 2 A^{2^i} C_n$. This is true for $i = 0$, let's prove it by induction on $i$.
    We assume that it is true for $i$ and show it for $i+1$. Towards that we just inspect what happens in the $i$-iteration of the for-loop in line 4. After line 5 counter $u$ has value $A^{2^i}$. Lines 6 and 7 should be considered together. Line 7 guarantees that firing of the $\Mult$ counter program in line 6 is correct, as for the correctness we need that values of counters $y_i$ and $z_i$ are zero at the end of $\Mult(x, y_i, z_i, u, z_{i+1}, z_{i+2}, \ldots, z_n)$. As guaranteed by Lemma~\ref{lem:multiplication} if zero-tests in line 7 are correct and
    \begin{equation}\label{eq:mult-correctness}
        C_i = A^{2^i} + 2 A^{2^i} C_{i+1} + 2 A^{2^i} C_{i+2} + \ldots + 2 A^{2^i} C_n
    \end{equation}
    then inside the counter program $\Mult$ all the counters $u$, $z_{i+1}$, $z_{i+2}, \ldots, z_n$ are multiplied by the current value of $x$, namely by $A^{2^i}$. Observe that for any values of $C_{i+1}, \ldots, C_n$ equality~\eqref{eq:mult-correctness} can be guaranteed by the correct guess of value $C_i$ in line 2. This means that after line 7 values of
    $x, u, y_{i+1}, z_{i+1}, \ldots, y_n, z_n$ are
    \[
    A^{2^i}, A^{2^{i+1}}, C_{i+1}, 2 A^{2^{i+1}} C_{i+1}, \ldots, C_n, 2 A^{2^{i+1}} C_n.
    \]
    Thus to finish the induction step it only remains to change value of $x$ to $A^{2^{i+1}}$
    and value of $u$ to zero. It is easy to observe that lines 8-11 implement exactly these changes.

    After the last iteration of the for-loop in line 4 counters $(x, y_n, z_n)$ have
    values $(A^{2^n}, C_n, 2 A^{2^n} C_n)$. Therefore it is enough to divide the third counter by $2$.
    This is implemented in line 12 by multiplication by $1/2$. Then we get a triple of counters
    with values $(A^{2^n}, C_n, A^{2^n} C_n)$, exactly as required in the lemma statement.
    Finally, it is straightforward to observe that the size of counter program is polynomial and it uses at most polynomially many zero-tests, which is allowed by Claim~\ref{cl:zero-tests}. This finishes the proof of Lemma~\ref{lem:twoexptriple}.
\end{proof}

\threeVASSZtriples*

\begin{proof}
    The proof is analogous to the proof of Lemma 10 in~\cite{CzerwinskiO22},
    which has very similar statement, with a few small modifications.
    One modification is that in our statement $\Vv$ does not depend on $B$ and $C$.
    The other modification is that in our statement only three counters are $\NN$-counters,
    while the other two (and many implicit, auxiliary ones) are $\ZZ$-counters.
    We do not repeat the proof here, but only justify briefly why the modifications can be obtained. Regarding the first modification it can be easily observed that in the proof of Lemma 10 in~\cite{CzerwinskiO22} VASS $\Vv$ is constructed independently of $B$ and $C$. The second modification requires slightly more work.
    Let us name the five counters of $\Vv$ as $x_1$, $x_2$, $x$, $y$ and $z$. The first two, namely $x_1$ and $x_2$ are supposed to simulate the two counters of $\Aa$, while the triple $(x, y, z)$ with initial values $(B, 2C, 2BC)$ is supposed to simulate $C$ zero-tests on $x_1$, $x_2$ as long as they are $B$-bounded.
    One can now observe that during the simulation of zero-tests the second and the third counter of the triple only gets decreases and is never increased. Therefore even if we assume that $y$ and $z$ are $\ZZ$-counters, but we require that initial and final values ($(2C, 2BC)$ and $(0,0)$ respectively) and nonnegative, then we are sure that along the whole run counter values were nonnegative. That means that indeed the construction from 
    the proof of Lemma 10 in~\cite{CzerwinskiO22} can be easily modified to show Lemma~\ref{lem:3-VASS+Z-triples}.
\end{proof}

\begin{corollary} \label{cor:3-vass+kz}
    Reachability in unary \zvass{3}{k} is $\Oh(k)$-\class{ExpSpace}-hard.
\end{corollary}

\begin{proof}[Proof idea.]
    This follows from the proof of~\cref{thm:unary-3-zvass-tower-hard}.
    The proof of this corollary is very similar to~{\cite[Corollary 5]{CzerwinskiLLLM19}} in which the authors prove that, for any $d \geq 1$, reachability in $(d+16)$-VASS is $d$-\class{ExpSpace}-hard.
    The idea is to construct a VASS that faithfully simulates a 3-counter machine with $d$-fold-exponentially bounded counters (a $d$-\class{ExpSpace}-complete problem~\cite{fischer1968counter,Schmitz16}).
    In particular, during proof of~\cref{thm:unary-3-zvass-tower-hard} on~\cpageref{page:exp-amplifier}, we introduce several $\ZZ$-counters for the sake of performing zero tests.
    Importantly, there are only constantly many zero-testing $\ZZ$-counters used in each copy of the \ExpAmplifier\ counter program.
    Thus, in the same way as in the proof of~\cref{thm:unary-3-zvass-tower-hard}, we can construct a $k$-fold-exponential amplifier by using \ExpAmplifier\ $k$ times; this requires $\Oh(k)$ many additional $\ZZ$-counters. 
    We note that unary encoding suffices: the \ExpAmplifier\ counter program only uses counter updates that are bounded by a fixed constant.
    With this, we can then simulate 3-counter machines with $d$-fold-exponentially bounded counters to obtain $\Oh(k)$-\class{ExpSpace}-hardness of reachability in unary \zvass{3}{k}.
\end{proof}

\section{Missing Proofs of~\cref{sec:klm-upper-bound}}
\label{app:klm-upper-bound}

\PropositionBadSequenceLengthBetter*

\begin{proof}
    Let \(\vect{x}_0, \vect{x}_1, \dots\) be an \((f_1, f_0, m)\)-controlled lexicographically decreasing sequence in \(\NN^d\).
    We will create a lexicographically decreasing sequence \(\vect{y}_0, \vect{y}_1, \dots\) in \(\NN^{d-1}\) which is \((f,f,m)\)-controlled for an elementary function \(f\), and then apply Proposition \ref{PropositionBadSequenceLength}.
    
    We first define the function \(f \colon \NN \to \NN\) as \(f(m):=f_1(f_0^{(m)}(m))\), where \(g^{(m)}\) is \(m\)-fold composition of \(g\). 
    In other words, the function \(f\) first applies \(f_0\) as often as given in the input, and then applies \(f_1\). 
    Observe that repeatedly taking some polynomial function a linear number of times leads (at most) to double-exponentiation.
    Hence \(f\) is indeed an elementary function. Furthermore, \(f\) is monotone and fulfills \(f(x) \geq x\) since both \(f_1\) and \(f_0\) have these properties.
    
    The sequence \(\vect{y}_0, \vect{y}_1, \dots\) is obtained from \(\vect{x}_0, \vect{x}_1, \dots\) by deleting every step which only decreases the least important coordinate. Since \(\vect{x}_0, \vect{x}_1, \dots\) is \((f_1, f_0, m)\)-controlled, it is easy to see that \(\vect{y}_0, \vect{y}_1, \dots\) is \((f,f,m)\)-controlled. The length bound of \(\Ff_{d+1}\) on \(\vect{y}_0, \vect{y}_1, \dots\) obtained by Proposition \ref{PropositionBadSequenceLength} transfers to the sequence \(\vect{x}_0, \vect{x}_1, \dots\)
\end{proof}

\LemmaPerfectnessImpliesRun*

\begin{proof}[Proof Sketch of Lemma~\ref{LemmaLabelPerfectnessImpliesRun}] 
    We will show how the conditions are used to create a run \(\rho\)); the intuition for perfectness and how it implies a run is as follows.
    The basis of \(\rho\) is formed by the ILP solution \(\vect{sol}\) guaranteed by condition (1). 
    Since \(\vect{sol}\) fulfills the Kirchhoff equations, \(\vect{sol}\) can be implemented as a run \(\rho_{\ZZ}\), which has the correct effect and is non-negative whenever it leaves an SCC (by the second set of equations). 
    However, inside an SCC it is not yet clear that the  \(\NN\)-counters will remain nonnegative because \(\rho_{\ZZ}\) is only a \(\ZZ\)-\emph{run}. 
    
    This is where we use condition (3) and pumpability. The configurations \(\qiin(\vect{x}_i), \qiout(\vect{y}_{i})\) where \(\rho_{\ZZ}\) enters and respectively leaves the \(i\)-th SCC \(\Vv_i\) can reach arbitrarily large values on \(\NN\) counters by adding the corresponding cycles \(\up_i, \dwn_i\) into \(\rho_{\ZZ}\) often enough. By adding these cycles, the run \(\rho_{\ZZ}\) will keep the $\NN$-counters nonnegative.
    Unfortunately, this spawns a new problem: the effect of \(\rho_{\ZZ}\) has now changed and we might not reach the intended target counter values. 
    Hence we use condition (2): take a homogeneous solution \(\vect{h}\) which uses every transition \(e\) a large number of times often, in particular more often than once more than the number of times \(\up_i\) and \(\dwn_i\) are taken. 
    Then, \(\vect{h}-\psi(\up_i)-\psi(\dwn_i)\), being a positive solution of the Kirchhoff equations, can be implemented using a cycle \(\diff_i\). Since \(\vect{h}=\psi(\up_i)+\psi(\dwn_i)+\psi(\diff_i)\) is a homogeneous solution, we obtain that \(\effect(\up_i)+\effect(\dwn_i)+\effect(\diff_i)=\vec{0}\) on all ``relevant'' counters. 
    Thus, by adding \(\diff_i\) as often as we added \(\up_i\) and \(\dwn_i\), we obtain the correct effect again. 
    
    In total, in the \(i\)-th SCC $\Vv_i$, we use the run
    \begin{equation*}
        \rho_{i,m}:=\up_i^m \, \rho_{\ZZ,i} \, \diff_i^m \,  \dwn_i^m,
    \end{equation*}
    for a sufficiently large number \(m \in \NN\), where \(\rho_{\ZZ,i}\) is the part of \(\rho_{\ZZ}\) over the SCC \(\Vv_i\). 
    The full run \(\rho\) is then defined as
    \begin{equation*}
        \rho \coloneqq \rho_{0,m} \, e_1 \, \rho_{1,m} \, \dots \, e_s \, \rho_{s,m},
    \end{equation*}
    where as before \(e_i\) is the transition connecting the \((i-1)\)-st SCC to the \(i\)-th SCC. 
    It can be checked (see, for example~\cite[Lemma 4.19]{LerouxS19}) that the obtained run \(\rho\) remains non-negative on \(\NN\) counters for large enough \(m \in \NN\).
    We conclude the proof as we have already argued it reaches the correct target counter values.
\end{proof}

\LemmaReplaceByPolySize*

\begin{proof}
    We claim that the reachability relation of \(\Vv\) can be captured by linear path schemes which have polynomially many simple cycles, with paths of at most exponential length in-between. We first show that the claim implies the lemma.
    
    \paragraph*{Proof of Lemma~\ref{LemmaReplaceByPolySize} assuming the above claim}
    Observe that the exponentially long paths between the cycles have an exponential (i.e.\ polynomial bitsize) effect. We can replace them by a negative transition \(t_-\) subtracting the least \(\NN\)-configuration where the path can be executed, and afterwards a positive transition \(t_+\) which adds back the subtracted value as well as adding the overall effect of the path. 
    Similarly, simple cycles have a polynomial bitsize effect. The result of the replacement is hence a polynomial size linear path scheme which we can use as \(\Vv_j\) as required.
    
    \paragraph*{Proof of claim} 
    Overall, this proof uses the same ideas as Lemma \ref{lem:linear-form-runs}. Consider any run \(\rho\) of the \(0\)-dimensional \zvass{}{} \(\Vv\). 
    Whenever we encounter a state \(q\) a second time, record the simple cycle \(\rho_1\) which was performed, and remove it from the run. Continue doing so until the run is a simple path, which hence has at most exponential length, call this exponential length run \(\rho_{simple}\).
    
    Using~\cref{cor:caratheodory}, with \(X = \{\effect(\rho_i)\mid \rho_i \text{ is a recorded cycle}\}\), we can select polynomially many \(\rho_1, \dots, \rho_\ell\) of these cycles to generate the same effect. Then \(\rho\) can be replaced by a run in the linear path scheme which has \(\rho_{simple}\) as underlying path, and for every cycle \(\rho_j\), we fix \emph{one} location in \(\rho_{simple}\) where the cycle \(\rho_j\) was used, and add a \(\rho_j\) cycle on this state of the LPS. Importantly, since every cycle \(\rho_j\) has effect \(0\) on \(\NN\)-counters, we are indeed allowed to move these cycles to different locations without fear of the $\NN$-counters dropping below zero, so long as they are enabled at this location. 
    Indeed, said cycles are enabled at their respective location holds by definition of \(\rho_j\).
\end{proof}

\end{document}